\documentclass[sn-mathphys-num,oneside]{sn-jnl}
\makeatletter
\@twosidefalse
\makeatother
\usepackage{graphicx}%
\usepackage{multirow}%
\usepackage{amsmath,amssymb,amsfonts}%
\usepackage{amsthm}%
\usepackage{mathrsfs}%
\usepackage[title]{appendix}%
\usepackage{xcolor}%
\usepackage{textcomp}%
\usepackage{manyfoot}%
\usepackage{booktabs}%
\usepackage{algorithm}%
\usepackage{algorithmicx}%
\usepackage{algpseudocode}%
\usepackage{listings}%
\usepackage{tensor}
\usepackage{bbm}
\usepackage{pst-node}
\usepackage{tikz}
\usetikzlibrary{arrows.meta}
\usetikzlibrary{angles, quotes}

\newtheoremstyle{mythrmstyleone}
  {5pt}   
  {5pt}    
  {\itshape} 
  {0pt}    
  {\bfseries} 
  {}      
  {0pt}    
  {}       

\theoremstyle{mythrmstyleone}%
\newtheorem{theorem}{Theorem}
\newtheorem{lemma}{Lemma}
\newtheorem{proposition}[theorem]{Proposition}%
\newtheorem{definition}{Definition}%

\usepackage{xcolor}

\usepackage[normalem]{ulem}

\usepackage{hyperref}

\begin{document}

\title{The Hamiltonian mechanics of exotic particles}

\author[1,2,a]{\fnm{Andrea} \sur{Amoretti}}

\author[1,2,b]{\fnm{Daniel K.} \sur{Brattan}}

\author[3,c]{\fnm{Luca} \sur{Martinoia}}

\affil[1]{\orgdiv{Dipartimento di Fisica}, \orgname{Universit\`{a} di Genova}, \orgaddress{\street{via Dodecaneso 33}, \city{Genova}, \postcode{I-16146}, \country{Italy}}}

\affil[2]{\orgdiv{Sezione di Genova}, \orgname{I.N.F.N.}, \orgaddress{\street{via Dodecaneso 33}, \city{Genova}, \postcode{I-16146}, \country{Italy}}}

\affil[3]{\orgdiv{Laboratory of Interdisciplinary Physics, Department of Physics and Astronomy ``G. Galilei''}, \orgname{Universit\`{a} di Padova}, \orgaddress{\street{via Marzolo, 8}, \city{Padova}, \postcode{35131}, \country{Italy}}}

\abstract{We develop Hamiltonian mechanics on Aristotelian manifolds, which lack local boost symmetry and admit absolute time and space structures. We construct invariant phase space dynamics, define free Hamiltonians, and establish a generalized Liouville theorem. Conserved quantities are identified via lifted Killing vectors. Extending to kinetic theory, we show that the charge current and stress tensor reproduce ideal hydrodynamics at leading order, with the ideal gas law emerging universally. Our framework provides a geometric and dynamical foundation for systems where boost invariance is absent, with applications including but not limited to: condensed matter, active matter and optimization dynamics.}


\maketitle
\thispagestyle{empty}
{\noindent\thanks{$\,^a$\href{mailto:andrea.amoretti@ge.infn.it}{andrea.amoretti@ge.infn.it}}\\
\thanks{$\,\,\;^b$\href{mailto:danny.brattan@gmail.com}{danny.brattan@gmail.com}}\\
\thanks{$\,\,\;^c$\href{mailto:luca.martinoia@unipd.it}{luca.martinoia@unipd.it}}}
\clearpage
\setcounter{page}{1} 
\tableofcontents

\section{Introduction}\label{sec:intro}

{\noindent The foundations of classical and quantum mechanics are deeply rooted in the symmetries of spacetime. In both relativistic and Galilean frameworks, local boost invariance plays a central role in determining the structure of phase space and the dynamics of particles and fields. Boost invariance is embedded in the geometric structure of standard spacetime models - namely, Lorentzian manifolds for relativistic theories and Newton-Cartan spacetimes for non-relativistic mechanics \cite{Arnold}. In these settings, boosts link different inertial frames, enforcing a relativity principle that severely constrains both the admissible equations of motion and the symmetries of the underlying geometry \cite{PhysRevD.31.1841}.}

{\ However, there exists a wide and growing class of physical, biological, and computational systems where boost symmetry is either broken, irrelevant or fundamentally absent. In active matter systems and biological collectives, such as bird flocks or cell tissues (see e.g. \cite{RevModPhys.85.1143}), the dynamics often refer to an absolute frame dictated by a medium or background, invalidating the equivalence of inertial observers \cite{Grosvenor:2024vcn}. In condensed matter physics, non-relativistic effective field theories describing systems with anisotropic or Lifshitz-type scaling (see e.g. \cite{Sachdev_2011,Fradkin_2013}) frequently violate boost invariance while retaining well-defined Hamiltonian dynamics \cite{disalvo2025hydrodynamicsgeneralizedelectronictwoband}. Similar structures emerge in the study of particle-based optimization algorithms, such as particle swarm optimization (see e.g. \cite{10.1162/EVCO_r_00180}), where particles evolve under Hamiltonian-like rules in a configuration space for which imposition of inertial frames or boost transformations are overly restrictive \cite{Olfati}.}

{\ These examples and issues suggest that the insistence on local boost invariance may be unnecessarily restrictive in many contexts. Moreover, with the advent and utilisation of gauge/gravity dualities, there have been many investigations into the nature and structure of (quasi-)hydrodynamics for strongly coupled field theories (for example, see \cite{Amoretti:2022acb} for a review) including arguments that boost agnostic hydrodynamics is the framework in which to understand steady driven states \cite{Amoretti:2022ovc,Brattan:2024dfv,Amoretti:2024obt,Amoretti:2024jig,Ahn:2025odk,ahn2025holographicdbraneconstructionsdynamical}. Many results have been derived in these rather esoteric theories, however the general utility of the formalisms and results remains dubious. Naturally, we would then like to work at weak coupling with computational kinetic theories where the underlying theory is not only understood and calculable, but intuitive. However, without a generic framework for understanding systems with a lack of boost invariance, this has not previously been possible.}

{\ Motivated by this, we develop in this work the foundations of a boost-agnostic Hamiltonian mechanics, a framework for describing particle dynamics on manifolds where no local boost symmetry is present. The appropriate geometric setting for this construction is the Aristotelian manifold (see e.g. \cite{Penrose:1968ar,10.1063/1.1664490,deBoer:2017ing,Figueroa-OFarrill:2018ilb,de_Boer_2020}), which features a globally defined clock form (defining absolute time) and a degenerate spatial metric (defining absolute spatial intervals), but no group action implementing boost transformations. This geometry provides a natural structure in which to formulate theories where velocity is an observer-independent attribute, and no relativity principle holds.}

{\ Core to what we develop here is the consistent coupling of our exotic particles (i.e.~those with boost-agnostic Hamiltonians) to curved manifolds; this allows us to define rigorously an analogue of the stress-energy-momentum (SEM) tensor and charge currents. However, doing so is quite involved. Hence, after reviewing minimal necessary details on the geometry of Aristotelian manifolds, we begin by applying our formalism to a class of easy to grasp systems on a flat manifold. We extend it to kinetic theory, detailing the evolution of particle distribution functions on Aristotelian spacetimes. We demonstrate that the SEM tensor complex and charge current derived from such ensembles reproduce the ideal hydrodynamic form at leading order in a derivative expansion. Surprisingly, despite the absence of boost symmetry, the system still obeys the ideal gas law, a feature we trace back to the structure of invariant phase space volume and energy conservation associated with time-like Killing fields.}

{\ Having then demonstrated the utility of our expressions, we will subsequently develop the formal structure of our theory. In particular, we show how to construct Hamiltonian mechanics on general Aristotelian manifolds in a way that preserves core features such as symplectic structure, conservation laws, and Liouville’s theorem. We establish a formulation of phase space dynamics in which the Hamiltonian vector fields lie entirely in the horizontal subbundle of the cotangent bundle. This leads to a nontrivial generalisation of Liouville’s theorem, valid on constraint surfaces associated with worldline-reparametrisation-invariant dynamics, such as those of exotic (i.e. having a non-boost invariant Hamiltonian) free particles.}

{\ A key technical element of our construction is the application of the uplifts of Killing vectors from the base Aristotelian manifold to the phase space. Despite the absence of a boost group, many Aristotelian spacetimes admit rich isometry groups, whose infinitesimal generators can be lifted to symplectomorphisms on phase space. We show that these uplifted vector fields generate conserved quantities via a Poisson bracket structure, and that the Hamiltonians associated with them are naturally preserved under free dynamics. This allows us to construct a class of free Hamiltonians - defined entirely from Aristotelian invariant scalar quantities - that generalise the standard kinetic energy terms in Galilean or relativistic theories.}

{\ In summary, this work lays the geometric and dynamical foundations of boost-agnostic Hamiltonian mechanics. The formalism we develop:
\begin{enumerate}
  \item generalises phase space dynamics to settings without local boost symmetry;
  \item identifies the class of free Hamiltonians and their associated conservation laws;
  \item establishes Liouville’s theorem on constraint surfaces for exotic particles;
  \item derives hydrodynamic observables and demonstrates the universality of the ideal gas law.
\end{enumerate}
}

{\ These results open the door to a systematic treatment of systems - physical, biological, or computational - whose dynamics are best described in the absence of boost invariance, while retaining a fully Hamiltonian formulation.}

\section{Conservation laws according to Aristotle}\label{sec:arismanifold}

{\noindent We are indebted to many previous works on Aristotelian manifolds (a non-exhaustive list includes~\cite{Penrose:1968ar,10.1063/1.1664490,deBoer:2017ing,Figueroa-OFarrill:2018ilb,de_Boer_2020}) for developing various parts of the formalism gathered here. In this section we shall recount what is necessary for our discussion of Hamiltonian mechanics giving a brief overview, from the author's biased perspective, rather than being pedagogical. Those sufficiently familiar with the topics presented can skip to the next section.}

\subsection{Summary of Aristotelian manifold}\label{sec:review}

{\noindent In Aristotelian geometry there are two key invariants related to the motion of a point particle. These are the time between points $\Delta t$ on its world-line and the spatial distance $\Delta l$. All observers must agree on the value of these quantities independently of their relative motion and thus on the average velocity $\Delta l/\Delta t$. We thus introduce infinitesimal quantities $\tau$  and $h$, analogous to the metric in a Lorentzian space, such that we can define these invariant lengths between arbitrary points on the worldline of a particle moving on a given manifold $M$. In terms of $\tau$ and $h$  the observables $\Delta t$ and $\Delta l$ are
	\begin{eqnarray}
		\Delta t = \int\mathrm{d} \lambda \; \tau_{\mu} \dot{x}^{\mu}(\lambda) \; , \qquad  \Delta l = \int\mathrm{d} \lambda \; \sqrt{ h_{\mu \nu} \dot{x}^{\mu}(\lambda) \dot{x}^{\nu}(\lambda) } \; ,
	\end{eqnarray}
where $x^{\mu}(\lambda)$ is the particle's position on the manifold, $\dot{x}^{\mu} = \mathrm{d}x^{\mu}(\lambda)/d\lambda$ and $\lambda$ is a parameterisation of a particle's world-line. Notice that both these quantities, $\Delta t$ and $\Delta l$, are invariant under reparameterisations of the world-line parameter $\lambda \rightarrow \lambda'(\lambda)$.}

{\ Given $\tau_{\mu}$ and $h_{\mu \nu}$ we assume that we can construct the quantities $\nu^{\mu}$, $\tilde{h}^{\mu \nu}$ which satisfy
	\begin{subequations}
	\label{Eq:RelationsSIandinvSIIntro}
	\begin{eqnarray}
		\nu^{\mu} \tau_{\mu} = -1 \; , \qquad h_{\mu \nu} \nu^{\mu} = 0 \; , \qquad \tilde{h}^{\mu \nu} \tau_{\nu} = 0 \; , 
	\end{eqnarray}
and, in particular, the coordinate basis completeness relation
	\begin{eqnarray}
		\tilde{h}^{\mu \rho} h_{\rho \nu} = \delta\indices{^\mu_\nu} + \nu^{\mu} \tau_{\nu} \; .
	\end{eqnarray}
	\end{subequations}
Here $\tau_{\mu}$ and $h_{\mu \nu}$ are termed the ``structure invariants'', $\nu^{\mu}$ and $\tilde{h}^{\mu \nu}$ are the ``inverse structure invariants''. A precise construction for them in terms of the vielbeins and the local symmetry group is given in appendix \ref{appendix:review}. Subsequently we can define the volume form
  \begin{subequations}
  \begin{eqnarray}
    \label{Eq:etalocalIntro}
    \mathrm{vol}_{M} &=& e \; \mathrm{d}x^{0} \wedge \mathrm{d} x^{1} \wedge \ldots \wedge \mathrm{d}x^{d} \; , \\
    \label{Eq:edefIntro}
    e &=&  \mathrm{det}\left(\tau_{\mu},e^{i}_{\mu} \right)  \; . 
  \end{eqnarray}
  \end{subequations}
where we assume that one can determine vielbeins $e^{i}_{\mu}$ such that $h_{\mu \nu} = \delta_{ij} e^{i}_{\mu} e^{j}_{\nu}$. For further information please refer to appendix \ref{appendix:review}.}

{\ We now add a connection to our manifold. In particular, we shall also assume that we can find a structure invariant compatible connection, such that
    \begin{subequations}
    \begin{eqnarray}
        \nabla_{\mu} \tau_{\nu} &=& \partial_{\mu} \tau_{\nu} - \Gamma_{\mu \nu}^{\rho} \tau_{\rho} = 0  \; , \\
        	\label{Eq:VanishingDhconstraint}
        \nabla_{\mu} h_{\nu \rho} &=& \partial_{\mu} h_{\nu \rho} - \Gamma_{\mu \rho}^{\sigma} h_{\sigma \nu} - \Gamma_{\mu \nu} ^{\sigma} h_{\sigma \rho} = 0 \; .
    \end{eqnarray}
    \end{subequations}
This places the following constraints on particular projections of the connection coefficients
   \begin{subequations}
   \label{Eq:ConnectionConstraintsIntro}
    \begin{eqnarray}
        	        \Gamma_{\mu \nu}^{\rho} \tau_{\rho} 
        &=& \partial_{\mu} \tau_{\nu} \; , \\
               2 \Gamma_{\mu (\rho}^{\sigma} h_{\nu) \sigma}
        &=&  \partial_{\mu} h_{\nu \rho} \; .
    \end{eqnarray}
    \end{subequations}
The second equation \eqref{Eq:VanishingDhconstraint} is not so distinct from what one has for the usual Christoffel connection in a relativistic theory, consequently we can use the same trick of permuting indices to find
	\begin{subequations}
    \begin{eqnarray}
            \partial_{\mu} h_{\nu \lambda} + \partial_{\nu} h_{\mu \lambda} - \partial_{\rho} h_{\mu \lambda}
        &=& 2 \Gamma_{(\mu \nu)}^{\sigma} h_{\sigma \lambda}
            + \Sigma\indices{_{\mu \lambda}^{\sigma}} h_{\nu \sigma}
            + \Sigma\indices{_{\nu \lambda}^{\sigma}} h_{\mu \sigma} \; , \\
            \Sigma\indices{_{\mu \nu}^{\sigma}}
         &=& 2 \Gamma_{[\mu \nu]}^{\sigma} \;  ,
    \end{eqnarray}
    \end{subequations}
where $\Sigma\indices{_{\mu \nu}^{\sigma}}$ is a tensor, rather than a connection, called the torsion. A key difference from the standard relativistic case is that we don't know a posteriori that we can set this torsion to zero.}

{\ Consequently, taking the connection and projecting the raised index we have
    \begin{eqnarray}
            \Gamma_{\mu \nu}^{\rho}
        &=& \Gamma_{\mu \nu}^{\sigma} \left( - \tau_{\sigma} \nu^{\rho} + h_{\sigma \alpha} \tilde{h}^{\alpha \rho} \right) 
        =  - \nu^{\rho} \partial_{\mu} \tau_{\nu}
            + \Gamma_{(\mu \nu)}^{\sigma} h_{\alpha \sigma} \tilde{h}^{\alpha \rho}
            + \frac{1}{2} \Sigma\indices{_{\mu \nu}^{\sigma}} h_{\alpha \sigma} \tilde{h}^{\alpha \rho} \nonumber \\
        &=& - \nu^{\rho} \partial_{\mu} \tau_{\nu}
            + \frac{1}{2} \tilde{h}^{\rho \lambda} \left(  \partial_{\mu} h_{\nu \lambda} + \partial_{\nu} h_{\mu \lambda} - \partial_{\rho} h_{\mu \lambda} \right) \nonumber \\
        &\;& + \frac{1}{2} \tilde{h}^{\rho \lambda} \left( \Sigma\indices{_{\mu \nu}^{\sigma}} h_{\lambda \sigma} - \Sigma\indices{_{\mu \lambda}^{\sigma}} h_{\nu \sigma} - \Sigma\indices{_{\nu \lambda}^{\sigma}} h_{\mu \sigma}  \right) \; .
    \end{eqnarray}
It should be noted that this connection has an important non-uniqueness. If we can find a tensor $K\indices{_{\mu \nu}^{\rho}}$ such that
	\begin{eqnarray}
		\label{Eq:Ambiguity}
		K\indices{_{\mu \nu}^{\rho}} \tau_{\rho} = 0 \; , \qquad
		K\indices{_{\mu \nu}^{\lambda}} h_{\lambda \rho}  + K\indices{_{\mu \rho}^{\lambda}} h_{\lambda \nu}  = 0 \; ,	
	\end{eqnarray}
then we can shift $\Gamma_{\mu \nu}^{\rho} \rightarrow \tilde{\Gamma}_{\mu \nu}^{\rho} = \Gamma_{\mu \nu}^{\rho}  + K\indices{_{\mu \nu}^{\rho}}$ while still satisfying \eqref{Eq:ConnectionConstraintsIntro}. In the geometrisation of Newtonian gravity this non-uniqueness can be used to introduce the Newtonian gravitational potential. In particular, one posits the existence of an additional field to $\tau$ and $h$, the mass form $m$, and identifies
	\begin{eqnarray}
		K\indices{_{\mu \nu}^{\rho}} &=& \tilde{h}^{\rho \lambda} \tau_{( \mu} F_{\nu) \rho} \; , \qquad F_{\mu \nu} = 2 \partial_{[\mu} m_{\nu]} \; ,
	\end{eqnarray}
where under appropriate assumptions $m_{\mu} \sim \phi \tau_{\mu}$ with $\phi$ is the Newtonian gravitational potential \cite{PhysRevD.31.1841,Andringa_2013,Geracie_2015,Hartong_2015}. In future developments we hope to be able to use this freedom to impose external forces on our system in a usefully geometric manner.}

{\ As is done in \cite{de_Boer_2020} one may partly eliminate the torsion tensor from the expression by using Lie derivative of $h$ along the vector field $\nu$. This latter quantity can be written as
	\begin{eqnarray}
			\left(\mathcal{L}_{\nu} h \right)_{\mu \nu}
		&=& \nu^{\lambda} \partial_{\lambda} h_{\mu \nu} + \partial_{\mu} \nu^{\lambda} h_{\lambda \nu} + \partial_{\nu} \nu^{\lambda} h_{\lambda \mu} \nonumber \\
		&=& \nu^{\lambda} \left( \Sigma\indices{_{\lambda \mu}^{\sigma}} h_{\sigma \nu} + \Sigma\indices{_{\lambda \nu}^{\sigma}} h_{\sigma \mu} \right)
	\end{eqnarray}
where we have employed compatibility of the covariant derivative with the inverse structure invariants
	\begin{eqnarray}
		\nabla_{\mu} \nu^{\nu} &=& \partial_{\mu} \nu^{\nu} + \Gamma_{\mu \rho}^{\nu} \nu^{\rho} = 0 \; .
	\end{eqnarray}
The connection can then be expressed as
	\begin{subequations}
	\label{Eq:DerivativeCompatibilityCoordIntro}
	\begin{eqnarray}
		\Gamma_{\mu \nu}^{\rho}
	        &=& - \nu^{\rho} \partial_{\mu} \tau_{\nu}
            + \frac{1}{2} \tilde{h}^{\rho \lambda} \left(  \partial_{\mu} h_{\nu \lambda} + \partial_{\nu} h_{\mu \lambda} - \partial_{\rho} h_{\mu \lambda} \right) + \frac{1}{2} \tilde{h}^{\rho \lambda} \tau_{\nu} (\mathcal{L}_{\nu} h)_{\mu \lambda}  \nonumber \\
        &\;& + K\indices{_{\mu \nu}^{\rho}} \; , \\
        		K\indices{_{\mu \nu}^{\rho}} 
	&=& - \frac{1}{2} \tilde{h}^{\rho \lambda} \tilde{h}^{\alpha \beta} h_{\beta \nu} \left( \Sigma\indices{_{\alpha \lambda}^{\sigma}} h_{\mu \sigma}  - \Sigma\indices{_{\mu \alpha}^{\sigma}} h_{\lambda \sigma} \right)  - \frac{1}{2} \tilde{h}^{\rho \lambda} \Sigma\indices{_{\mu \lambda}^{\sigma}} h_{\nu \sigma}  \; ,
	\end{eqnarray}
	\end{subequations}
where $K\indices{_{\mu \nu}^{\rho}}$ is a tensor satisfying \eqref{Eq:Ambiguity} but built without the introduction of a mass one-form. We can then either view the torsion as some tensor we must supply in addition to the structure invariants, or define $K\indices{_{\mu \nu}^{\rho}}$ as is done in \cite{de_Boer_2020}.}

{\ For our structure invariant connection, we can see that the torsion is generically non-zero; and forcing it to be will at least impose a constraint on the clock form. This means that certain familiar results from (pseudo-)Riemannian geometry become more complicated. For example, the commutator of derivatives acting on any tensor fields includes an additional term
 	\begin{subequations}
	\begin{eqnarray}
			\label{Eq:CommutatorTensor}
			 \left[ \nabla_{\mu}, \nabla_{\nu} \right] T\indices{_{\rho_{1} \ldots \rho_{n}}^{\sigma_{1} \ldots \sigma_{m}}}
		&=&  \sum_{j=1}^{m} R\indices{_{\mu \nu \lambda}^{\sigma_{j}}} T\indices{_{\rho_{1} \ldots \rho_{n}}^{\sigma_{1} \ldots \sigma_{j-1} \lambda \sigma_{j+1} \ldots \sigma_{m}}} \nonumber \\
 		&\;& - \sum_{i=1}^{n}  R\indices{_{\mu \nu \rho_{i}}^{\lambda}} T\indices{_{\rho_{1} \ldots \rho_{i-1} \lambda \rho_{i+1} \ldots \rho_{n}}^{\sigma_{1} \ldots \sigma_{m}}} \nonumber \\
		&\;& -  \Sigma\indices{_{\mu \nu}^{\lambda}} \nabla_{\lambda} T\indices{_{\rho_{1} \ldots \rho_{n}}^{\sigma_{1} \ldots \sigma_{m}}} \; , 
	\end{eqnarray}
where we have defined the curvature tensor in the usual manner
 	\begin{eqnarray}
		\label{Eq:CurvatureDef}
		 R\indices{_{\mu \nu \sigma}^{\rho}} = 
		2 \left( \partial_{[\mu} \Gamma^{\rho}_{\nu] \sigma}  + \Gamma^{\rho}_{[\mu| \alpha}  \Gamma^{\alpha}_{|\nu] \sigma}  \right) \; . 
	\end{eqnarray}
	\end{subequations}
Extreme caution must be exercised by those familiar with relativistic physics in manipulating $R\indices{_{\mu \nu \rho}^{\sigma}}$ as common features of the Riemann tensor, such as certain permutation symmetries of the lowered indices detailed in lemma \ref{lem:Riemann} of appendix \ref{appendix:review}, fail to hold for the Aristotelian curvature tensor.}

\subsection{Killing vectors on Aristotelian manifolds }\label{sec:AristotleanKilling}

{\noindent A feature of key importance in our work will be the relationship between symmetries and the notion of a particle being ``free''. In this section we shall discuss the generalisation of infinitesimal symmetries to the Aristotelian case. In particular, one of the benefits of building our system on top of curved manifolds is that it allows us to precisely define the analogue of the stress-energy-momentum (SEM) tensor - and thus energy and momentum as components of this tensor - in an Aristotelian setup.}

{\ Suppose we are given some functional, such as an action for a particle, defined on our Aristotelian manifold. We can introduce tensors $\tilde{T}_{\mu}$ and $\tilde{T}_{\mu \nu}$ by the variation of the action with respect to the (inverse) structure invariants\footnote{We note that there is a small ambiguity in the definition of $\tilde{T}_{\mu \nu}$ under $\tilde{T}_{\mu \nu} \rightarrow \tilde{T}_{\mu \nu} + c \tau_{\mu} \tau_{\nu}$ as $\tau_{\mu} \delta h^{\mu \nu} \tau_{\nu}=0$. Consequently, we will choose a convention where the totally time-like component of this tensor vanishes as it does not enter into the conservation equation \eqref{Eq:NonconvariantSEMcons}. With that said, this decomposition of the SEM tensor complex in terms of the inverse structure invariants can be compared to that of \cite{de_Boer_2020}, where
  \begin{align}
         & \tilde{T}_{\alpha \nu} = - T^{\sigma} \left( h_{\alpha \sigma} \tau_{\nu} + h_{\nu \sigma} \tau_{\alpha} \right) + h_{\alpha \sigma} T^{\sigma \beta} h_{\beta \nu} \; , \qquad
             \tilde{T}_{\nu} = - \tau_{\sigma} T^{\sigma} \tau_{\nu} + \tau_{\sigma} T^{\sigma \beta} h_{\beta \nu} \; .
  \end{align}
Again, where we have taken the convention that the totally time-like components of the $2$-index tensors to vanish.} such that
	\begin{eqnarray}
		\label{Eq:variationaction}
		\delta \mathscr{S} &=& \int\mathrm{d}^{d+1}x \; e \left[ \tilde{T}_{\mu} \delta \nu^{\mu} - \frac{1}{2} \tilde{T}_{\mu \nu} \delta \tilde{h}^{\mu \nu} \right] \; , 
	\end{eqnarray}
where $\mathscr{S}=\int\mathrm{d}^{d+1}x\ eL$ is the action (or other relevant functional), $L$ the Lagrangian density, and $\tilde{T}_{\mu \nu}$ is symmetric in its indices. An infinitesimal diffeomorphism acts on a vector field and contravariant 2-tensor as
	\begin{eqnarray}
		\label{Eq:diffeomorphismofinverse}
		\nu^{\mu} \rightarrow \nu^{\mu} +  (\mathcal{L}_{\xi} \nu)^{\mu} \; , \qquad
		\tilde{h}^{\mu \nu} \rightarrow \tilde{h}^{\mu \nu} +  (\mathcal{L}_{\xi} \tilde{h})^{\mu \nu} \; , 
	\end{eqnarray}
where $\xi$ is the vector field generating the diffeomorphism. Assuming $\mathscr{S}$ to be diffeomorphism invariant, as one expects for any reasonable theory, it is not hard to show that if the variation of our action \eqref{Eq:variationaction} is generated by the infinitesimal diffeomorphism then
   \begin{eqnarray}
   		\delta \mathscr{S}
	&=& \int\mathrm{d}^{d+1}x \; e \left[ \frac{1}{e} \frac{\delta (e L)}{\delta \nu^{\mu}} \mathcal{L}_{\xi} \nu^{\mu}  +  \frac{1}{e}  \frac{\delta (e L)}{\delta \tilde{h}^{\mu \nu}} \mathcal{L}_{\xi} \tilde{h}^{\mu \nu}  \right] \nonumber \\
	&=&  \int\mathrm{d}^{d+1} x \; e \left[  \frac{1}{e} \partial_{\rho} \left( - e \tilde{T}_{\sigma} \nu^{\rho} +  e \tilde{T}_{\sigma \nu}  \tilde{h}^{\nu \rho} \right) +  \tilde{T}_{\mu} \partial_{\sigma} \nu^{\mu} 
		 - \frac{1}{2} \tilde{T}_{\mu \nu} \partial_{\sigma} \tilde{h}^{\mu \nu} \right] \xi^{\sigma} = 0 \; , \qquad
   \end{eqnarray}
up to boundary terms which we assume vanish. If we further define the SEM tensor complex as the $(1,1)$-tensor object
  \begin{eqnarray}
  	\label{Eq:defSEMtensorcomplex}
    T\indices{^\mu_\nu} &=& - \nu^{\mu} \tilde{T}_{\nu} + \tilde{h}^{\mu \rho} \tilde{T}_{\rho \nu} \; . 
  \end{eqnarray}
 then the following (conservation) equation,
	\begin{eqnarray}
		\label{Eq:NonconvariantSEMcons}
		\frac{1}{e} \partial_{\rho} \left( e T\indices{^{\rho}_{\sigma}}  \right) + \tilde{T}_{\mu} \partial_{\sigma} \nu^{\mu} 
		 - \frac{1}{2} \tilde{T}_{\mu \nu} \partial_{\sigma} \tilde{h}^{\mu \nu} &=& 0 \; ,
	\end{eqnarray}
is an identity that must be satisfied by any solution for the particle's motion. This is the analogue of SEM tensor conservation in a relativistic theory.}

{\ Among the diffeomorphism transformations \eqref{Eq:diffeomorphismofinverse} generated by different vector fields $\xi$ there is a special class, those which leave the inverse structure invariants unchanged, which motivates the following definition:}

\begin{definition}[Aristotelian Killing vector]\label{def:Killingfield}
{\ A Killing vector field $\xi$ on $M$, an Aristotelian manifold, is defined to satisfy the following relations
 \begin{subequations}
 \label{Eq:Killingconditions}
  \begin{eqnarray}
    \mathcal{L}_{\xi} \tau &=& 0 \; , \\
    \mathcal{L}_{\xi}  h &=& 0  \; ,
  \end{eqnarray}
  \end{subequations}
where $\tau$ and $h$ are the structure invariants.}
\end{definition}

{\ The above definition generalises the usual notion of Killing vector in Lorentzian geometry (namely constancy of the metric under the Lie derivative) to Aristotelian geometries and extends in the expected way to multiple Killing vectors. Once we have developed charge conservation in section \ref{sec:aristotlechargecons}, we shall show how such Aristotelian Killing vectors can be used to isolate the energy components of the SEM tensor complex on an arbitrarily curved manifold.}

{\ As a consequence of the relationships between structure invariants and their inverses \eqref{Eq:RelationsSIandinvSIIntro}, one can show that 
		\begin{eqnarray}
			\label{Eq:InverseKillingconditions}
			 \mathcal{L}_{\xi} \nu &=& 0 \; , \qquad  		 \mathcal{L}_{\xi}  \tilde{h} = 0  \; ,
		\end{eqnarray}
whenever \eqref{Eq:Killingconditions} hold. All these conditions, \eqref{Eq:Killingconditions} and \eqref{Eq:InverseKillingconditions}, can be covariantised and in terms of the structure-invariant-compatible covariant derivative \eqref{Eq:DerivativeCompatibilityCoordIntro} one finds
	\begin{subequations}
	\label{Eq:CovKillingConditions}
  \begin{eqnarray}
    	\mathcal{L}_{\xi} \tau &=& \tau_{\rho} \left( \nabla_{\nu} \xi^{\rho} - \Sigma\indices{_{\nu \mu}^{\rho}} \xi^{\mu} \right) \mathrm{d} x^{\nu} = 0 \; , \\
    \mathcal{L}_{\xi} \nu &=& - \nu^{\sigma} \left( \nabla_{\sigma} \xi^{\mu} - \Sigma\indices{_{\sigma \rho}^{\mu}} \xi^{\rho}  \right) \partial_{\mu} = 0 \; , \\
     \mathcal{L}_{\xi} h &=& \left[ h_{\mu \sigma} \left( \nabla_{\rho} \xi^{\mu} - \Sigma\indices{_{\rho \nu}^{\mu}} \xi^{\nu} \right) + h_{\mu \rho} \left( \nabla_{\sigma} \xi^{\mu} - \Sigma\indices{_{\sigma \nu}^{\mu}} \xi^{\nu} \right) \right] \mathrm{d}x^{\rho} \otimes \mathrm{d}x^{\sigma} = 0 \; , \qquad \\
    	   \mathcal{L}_{\xi} \tilde{h} 
        &=& - \left[ \tilde{h}^{\mu \rho} \left( \nabla_{\rho} \xi^{\nu} - \Sigma\indices{_{\rho \sigma}^{\nu}} \xi^{\sigma} \right) + \tilde{h}^{\nu \rho} \left( \nabla_{\rho} \xi^{\mu} - \Sigma\indices{_{\rho \sigma}^{\mu}} \xi^{\sigma} \right) \right] \partial_{\mu} \otimes \partial_{\nu} = 0 \; . \qquad
  \end{eqnarray}
  	\end{subequations}
Moreover, combining the above with identities from appendix \ref{appendix:review} we see that
	\begin{eqnarray}
		\label{Eq:KillingIdentityTrace}
		\nabla_{\mu} \xi^{\mu} - \Sigma\indices{_{\mu \nu}^{\mu}} \xi^{\nu}
		&=& \frac{1}{e} \partial_{\mu} \left( e \xi^{\mu} \right) = 0 
	\end{eqnarray}
and consequently
	\begin{eqnarray}
		\mathcal{L}_{\xi} \mathrm{vol}_{M} = \frac{1}{e} \partial_{\mu} \left( e \xi^{\mu} \right)  \mathrm{vol}_{M} = 0 \; , 
	\end{eqnarray}
where $\mathrm{vol}_{M}$ is the volume form given in \eqref{Eq:etalocalIntro}.}

{\ In a Lorentzian geometry with a covariant derivative compatible connection, the Levi-Civita connection and subsequently all the curvature terms are invariants under the action of the Lie derivative along a Killing vector field. More generally however, the Lie derivative of the coordinate basis connection coefficients and of the curvature tensor take the form
	\begin{subequations}
	\label{Eq:Variationofcurvaturequantities}
	\begin{eqnarray}
		\label{Eq:Variationofconnection}
		\mathcal{L}_{\xi} \Gamma_{\mu \nu}^{\rho} &=& \nabla_{\mu} \left( \nabla_{\nu} \xi^{\rho} - \Sigma\indices{_{\nu \sigma}^{\rho}} \xi^{\sigma} \right) + R\indices{_{\sigma \mu \nu}^{\rho}} \xi^{\sigma} \; , 
	\end{eqnarray} 
and
	\begin{eqnarray}
		 \mathcal{L}_{\xi} R\indices{_{\mu \nu \rho}^{\sigma}} = \nabla_{\mu} \mathcal{L}_{\xi} \Gamma_{\nu \rho}^{\sigma} - \nabla_{\nu} \mathcal{L}_{\xi} \Gamma_{\mu \rho}^{\sigma}
		+ \Sigma\indices{_{\mu \nu}^{\lambda}} \mathcal{L}_{\xi} \Gamma_{\lambda \rho}^{\sigma} \; ,
	\end{eqnarray}
	\end{subequations}
respectively. One can find the derivation in \cite{Yano1957-lv}. In Lorentzian spacetimes with torsionless, metric-compatible connections it is possible to show that the right hand side of \eqref{Eq:Variationofconnection} vanishes when $\xi$ is Killing. This does not generally apply to Aristotelian spaces\footnote{An exception are particular contractions of \eqref{Eq:Variationofcurvaturequantities} with $\tau$ and $\nu$ which vanish on account of \eqref{Eq:RiemannContracttau} and \eqref{Eq:RiemannContractnu}.} nor indeed, Lorentzian connections with torsion \cite{Peterson:2019uzn}. Consequently, curvature quantities will not be invariants under flows generated by Aristotelian Killing vector fields. As such, some authors require invariance of the connection as part of the definition of a Killing vector, which ensures that geodesics map to geodesics under flow along this vector. However, this is not necessary for the SEM tensor complex to be preserved and we shall not assume this. Consequently, this greatly restricts the quantities one can write that are invariant scalars under integral curves of the Killing fields, a fact that will be important when we define our free Hamiltonians.}

\subsection{Charge, energy and momentum in Aristotelian geometries }\label{sec:aristotlechargecons}

{\noindent With Killing vectors now defined, we can develop the analogue of charge conservation on an Aristotelian manifold. In particular, we shall show that if a vector field satisfies a particular conservation law within a closed submanifold $S \subset M$, then an integral over the boundary of $S$ vanishes. In cases where the volume form can be decomposed suitably, this corresponds to standard charge conservation and then we shall identify notions of charge, energy and momentum density for systems moving on $M$.}

{\ Let $J$ be a vector field in $M$. We first note that our volume form on $M$, $\mathrm{vol}_{M}$, defined in \eqref{Eq:etalocalIntro} satisfies the following relation
  \begin{eqnarray}
    \label{Eq:LieDerivativeCurrent}
    \mathcal{L}_{J} \mathrm{vol}_{M} &=& \left( \nabla_{\mu} J^{\mu} - \Sigma\indices{_{\mu \nu}^{\mu}} J^{\nu} \right) \mathrm{vol}_{M} \; .
  \end{eqnarray}
To see this we work in local coordinates so that when we compute that the Lie derivative of $\mathrm{vol}_{M}$ along $J$ we find
	\begin{subequations}
 	\begin{eqnarray}
		 	\mathcal{L}_{J} \mathrm{vol}_{M} 
		 &=& \left( \partial_{\mu} J^{\mu} + \mathrm{det}(e^{I}_{\mu})^{-1} J^{\nu} \partial_{\nu} \mathrm{det}(e^{I}_{\mu})  \right) \mathrm{vol}_{M} \; ,  \\
		 	e^{I}_{\mu}
		&=& (\tau_{\mu},e^{i}_{\mu}) \; . 
	\end{eqnarray}
	\end{subequations}
We can then find that
	\begin{eqnarray}
			\mathrm{det}(e^{I}_{\nu})^{-1} \partial_{\mu} \mathrm{det}(e^{I}_{\nu})
		= e_{I}^{\nu} \partial_{\mu} e^{I}_{\nu} = - \nu^{\nu} \partial_{\mu} \tau_{\nu} + e_{i}^{\nu} \partial_{\mu} e^{i}_{\nu} = \Gamma_{\mu \nu}^{\nu}  \; .
	\end{eqnarray}
Thus we conclude that
	\begin{eqnarray}
			\mathcal{L}_{J} \mathrm{vol}_{M} = \left(  \partial_{\mu} J^{\mu} +  \Gamma_{\nu \mu}^{\mu} J^{\nu} \right) \mathrm{vol}_{M} =  \left( \nabla_{\mu} J^{\mu} - \Sigma\indices{_{\mu \nu}^{\mu}} J^{\nu} \right) \mathrm{vol}_{M} \; , 
	\end{eqnarray}
which is the result displayed in \eqref{Eq:LieDerivativeCurrent}. Subsequently, applying Cartan's magic formula which states that for any form $\omega$ and vector $X$ the following holds
  \begin{eqnarray}
   \label{Eq:CartanMagic}
   \mathcal{L}_{X} \omega &=& \mathrm{d} \left( i_{X} \omega \right) - i_{X} \mathrm{d} \omega \; ,
  \end{eqnarray}
to \eqref{Eq:LieDerivativeCurrent} we find
  \begin{eqnarray}
    \label{Eq:Conservationvsinterior}
    \mathrm{d} \left( i_{J} \mathrm{vol}_{M} \right) &=& \left( \nabla_{\mu} J^{\mu} - \Sigma\indices{_{\mu \nu}^{\mu}} J^{\nu} \right) \mathrm{vol}_{M} \; ,
  \end{eqnarray}
where $i_{J}$ is the interior product and we have used that $\mathrm{d} \mathrm{vol}_{M}=0$ as $\mathrm{vol}_{M}$ is a top-form on the manifold $M$.}

{\ We can then employ the generalised Stoke's theorem, which states that for a smooth closed set $S \subset M$ the integral of a form $\omega$ over the boundary $\partial S$ of the closed set is related to the integral over $S$ of its exterior derivative by
    \begin{eqnarray}
    	\label{Eq:Stokes}
        \int_{S} \mathrm{d} \omega &=& \int_{\partial S} \iota_{S}^{*} \omega
    \end{eqnarray}
where $\iota_{S}:\partial S \rightarrow S$ is the inclusion map and $^{*}$ indicates the pullback. Thus \eqref{Eq:Conservationvsinterior}, once integrated over any closed surface $S$, tells us that
  \begin{eqnarray}
    \int_{\partial S} \iota_{S}^{*} (i_{J} \mathrm{vol}_{M}) &=& \int_{S} \mathrm{vol}_{M} \left( \nabla_{\mu} J^{\mu} - \Sigma\indices{_{\mu \nu}^{\mu}} J^{\nu} \right) \; .
  \end{eqnarray}
As the surface $S$ is completely arbitrary we make the following identification (with the usual caveats in doing so)
  \begin{eqnarray}
    \label{Eq:ChargeConservation}
  	\nabla_{\mu} J^{\mu} - \Sigma\indices{_{\mu \nu}^{\mu}} J^{\nu} = 0 \qquad \Leftrightarrow \qquad  \int_{\partial S} \iota_{S}^{*} (i_{J} \mathrm{vol}_{M})  =0 \; .
  \end{eqnarray}
This conservation law for the current can also be rewritten without the covariant derivative as 
  \begin{eqnarray}
    \frac{1}{e} \partial_{\mu} \left(  e J^{\mu} \right) &=& 0 \; ,
  \end{eqnarray}
where we remind the reader that $e$ is the volume-form scalar introduced in \eqref{Eq:edefIntro}.}

{\ Let us assume for a moment that we have a manifold where $\tau = \mathrm{d}t$, so that the manifold is foliated by surfaces of constant $t$. Consider the submanifold $S$ between two hypersurfaces $t=t_{0}$ and $t=t_{1}>t_{0}$ and let us write the volume form as
  \begin{eqnarray}
  	\label{Eq:ChargeVolumeFormDecomp}
    \mathrm{vol}_{M} &=& \mathrm{d} t \wedge \sigma_{\tau} \; .
  \end{eqnarray}
That we can do this locally follows from the fact that, under mild regularity conditions, one can construct a Gelfand-Leray form \cite{Arnold2012-it}, $\sigma_{\tau} = \mathcal{L}_{V} \mathrm{vol}_{M}$, where $V$ is any vector field satisfying $\tau[V]=1$. Subsequently, we find
  \begin{eqnarray}
    \iota_{S}^{*} (i_{J} \mathrm{vol}_{M}) &=& \iota_{S}^{*} (i_{J} (\mathrm{d} t \wedge \sigma_{\tau})) = \left\{ \begin{array}{c}
                  (J^{t})_{t=t_{0}} \sigma_{\tau} \; , \qquad t=t_{0} \; ,  \\
                  -(J^{t})_{t=t_{1}} \sigma_{\tau} \; , \qquad t=t_{1}
                \end{array} \right. \; . 
  \end{eqnarray}
where the minus sign accounts for whether $J^{t}$ points into the closed submanifold ($t=t_{0}$) or out of it ($t=t_{1}$). Thus if $J$ satisfies the conservation equation \eqref{Eq:ChargeConservation} we have
	\begin{eqnarray}
		\int_{t=t_{1}} \;  J^{t} \sigma_{\tau} &=& \int_{t=t_{0}} \; J^{t}\sigma_{\tau} \; . 
	\end{eqnarray}
The interpretation of this equation is that the total charge is conserved in time. In general, we then can identify the charge (or number) density associated with a given current $J^{\mu}$ by
	\begin{eqnarray}
		\label{Eq:GenericChargeDensityDef}
		n &=& \tau_{\mu} J^{\mu} 	\; . 
	\end{eqnarray}
The existence of such a scalar quantity in terms of a (globally defined) conserved vector field $J^{\mu}$ extends to arbitrary Aristotelian manifolds, not just ones with foliations by a time coordinate, although its interpretation as the charge density is cleanest in such simplified cases.}

{\ Turning now to the SEM tensor complex defined in \eqref{Eq:defSEMtensorcomplex}, we note that upon contracting $T\indices{^\mu_\nu}$ with an arbitrary vector field $X$ we generate a vector field i.e. $T\indices{^\mu_\nu} X^{\nu}$. Subsequently, following our argument above, we can relate conservation of the SEM tensor complex to conservation of the current $T\indices{^\mu_\nu}$ passing through a closed surface in $M$. Replacing $J^{\mu}$ by $T\indices{^\mu_\nu} X^{\nu}$ in the above proof of current conservation we find
  \begin{eqnarray}
    \int_{\partial S} \iota_{S}^{*} (i_{T X} \mathrm{vol}_{M}) &=& \int_{S} \mathrm{vol}_{M} \left( \nabla_{\mu} \left( T\indices{^{\mu}_{\nu}} X^{\nu} \right) - \Sigma\indices{_{\mu \nu}^{\mu}} T\indices{^{\nu}_{\sigma}}  X^{\sigma} \right) \nonumber \\
    &=&  \int_{S} \mathrm{vol}_{M} \left( \nabla_{\nu} X^{\sigma} - \Sigma\indices{_{\nu \rho}^{\sigma}} X^{\rho} \right) T\indices{^{\nu}_{\sigma}} \nonumber \\
     &\;& +  \int_{S} \mathrm{vol}_{M} \left( \nabla_{\mu} T\indices{^{\mu}_{\sigma}} - \Sigma\indices{_{\mu \nu}^{\mu}} T\indices{^{\nu}_{\sigma}} + \Sigma\indices{_{\nu \sigma}^{\rho}}  T\indices{^{\nu}_{\rho}} \right) X^{\sigma} \; . \qquad
  \end{eqnarray}
The integrand of the final term,
	\begin{eqnarray}
		\label{Eq:CovConsSEM}
		 \nabla_{\mu} T\indices{^{\mu}_{\sigma}} - \Sigma\indices{_{\mu \nu}^{\mu}} T\indices{^{\nu}_{\sigma}} + \Sigma\indices{_{\nu \sigma}^{\rho}}  T\indices{^{\nu}_{\rho}} \; ,
	\end{eqnarray}
is no more than the covariantised left hand side of \eqref{Eq:NonconvariantSEMcons} in the presence of a structure-invariant-compatible covariant derivative. Using that \eqref{Eq:CovConsSEM} vanishes, we do not however find an integral for total charge conservation. Instead we see that
	\begin{eqnarray}
		\label{Eq:RelatingchargesanddivergenceX}
		 \int_{\partial S} \iota_{S}^{*} (i_{T X} \mathrm{vol}_{M})  -  \int_{S} \mathrm{vol}_{M} \left( \nabla_{\nu} X^{\sigma} - \Sigma\indices{_{\nu \rho}^{\sigma}} X^{\rho} \right) T\indices{^{\nu}_{\sigma}} = 0 \; . 
	\end{eqnarray}
whenever the vector field $X$ is arbitrary. However, if we further choose $X$ to be an Aristotelian Killing vector satisfying \eqref{Eq:CovKillingConditions}, we find
	\begin{eqnarray}
			 \left( \nabla_{\nu} \xi^{\sigma} - \Sigma\indices{_{\mu \nu}^{\mu}} \xi^{\sigma} \right) T\indices{^{\nu}_{\sigma}}
		 &=&  \left( \nabla_{\nu} \xi^{\sigma} - \Sigma\indices{_{\mu \nu}^{\mu}} \xi^{\sigma} \right) \nu^{\nu} \tilde{T}_{\sigma} \nonumber \\
		  &\;&+  \left( \nabla_{\nu} \xi^{\sigma} - \Sigma\indices{_{\mu \nu}^{\mu}} \xi^{\sigma} \right) \tilde{h}^{\rho \nu} \tilde{T}_{(\rho \sigma)}  = 0 \; , 
	\end{eqnarray}
where we have  used that $\tilde{T}_{\rho \sigma}$ is symmetric in its indices. Consequently, the vector field $T\indices{^\mu_\nu} \xi^{\nu}$, whose projections represent energy and spatial momentum, is conserved i.e.
	\begin{eqnarray}
		\int_{\partial S} \iota_{S}^{*} (i_{T \xi} \mathrm{vol}_{M})  &=& 0 \; . 
	\end{eqnarray}
Naturally, if we can decompose the volume form appropriately such as in \eqref{Eq:ChargeVolumeFormDecomp}, we would like to interpret the result in terms of energy conservation. For this interpretation to make physical sense however, we must ensure that $\xi^{\mu}$ is a future pointing time-like vector field i.e.
	\begin{eqnarray}
		\tau_{\mu} \xi^{\mu} > 0 \; . 
	\end{eqnarray}
Subsequently, we introduce the energy relative to $\xi^{\mu}$ defined by
		\begin{eqnarray}
			\label{Eq:GenericEnergyDef}
			\epsilon_{\xi} &=& - \tau_{\mu} T\indices{^\mu_\nu} \xi^{\nu} \; .
		\end{eqnarray}
When we consider kinetic theory and hydrodynamics we will identify $\xi$, up to normalisation, with the fluid velocity $u^{\mu}$.}

\section{Application: collisions, kinetic theory and hydrodynamics}\label{sec:applications}

{\noindent In section \ref{sec:Hamiltonian} we shall develop the full formalism of Hamiltonian mechanics on Aristotelian manifolds. However, this development is formal and in many relevant applications we shall not require the full generality developed in that section; moreover this can obscure the important physical details. The purpose of that development is to justify our identifications of charge current and SEM tensor complex through coupling the system to curved backgrounds. Here then we state but a few of the results and apply them to an example we feel is physically appealing.}

{\ Let us introduce Cartesian coordinates on the Aristotelian manifold $x^{0}, \ldots x^{d}$ such that $\tau = (1,\vec{0})$ and $h=\delta_{ij}$.  We introduce a momentum $p_{\mu} = (p_{0},p_{i})$ and consider the following free Hamiltonians
	\begin{subequations}
	\label{Eq:Application}
	\begin{eqnarray}
		\label{eq:HSingleDispersion}
		H_{*} = \lambda \left( - p_{0} + \tilde{H}(\vec{p}^2) \right) \; .
	\end{eqnarray}
where, on-shell (i.e. on a level-set of $H_{*}$) we can solve for $p_{0}$ to find
	\begin{eqnarray}
		\label{Eq:p0sub}
		p_{0} =  \tilde{H}(\vec{p}^2)  - \frac{c}{\lambda} \; .
	\end{eqnarray}
One can think of $p_{0}$ as the usual ``energy'' of the particle and $\vec{p}$ as the spatial momentum. The constant $c$ represents the choice of zero energy. The function $\tilde{H}(\vec{p}^2)$ is in principle completely arbitrary and thus generally boost invariance is broken except in the very special case $\tilde{H}(\vec{p}^2)=\frac{\vec{p}^2}{2m}$ i.e. the standard non-relativistic kinetic energy. We state here that the Hamiltonian equations, governing the motion of a particle, can be written as	
	\begin{eqnarray}
		\label{Eq:HamiltonsEqns1}
		\frac{d}{d\lambda} x^{\mu}(\lambda) &=& \frac{\partial H}{\partial p_{\mu}(\lambda)} = \left\{ H, x^{\mu}(\lambda) \right\} \; , \\
		\label{Eq:HamiltonsEqns2}
		\frac{d}{d\lambda} p_{\mu}(\lambda) &=& - \frac{\partial H}{\partial x^{\mu}(\lambda)} = - \left\{ H, p_{\mu}(\lambda) \right\}  \; ,
	\end{eqnarray}
where $\lambda$ is a parameter describing the worldline of a given particle, the above applies independently both of the coordinate system and choice of Hamiltonian while
	\begin{eqnarray}
		\label{Eq:PoissonLocal}
		\left\{ F, G \right\} := \frac{\partial F}{\partial p_{\mu}} \frac{\partial G}{\partial x^{\mu}} - \frac{\partial G}{\partial p_{\mu}} \frac{\partial F}{\partial x^{\mu}} \; ,
	\end{eqnarray}
is the Poisson bracket. That these are the correct equations will be discussed in section \ref{sec:Hamiltonian}. Meanwhile, in the same section, we shall find that the charge current and SEM tensor complex of this system \eqref{Eq:p0sub} take the form
	\begin{eqnarray}
		\label{Eq:SingleDisperionCurrentSEMIntegralSimpleI}
		J^{\mu} \partial_{\mu} &=& \left(  \int\mathrm{d}^{d}p \; f \right) \partial_{t}+ \left( 2 h^{ij} \int\mathrm{d}^{d}p \; \left(\frac{\partial \tilde{H}}{\partial \vec{p}^2}\right) p_{j} f \right) \partial_{i} \; , \\
		\label{Eq:SingleDisperionCurrentSEMIntegralSimpleII}
		T\indices{^\mu_\nu} \partial_{\mu} \otimes \mathrm{d}x^{\nu} &=& \left( \int\mathrm{d}^{d}p \; \tilde{H} f \right) \partial_{t} \otimes \mathrm{d}t  +  \left( \int\mathrm{d}^{d}p \;  p_{i} f \right) \partial_{t} \otimes \mathrm{d}x^{i} \nonumber \\
		&\;& +  \left( 2 h^{ik} \int\mathrm{d}^{d}p \; \frac{\partial \tilde{H}}{\partial \vec{p}^2} p_{k} \tilde{H} f \right) \partial_{i} \otimes \mathrm{d}t \nonumber \\
		&\;& +  \left( 2 h^{ik} \int\mathrm{d}^{d}p \; \frac{\partial \tilde{H}}{\partial \vec{p}^2} p_{k} p_{j} f \right) \partial_{i} \otimes \mathrm{d}x^{j} \; , 
	\end{eqnarray}
where $f=f(t,\vec{x},\vec{p})$ is the single-particle distribution function (not necessarily in thermodynamic equilibrium). This function, $f$, measures the probability for a single particle to have a particular momentum and position. The function $f$ in the absence of collisions satisfies
	\begin{eqnarray}
			\label{Eq:ModifiedBoltzmannEqn}
			\frac{\partial f}{\partial t} + \frac{\partial \tilde{H}}{\partial p_{i}} \frac{\partial f}{\partial x^{i}} &=& 0 \; , 
	\end{eqnarray}
	\end{subequations}
for our Hamiltonians $H_{*}$ defined in \eqref{eq:HSingleDispersion} on flat manifolds without collisions. Extending the above equations to the case of collisions we shall do in this section. Together, these expressions \eqref{Eq:Application} are sufficient for the applications we discuss here and we hope the reader will accept them for now as a minor generalisation of what they already know.}

\subsection{Elastic scattering }

{\noindent In the next section we wish to discuss kinetic theory which inevitably involves the consideration of collisions. There are some peculiarities in the realisation of elastic scattering (even in flat space) of our exotic particles that are worthy of discussion and we shall address them first.}

{\ We concern ourselves the analogue of $2 \rightarrow 2$ elastic scattering between two hard spheres of radius $R$. Generally $p_{\mu}$ is neither conserved under Hamiltonian flow nor flows generated by Killing fields due to curvature terms\footnote{The contraction $p_{\mu} \nu^{\mu}$ is the exception as it is conserved both under free Hamiltonian flow and any Killing flows.}. However, we can always assume that scattering happens over a region where such curvature terms are negligible or, as we are doing, we can work in flat space. This will be sufficient to discuss their distinct behaviour. We will also further assume that our particles are distinguishable so that we can append the labels $(1)$ and $(2)$ to their respective momenta.}

\begin{figure}[h] 
  \centering
\begin{tikzpicture}[scale=1.2, every node/.style={font=\small}]

  \def\L{2.4}

  \def\R{0.8}

  \coordinate (p11) at (0,0);
  \coordinate (p12) at (\L,0);
  \coordinate (p22) at (\L+1.414*\R,-1.414*\R);
  \coordinate (p21) at (\L+1.414*\R,-\L-1.414*\R);
  \coordinate (p32) at (\L+0.5*\L,0.866*\L);
  \coordinate (p42) at (\L+1.414*\R+0.866*\L,-1.414*\R+0.5*\L);
  \coordinate (corner) at (\L+1.414*\R,0);
  
  \draw[->, thick] (p11) -- (p12) node[midway, above] {$\vec{p}_{(1)}$};
  \draw[->, thick] (p21) -- (p22) node[midway, left, yshift=-5pt] {$\vec{p}_{(2)}$};
  \draw[->, thick] (p12) -- (p32) node[midway, left, yshift=2pt] {$\vec{p}_{(1)}'$};
  \draw[->, thick] (p22) -- (p42) node[midway, below, xshift=2pt] {$\vec{p}_{(2)}'$};
  \draw[dotted,thick] (p12) -- (corner);
  \draw[dotted,thick] (p22) -- (corner);
  \draw[dotted,thick] (p12) -- (p22);

    \draw[fill=blue!30, opacity=0.6] (p12) circle (\R);
    \node at (p12) [above, xshift=-6pt] {$\vec{x}_{(1)}$};
    \draw[fill=blue!30, opacity=0.6] (p22) circle (\R);
    \node at (p22) [below, xshift=12pt] {$\vec{x}_{(2)}$};
    
  \pic [draw, angle radius=6pt] {right angle = p12--corner--p22};
\end{tikzpicture}
  \caption{An illustration of an elastic scattering between two hard sphere, distinguishable particles who initial momenta are at right angles to each other.}  
  \label{fig:scattering}  
\end{figure}
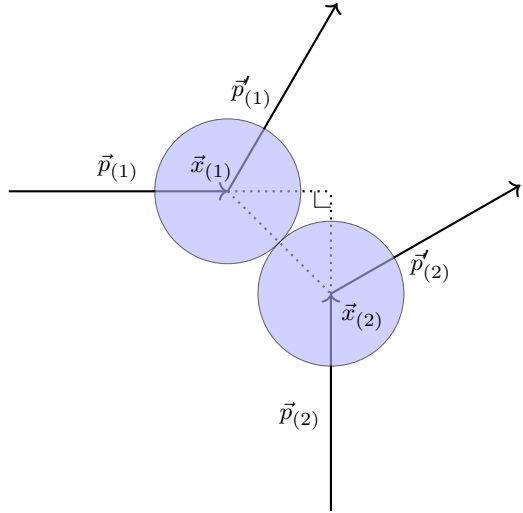

{\  Let $\vec{p}_{(k)}$ be the spatial momentum of the $k^{th}$ particle before the collision and $\vec{p}'_{(k)}$ the momentum of the same particle after the collision. Let us further define $\vec{J}_{(1)}$ and $\vec{J}_{(2)}$ such that
	\begin{eqnarray}
		\vec{p}_{(1)}' - \vec{p}_{(1)} = \vec{J}_{(1)} \; , \qquad \vec{p}_{(2)}' - \vec{p}_{(2)} = \vec{J}_{(2)} \; .
	\end{eqnarray} 
Conservation of spatial momentum in the scattering region (and indeed throughout the spacetime if it is flat everywhere) implies that $\vec{J}_{(1)} = - \vec{J}_{(2)} = \vec{J}$.  In a given scattering we are aware of the variables $\vec{p}_{(1)}$, $\vec{p}_{(2)}$, $\vec{x}_{(1)}$ and $\vec{x}_{(2)}$ as shown in fig.~\ref{fig:scattering} and it is our desire to determine $\vec{p}_{(1)}'$ and $\vec{p}_{(2)}'$. We shall further simplify our example taking the incoming particles to have equal magnitude momentum $\left\| \vec{p}_{(1)} \right\| = \left\| \vec{p}_{(2)} \right\| = p $ and to be approaching each other at right-angles $\vec{p}_{(1)} \cdot \vec{p}_{(2)} = 0$.  The vector $\vec{J}$ will point along the shortest distance between the cores of the two particles and we introduce the parameters $\lambda$ and $\vec{R}$ according to the relations
	\begin{eqnarray}
		\vec{J} =  \lambda p \hat{R} \; , \qquad \hat{R} = \frac{\vec{x}_{(1)} - \vec{x}_{(2)}}{\left\| \vec{x}_{(1)} - \vec{x}_{(2)} \right\|} \; , \qquad
		 \hat{p}_{(1)} \cdot \hat{R} = - \frac{1}{\sqrt{2}} \; , \qquad  \hat{p}_{(2)} \cdot \hat{R} = \frac{1}{\sqrt{2}} \; , \qquad
	\end{eqnarray}
Thus to determine $\vec{p}_{(1)}'$ and $\vec{p}_{(2)}'$ it is sufficient to solve for $\lambda$.}

{\ To compute $\lambda$ we need an additional piece of information, the dispersion relation. It is only necessary to consider the following dispersion relation
	\begin{eqnarray}
		\tilde{H}(\vec{p}^2) &=& (\vec{p}^2)^2 - \alpha \vec{p}^2 \; ,
	\end{eqnarray}
to make our point. The quantum analogue of the above expression is relevant for tri-layer graphene with an appropriate alignment between the layers  \cite{McCann_2013}. Conservation of the total energy (the total $p_{0}$) throughout the interaction gives the following constraint
	\begin{eqnarray}
		\tilde{H}(\vec{p}_{(1)}^2) + \tilde{H}(\vec{p}_{(2)}^2 )&=& \tilde{H}(\vec{p}_{(1)}'^2) + \tilde{H}(\vec{p}_{(2)}'^2) \; .
	\end{eqnarray}
Substituting in our expressions and rearranging we find for this particular Hamiltonian and initial parameters we must solve
	\begin{eqnarray}
		0 &=& \left( ( \lambda^2 - \sqrt{2} \lambda +1 )^2 -1 \right) p^2 - \alpha \left( \lambda^2 - \sqrt{2} \lambda \right) \; ,
	\end{eqnarray}
to determine the final momenta. As expected $\lambda = 0$ is a solution - the particles ignore each other - but there are three other solutions. These are
	\begin{eqnarray}
		\lambda = - \sqrt{2} \; , \qquad \lambda = - \frac{1}{\sqrt{2}} \pm \sqrt{\frac{3}{2}} \sqrt{\frac{2}{3} \frac{\alpha}{p^2} - 1} \; . 
	\end{eqnarray}
Consequently, when the momentum of the incoming particles satisfies $p^2 \leq \frac{2}{3} \alpha$ there are four channels through which an elastic collision may be realised. As we can always imagine a collision of sufficiently low momentum occurring in our gas, one can never ignore these additional channels. This should be compared to the usual Galilean case where there are only two ways for the particles to scatter - ignore each other or deflect along unique trajectories. Dealing with these additional channels, and understanding their consequence for the behaviour of the gas, is an interesting challenge we seek to address in future work. In what follows we need only assume that some choice is made to resolve such collisions.}

\subsection{Kinetic theory }

{\noindent Kinetic theory is the effective description of large numbers of particles in the dilute limit. In this section we first develop the collisionless kinetic theory of our exotic particles, before discussing the BBGKY hierarchy \cite{Cercignani1987-ij} and the consequences of collisions. In the next section we shall apply the formalism that we develop here to particles with a single dispersion relation and find in the collisionless limit ideal hydrodynamics.}

{\ To begin we need to define a notion of temperature; to do so, we must suppose that the Aristotelian spacetime has some time-like Killing vector $\beta^{\mu}$ satisfying the Killing conditions of \eqref{Eq:Killingconditions}. In general, the existence of such a time-like Killing vector allows us to construct the following scalar and vector quantities that are conserved under flows generated by the Killing field \cite{de_Boer_2020}:
    \begin{eqnarray}
    	\label{Eq:HydroTempandu}
        T := \frac{1}{\beta^{\mu} \tau_{\mu}} \; , \qquad u^{\mu} := T \beta^{\mu} \; .
    \end{eqnarray}
In the case of field theory, when one continues to imaginary time - if the correlation functions have a periodicity in the imaginary time direction - then one can identify $T$ with the inverse of that periodicity up to normalisation constants i.e.~the temperature. The above \eqref{Eq:HydroTempandu} hold independently of the curvature of the spacetime, but in the case of a flat Aristotelian manifold we can identify
	\begin{eqnarray}
		u^{\mu} = (1,\vec{v}) \; , \qquad \beta^{\mu} = \frac{1}{T} (1,\vec{v})  \; ,
	\end{eqnarray}
where $\vec{v}$ is the spatial velocity familiar to us from Galilean mechanics.}

{\ We begin with a one-particle distribution function $f(x^{\mu}(\lambda),p_{\mu}(\lambda))$ for a free particle which moves on a levelset of \eqref{eq:HSingleDispersion}. The distribution $f$ describes the likelihood of finding the single particle of interest in some region of our spacetime. It satisfies
	\begin{eqnarray}
			\frac{d}{d \lambda} f(x^{\mu}(\lambda),p_{\mu}(\lambda))
		&=& \frac{\partial p_{\mu}}{\partial \lambda} \frac{\partial f}{\partial p_{\mu}} +  \frac{\partial x^{\mu}}{\partial \lambda} \frac{\partial f}{\partial x^{\mu}} =  - \frac{\partial H_{*}}{\partial x^{\mu}}  \frac{\partial f}{\partial p_{\mu}} +  \frac{\partial H_{*}}{\partial p_{\mu}} \frac{\partial f}{\partial x^{\mu}} = 0 \; ,
	\end{eqnarray}
where we have employed the Hamiltonian equations in \eqref{Eq:HamiltonsEqns1} and \eqref{Eq:HamiltonsEqns2}; these clearly give the same expression as \eqref{Eq:ModifiedBoltzmannEqn}. The generalisation to a distribution $f_{N}$ describing multiple free (non-interacting) particles is well known. However, in the case that there are interactions between these particles, we expect the one-particle distribution function to depend on the particle world-line parameter and thus the above equation generalises to
	\begin{eqnarray}
			\frac{d}{d \lambda} f_{N}(x_{(i)}^{\mu}(\lambda),p^{(i)}_{\mu}(\lambda);\lambda)
		&=&  \left\{ H, f_{N}\right\} + \frac{\partial f_{N}}{\partial \lambda} = 0 \; ,
	\end{eqnarray}
where $(i)$ labels the particles and $H$ is the full Hamiltonian describing the $N$-particles and their interactions. Using standard arguments to construct the BBGKY hierarchy \cite{yvon,bogoliubov1946kinetic,10.1063/1.1724117,doi:10.1098/rspa.1946.0093,10.1063/1.1746292} around $H_{*}$ the equation of motion satisfied by the one-particle distribution function is modified to
	\begin{eqnarray}
			\label{Eq:BoltzmannEquation}
			\frac{d}{d \lambda} f(x^{\mu}(\lambda),p_{\mu}(\lambda);\lambda)
		&=&  \left\{ H_{*}, f\right\} + \left(\frac{\partial f}{\partial \lambda}\right)_{\mathrm{collisions}} = 0 \; ,
	\end{eqnarray}
where the final term accounts for the effect of collisions and $f$ is once again the one-particle distribution, only now it is modified by the presence of the interactions. This is the generalisation of \eqref{Eq:ModifiedBoltzmannEqn} to include collisions.}

{\ Assuming molecular chaos (that the particle velocities are uncorrelated during a collision) and that any interactions preserve time reversal, spatial parity and translation invariance in flat space leads to a collision integral of the form
	\begin{eqnarray}
		&\;& C_{w}[f,f;x^{\mu},p_{\mu}] \nonumber \\
		&=& \int\mathrm{d}^{d} \vec{p}_{(1)}' d^{d}\vec{p}_{(2)}' d^{d}\vec{p}_{(1)} \; w\left(\vec{p},\vec{p}_{(2)},p_{(1)}',\vec{p}_{(2)}'\right) \left[ f(\vec{p}_{(1)}',\vec{r}) f(\vec{p}_{(2)}',\vec{r}) - f(\vec{p},\vec{r}) f(\vec{p}_{(2)},\vec{r}) \right] \; , \qquad
	\end{eqnarray}
where $w$ is the collision kernel parameterising the interaction i.e.
	\begin{eqnarray}
		\left\{ H_{*},f \right\} &=& -  C_{w}[f,f;x^{\mu},p_{\mu}] \; , 
	\end{eqnarray}
independently of the particle worldline. If we want this to vanish without recourse to fixing the precise form of the interaction we can apply the principle of detailed balance which imposes that
	\begin{eqnarray}
		f(\vec{p}_{(1)}',\vec{r}) f(\vec{p}_{(2)}',\vec{r}) = f(\vec{p},\vec{r}) f(\vec{p}_{(2)},\vec{r}) \; ,
	\end{eqnarray}
or more usefully
	\begin{eqnarray}
		\label{Eq:DetailedBalance}
		\ln f(\vec{p}_{(1)}',\vec{r}) + \ln f(\vec{p}_{(2)}',\vec{r}) = \ln f(\vec{p},\vec{r}) + \ln f(\vec{p}_{(2)},\vec{r}) \; . 
	\end{eqnarray}
Further, a way to satisfy this trivially for all possible momenta is to assume that the distribution function is constructed from invariants such as energy and spatial momentum. In particular we can identify the standard one-particle equilibrium distribution function to be
	\begin{eqnarray}
		\label{Eq:Standard1Distribution}
		f_{s} = \left. \kappa e^{-\frac{p_{0} - \vec{p} \cdot \vec{v}}{T}} \right|_{p_{0} = \tilde{H}(\vec{p}^2)} = \kappa e^{-\frac{\tilde{H}(\vec{p}^2) - \vec{p} \cdot \vec{v}}{T}} \; , 
	\end{eqnarray}
where $\kappa$ is a normalisation constant, which clearly solves \eqref{Eq:DetailedBalance} upon using conservation of total energy $p_{0}$ and spatial momentum across the interaction. We shall then call \eqref{Eq:Standard1Distribution}, when evaluated on any level set of the Hamiltonian, the standard local thermodynamic equilibrium distribution for that Hamiltonian.}

{\ An important result in kinetic theory is that for reasonable collision terms, the divergence of a class of currents - identified with entropy - is positive definite (as follows from the H-theorem \cite{Cercignani1987-ij}). Given a system satisfying \eqref{Eq:p0sub} and described by a generic one-particle distribution function $f$ we define the associated entropy current by
	\begin{eqnarray}
		\label{Eq:EntropyDefinition}
		J_{s}^{\mu} \partial_{\mu} = \left[ \int\mathrm{d}^{d}\vec{p} \; f \ln \left( A f \right) \right] \partial_{t} + \left[ 2 h^{ij} \int\mathrm{d}^{d}\vec{p} \; \frac{\partial \tilde{H}}{\partial \vec{p}^2} p_{j} f \ln \left( A f \right)  \right] \partial_{i} \; ,	
	\end{eqnarray}
where $J_{s}^{\mu}$, defined on an arbitrary Aristotelian manifold, and $A$ is a normalisation term present to make the argument of the logarithm dimensionless\footnote{The appearance of the constant $A$ is related to the fact that we can shift $S^{\mu} \rightarrow \tilde{S}^{\mu} = S^{\mu} + c J^{\mu}$ where $\tilde{S}^{\mu}$ is also a conserved current.}. The time component is familiar from kinetic theory and is the usual expression for the entropy density, the spatial part requires some results from section \ref{sec:Hamiltonian} to derive. Regardless, this entropy current is positive for reasonable collision kernels which has important consequences for hydrodynamics. In particular, positivity of this quantity can be used to constrain transport coefficients beyond leading order in a small derivative expansion.}

\subsection{Ideal hydrodynamics of particles with a single dispersion relation }\label{sec:Anharmonic}

{\noindent Luckily, flat space is the appropriate limit to consider ideal hydrodynamics; this follows from the fact that torsion is order one in derivatives and curvature is order two while the ideal theory sits at zeroth order. For the single dispersion relation Hamiltonians of \eqref{eq:HSingleDispersion}, the standard one-particle distribution function takes the form
	\begin{eqnarray}
    	\label{Eq:OneParticleDistr}
        f_{s} &=& \kappa \exp \left( - \frac{1}{T} \left( \tilde{H}(\vec{p}^2) - \vec{v} \cdot \vec{p} \right) \right) \; , 
          \end{eqnarray}
where $\kappa$ is some normalisation constant which we shall fix shortly. We would now like to show that our formalism reproduces the ideal hydrodynamic charge current and SEM tensor complex.}

{\ Let's us briefly review what is expected for a boost-agnostic hydrodynamic theory: we suppose that there exists a hydrostatic generating functional \cite{de_Boer_2020} defined on a weakly curved Aristotelian manifold in the presence of an external gauge field. The most general generating functional on such a manifold that one can write down at leading (zero) order in derivatives that respects the symmetries is
	\begin{subequations}
	\label{Eq:HydroGeneratingFunctional}
	\begin{eqnarray}
		W &=& \int\mathrm{d}^{d+1} x \; e P(T,\mu,u^{\mu} h_{\mu \nu} u^{\nu}) \; , \\
		T &=& \frac{1}{\beta^{\mu} \tau_{\mu}} \; , \qquad u^{\mu} = T \beta^{\mu} \; ,  \qquad \mu = T \left( A_{\mu} \beta^{\mu} + \Lambda  \right) \; ,  
	\end{eqnarray}
	\end{subequations}
where $P$ is at the moment an arbitrary function of the effective hydrodynamic fields that will eventually turn out to be the pressure, and $\Lambda$ is a gauge parameter present to ensure that $\mu$ is gauge invariant. The variation of the structure invariants and the measure $e$ in terms of the variation of the inverse structures invariants are:
	\begin{subequations}
		\begin{eqnarray}
			\delta \tau_{\mu}
	&=& - h_{\mu (\sigma} \tau_{\nu)} \delta \tilde{h}^{\sigma \nu}  + \tau_{\mu} \tau_{\nu} \delta \nu^{\nu} \; ,  \\
		\delta h_{\mu \nu}
	&=& - h_{\mu \rho} \delta \tilde{h}^{\rho \sigma} h_{\sigma \nu} + \left( \tau_{\mu} h_{\nu \rho} + \tau_{\nu} h_{\mu \rho} \right) \delta \nu^{\rho}  \; ,   \qquad \\
		\delta e
	&=& e \left( \tau_{\mu} \delta \nu^{\mu} - \frac{1}{2} h_{\mu \nu} \delta \tilde{h}^{\mu \nu} \right) \; . 
		\end{eqnarray}
	\end{subequations}
Hence the variation of the temperature and velocity under variations that preserve the Killing nature of the vector field $\beta^{\mu}$ take the form
	\begin{subequations}
\begin{eqnarray}
    \delta T
    &=& - T \left( \tau_{\nu} \delta \nu^{\nu} - u^{\sigma} h_{\sigma (\mu} \tau_{\nu)} \delta \tilde{h}^{\mu \nu} \right) \; ,  \\
    	 	\delta \mu
	&=& - \mu \left( \tau_{\nu} \delta \nu^{\nu} - u^{\sigma} h_{\sigma (\mu} \tau_{\nu)} \delta \tilde{h}^{\mu \nu} \right)
		+ u^{\mu} \delta A_{\mu}  \; ,  \\
    	   \delta u^\mu
    &=& - u^{\mu} \left( \tau_{\nu} \delta \nu^{\nu} - u^{\sigma} h_{\sigma (\alpha} \tau_{\beta)} \delta \tilde{h}^{\alpha \beta} \right)  \; . 
\end{eqnarray}
\end{subequations}
The terms $\tilde{T}_{\mu}$ and $\tilde{T}_{\mu \nu}$, which combine to make the SEM tensor complex of \eqref{Eq:defSEMtensorcomplex}, are then defined by
	\begin{eqnarray}
			\delta W
		  &=& \int\mathrm{d}^{d+1} x \;  e \left[ \tilde{T}_{\mu} \delta \nu^{\mu} - \frac{1}{2} \tilde{T}_{\mu \nu} \delta \tilde{h}^{\mu \nu} + J^{\mu} \delta A_{\mu} \right]
	\end{eqnarray}
where, in keeping with the considerations that led to \eqref{Eq:defSEMtensorcomplex}, we have used the inverse structure invariants. This expression should be compared to that of \cite{de_Boer_2020} where the SEM tensor complex is defined with respect to the structure invariants. Subsequently, one finds the following hydrodynamic constitutive relations
	\begin{subequations}
	\label{Eq:HydrodynamicConstitutiveRelations}
	\begin{eqnarray}
			J^{\mu}
		&=& \frac{\partial P}{\partial \mu} u^{\mu} \; , \\
			\tilde{T}_{\mu}
		&=& \left( P - T \frac{\partial P}{\partial T} - \mu \frac{\partial P}{\partial \mu} - 2 u^2 \frac{\partial P}{\partial u^2} \right) \tau_{\mu} + 2 \frac{\partial P}{\partial u^2} h_{\mu \rho} u^{\rho} \; , \\
			\tilde{T}_{\mu \nu}
		&=& P h_{\mu \nu} - 2 \left( T \frac{\partial P}{\partial T} + \mu \frac{\partial P}{\partial \mu} + 2 u^2 \frac{\partial P}{\partial u^2} \right) u^{\sigma} h_{\sigma (\mu} \tau_{\nu)} + 2 \frac{\partial P}{\partial u^2} u^{\rho} h_{\rho \mu} u^{\sigma} h_{\sigma \nu} \; .		\end{eqnarray}
The SEM tensor complex, constructed from $\tilde{T}_{\mu}$ and $\tilde{T}_{\mu \nu}$ in \eqref{Eq:defSEMtensorcomplex}, takes the form
	\begin{eqnarray}
			T\indices{^\mu_\nu} 
		&=& - \nu^{\mu} \left[ \left( P  - T \partial_{T} P - \mu \partial_{\mu} P - 2 u^2 \partial_{u^2} P  \right) \tau_{\nu} + 2 \partial_{u^2} P u^{\alpha} h_{\alpha \nu} \right] \nonumber \\
		&\;& + \tilde{h}^{\mu \rho} \left[ P h_{\rho \nu} - \left( T \partial_{T} P + \mu \partial_{\mu} P + 2 u^2 \partial_{u^2} P \right) h_{\rho \sigma} u^{\sigma} \tau_{\nu} \right. \nonumber \\
		&\;& \left. \hphantom{+ \tilde{h}^{\mu \rho} \left[ P h_{\rho \nu} \right.} + 2 \partial_{u^2} P h_{\rho \sigma} u^{\sigma} h_{\nu \lambda} u^{\lambda} \right] \; . \qquad
	\end{eqnarray}
	\end{subequations}
Further, using the identifications \eqref{Eq:HydroTempandu}, so that $u^{\mu} \tau_{\mu}=1$ and
	\begin{eqnarray}
		u^{\mu} = - \nu^{\mu} + h^{\mu \nu} h_{\nu \rho} u^{\rho} \; , 
	\end{eqnarray}
we can simplify the SEM tensor complex to
	\begin{eqnarray}
			T\indices{^\mu_\nu} 
		&=& - \left( - P  + T \partial_{T} P + \mu \partial_{\mu} P + 2 u^2 \partial_{u^2} P  \right) u^{\mu} \tau_{\nu} 
			- P h^{\mu \rho} h_{\rho \sigma} u^{\sigma} \tau_{\nu} + P h^{\mu \rho} h_{\rho \nu} \nonumber \\
		&\;& + 2  \partial_{u^2} P h_{\rho \sigma} u^{\sigma} \left( u^{\mu} + \nu^{\mu} \right) \; .
	\end{eqnarray}
The relationship between the coefficients of the components of the SEM tensor complex are a consequence of the existence of a generating functional and one can proceed order by order in derivatives, correcting \eqref{Eq:HydroGeneratingFunctional} and generating the hydrostatic part of the constitutive relations.}

{\  On a flat Aristotelian manifold we would want to make the following identifications
	\begin{subequations}
	\label{Eq:ConstitutiveRelationsPlusThermo}
	\begin{align}
		\label{Eq:ConstitutiveRelations1}
		& J^{0} = n \; , \qquad \; \ T\indices{^0_0} = - \epsilon \; , \qquad T\indices{^0_i} = \rho v_{i} = P_{i} \; , \qquad \\
		\label{Eq:ConstitutiveRelations2}
		& J^{i} = n v^{i}  \; , \qquad T\indices{^i_0} = - \left( \epsilon + P \right) v^{i} = - J^{i}_{\epsilon} \; , \qquad  \\
		\label{Eq:ConstitutiveRelations3}
		& T\indices{^i_j} = P \Pi\indices{^i_j} + \left( P + \rho \vec{v}^2 \right) \hat{v}^{i} \hat{v}_{j} \; , \qquad \Pi^{ij} = \delta^{ij} - \hat{v}^{i} \hat{v}^{j} \; .	
	\end{align}
where $n$ is defined in terms of the charge current $J^{\mu}$ by \eqref{Eq:GenericChargeDensityDef}. Notice in particular that the charge density $n$ and the kinetic mass density $\rho$ are two independent thermodynamic functions, contrary to what happens for (Galilean boost invariant) non-relativistic fluids. The energy density with respect to $u^{\mu}$ is defined in \eqref{Eq:GenericEnergyDef} and corresponds to
	\begin{eqnarray}
		\epsilon_{u} &=& - \left( T\indices{^0_0} + T\indices{^{0}_{i}} v^{i} \right) = \epsilon - \rho \vec{v}^2 \; .
	\end{eqnarray}
In the case of flat space we have as many Killing vectors as spatial directions, i.e. $\xi^{\mu}_{i=1,\ldots,d}$, which then allows us to define the spatial momentum
	\begin{eqnarray}
		P_{i} =  \tau_{\mu} T\indices{^\mu_\nu} \xi^{\nu}_{i} \; , 
	\end{eqnarray}
Consequently, the following quantity,
	\begin{eqnarray}
		\epsilon = \epsilon_{u} + \vec{P} \cdot \vec{v} \; , 
	\end{eqnarray}
which is identified with the internal energy is conserved. By comparing these expressions with the constitutive relations following from the generating functional we see that $n$, $\epsilon$ and $\rho$ are given in terms of the pressure $P$ and its derivatives according to the thermodynamic relations:
	\begin{eqnarray}
		\label{Eq:ThermoRelations}
		\epsilon =  - P  + T \partial_{T} P + \mu \partial_{\mu} P + 2 \vec{v}^2 \partial_{\vec{v}^2} P \; , \qquad  n = \partial_{\mu} P \; , \qquad \rho = 2 \frac{\partial P}{\partial \vec{v}^2} \; .
	\end{eqnarray}
		\end{subequations}
Having developed the hydrodynamic constitutive relations above, our current task is to show that, for the standard one-particle distribution function \eqref{Eq:OneParticleDistr} using the generic expressions for the charge current \eqref{Eq:SingleDisperionCurrentSEMIntegralSimpleI} and SEM tensor complex in \eqref{Eq:SingleDisperionCurrentSEMIntegralSimpleII}, the kinetic theory formalism that we have developed reproduces the expressions in \eqref{Eq:ConstitutiveRelationsPlusThermo}. This makes our kinetic theory analogous to the relativistic and Galilean cases, only absent of boost invariance.}

{\ In particular, using \eqref{Eq:SingleDisperionCurrentSEMIntegralSimpleI} and \eqref{Eq:SingleDisperionCurrentSEMIntegralSimpleII}, and comparing with \eqref{Eq:ConstitutiveRelations1} we can see that the charge densities in terms of the given one-particle distribution function have the form 	\begin{subequations}
	\label{Eq:KineticConservedCharges}
	\begin{eqnarray}
			n 
		&=& \sum_{m=0}^{\infty} \frac{1}{m!} \left(\frac{v}{T}\right)^{m} \left( \int_{p=0}^{\infty}\mathrm{d}p \; p^{d+m-1} e^{-\frac{\tilde{H}}{T}} \right) \int\mathrm{d}\Omega \; \cos^{m}(\theta) \; , \\
			\epsilon
		&=& \sum_{m=0}^{\infty} \frac{1}{m!} \left(\frac{v}{T}\right)^{m} \left( \int_{p=0}^{\infty}\mathrm{d}p \; p^{d+m-1} \tilde{H} e^{-\frac{\tilde{H}}{T}} \right) \int\mathrm{d}\Omega \;  \cos^{m}(\theta) \; , \\
			\rho
		&=& \frac{v^{j}}{\vec{v}^2} \left[ \sum_{m=0}^{\infty} \frac{1}{m!} \left(\frac{v}{T}\right)^{m} \left( \int_{p=0}^{\infty}\mathrm{d}p \; p^{d+m} e^{-\frac{\tilde{H}}{T}}\right) \int\mathrm{d}\Omega \;  \hat{p}_{j} \cos^{m}(\theta) \right] \; ,
	\end{eqnarray}
	\end{subequations}
where $\theta$ is the angle between $\vec{p}$ and $\vec{v}$, and we have introduced the notation $p= \left\| \vec{p} \right\|$ and $v= \left\| \vec{v} \right\|$. In deriving these expressions we have assumed that the summation converges appropriately once we expand $\exp( \vec{p} \cdot \vec{v} / T)$ as a power series i.e.
   \begin{eqnarray}
        f_{s}=  \kappa \sum_{m=0}^{\infty} \frac{1}{m!} \left( \frac{p v}{T} \cos(\theta)  \right)^{m}   \exp \left( - \frac{\tilde{H}(p^2)}{T} \right)   \; . 
    \end{eqnarray}
Corresponding to the expressions for the conserved charges are the respective spatial currents \eqref{Eq:ConstitutiveRelations2} and \eqref{Eq:ConstitutiveRelations3} given by
	\begin{subequations}
	\label{Eq:KineticConservedCurrents}
	\begin{eqnarray}
			J^{i} 
		&=& 2 \hat{v}^{i} \sum_{m=0}^{\infty} \frac{1}{m!}  \left(\frac{v}{T}\right)^{m} \left( \int_{p=0}^{\infty}\mathrm{d}p \; p^{d+m} \frac{\partial \tilde{H}}{\partial p^2} e^{-\frac{\tilde{H}}{T}} \right) \int\mathrm{d}\Omega \; \cos^{m+1}(\theta) \; , \\
		        J^{i}_{\epsilon} 
		&=& 2 \hat{v}^{i} \sum_{m=0}^{\infty} \frac{1}{m!} \left(\frac{v}{T}\right)^{m} \left( \int_{p=0}^{\infty}\mathrm{d}p \; p^{d+m} \tilde{H} \frac{\partial \tilde{H}}{\partial p^2} e^{-\frac{\tilde{H}}{T}} \right) \int\mathrm{d}\Omega \; \cos^{m+1}(\theta)  \; ,  \\
		        T\indices{^{i}_{j}}
		&=& 2 \hat{v}^{i} \sum_{m=0}^{\infty} \frac{1}{m!} \left(\frac{v}{T}\right)^{m} \left( \int_{p=0}^{\infty}\mathrm{d}p \; p^{d+m+1} \frac{\partial \tilde{H}}{\partial p^2}  e^{-\frac{\tilde{H}}{T}} \right) \int\mathrm{d}\Omega \; \hat{p}_{j} \cos^{m+1}(\theta) \nonumber \\
		&+& 2 \sum_{m=0}^{\infty} \frac{1}{m!} \left(\frac{v}{T}\right)^{m} \left( \int_{p=0}^{\infty}\mathrm{d}p \; p^{d+m+1} \frac{\partial \tilde{H}}{\partial p^2} e^{-\frac{\tilde{H}}{T}} \right) \int\mathrm{d}\Omega \; \hat{p}_{i}^{\perp} \hat{p}_{j} \cos^{m}(\theta) \sin(\theta) \; , \qquad \;
	\end{eqnarray}
	\end{subequations}
where we have already employed certain results from appendix \ref{appendix:integrals} to simplify some of the expressions. We can further decompose the $j$ index of $T\indices{^0_j}$ and $T\indices{^i_j}$ parallel and perpendicular to the velocity parameter $\vec{v}$ to obtain
	\begin{subequations}
	\label{Eq:KineticConservedSimplified}
	\begin{eqnarray}
			\rho 
		&=& \sum_{m=0}^{\infty} \frac{1}{m! T} \left(\frac{v}{T}\right)^{m-1} \left( \int_{p=0}^{\infty}\mathrm{d}p \; p^{d+m} e^{-\frac{\tilde{H}}{T}}\right) \int\mathrm{d}\Omega \;  \cos^{m+1}(\theta) \; , \\
			T\indices{^{i}_{j}}
		&=& 2 \hat{v}^{i} \hat{v}_{j} \sum_{m=0}^{\infty} \frac{1}{m!} \left(\frac{v}{T}\right)^{m} \left( \int_{p=0}^{\infty}\mathrm{d}p \; p^{d+m+1} \frac{\partial \tilde{H}}{\partial p^2}  e^{-\frac{\tilde{H}}{T}} \right) \int\mathrm{d}\Omega \; \cos^{m+2}(\theta) \\
		&+& 2 \sum_{m=0}^{\infty} \frac{1}{m!}  \left(\frac{v}{T}\right)^{m} \left( \int_{p=0}^{\infty}\mathrm{d}p \; p^{d+m+1} \frac{\partial \tilde{H}}{\partial p^2} e^{-\frac{\tilde{H}}{T}} \right) \int\mathrm{d}\Omega \; \hat{p}_{i}^{\perp} \hat{p}_{j}^{\perp} \cos^{m}(\theta) \sin^2(\theta) \; , \nonumber
	\end{eqnarray}
	\end{subequations}
after again employing identities from appendix \ref{appendix:integrals}.}

{\ To evaluate the relevant integrals in \eqref{Eq:KineticConservedCharges}, \eqref{Eq:KineticConservedCurrents} and \eqref{Eq:KineticConservedSimplified} in a quasi-compact form, let us define
	\begin{eqnarray}
		c_{m} &=&  \frac{ \Gamma\left(m+\frac{1}{2}\right)}{(2m)! \Gamma\left(m+\frac{d}{2}\right)} \; , \qquad
			\frac{c_{m+1}}{c_{m}} =  \frac{1}{2 (m+1) (2m+ d)} \; . 
	\end{eqnarray}
Computing the angular integral in the number density and its associated spatial charge current using results in appendix \ref{appendix:integrals} we find
	\begin{subequations}
	\begin{eqnarray}
			\label{Eq:GenericChargeDensity}
			n
		&=& 2 \pi^{\frac{d-1}{2}}  \sum_{m=0}^{\infty} c_{m} \left(\frac{v}{T}\right)^{2m} \left( \int_{p=0}^{\infty}\mathrm{d}p \; p^{d+2m-1} e^{-\frac{\tilde{H}}{T}} \right) \; , \qquad \;  \\
			J^{i} 
		&=& \left[ 2 \pi^{\frac{d-1}{2}} \sum_{m=0}^{\infty} c_{m} \left(\frac{v}{T}\right)^{2m} \left(  \left[ - \frac{p^{d+m}}{d+m}  e^{-\frac{\tilde{H}}{T}} \right]_{p=0}^{\infty} + \int_{p=0}^{\infty}\mathrm{d}p \; p^{d+2m-1} e^{-\frac{\tilde{H}}{T}} \right) \right]  v^{i}  \nonumber \\
		&=& n v^{i} \; ,
	\end{eqnarray}
	\end{subequations}
where we have integrated by parts the momentum integral and assumed
	\begin{eqnarray}
		\label{Eq:Decay}
		 \left[ p^{d+m}  e^{-\frac{\tilde{H}}{T}} \right]_{p=0}^{\infty} = 0 \; , 
	\end{eqnarray}
which holds for any reasonable choice of $\tilde{H}$ with $m \geq 0$, $d > 0$. Consequently, when the one-particle distribution function takes the standard local form \eqref{Eq:OneParticleDistr}, independently of the functional form of $\tilde{H}(p^2)$, one necessarily finds that the particle number current has the form of the ideal hydrodynamic constitutive relation i.e. $J^{\mu} = n u^{\mu}$.}

{\ Turning to the SEM tensor complex, there are five distinct tensor structures in \eqref{Eq:ConstitutiveRelations2} and \eqref{Eq:ConstitutiveRelations3}, and only three thermodynamic scalars that appear: $\epsilon$, $P$ and $\rho$. Of these scalars, the energy density and mass density have the form:
	\begin{subequations}
	\begin{eqnarray}
			\label{Eq:InternalEnergyDensity}
			\epsilon 
		&=& 2 \pi^{\frac{d-1}{2}}  \sum_{m=0}^{\infty}c_{m}  \left(\frac{v}{T}\right)^{2m} \left( \int_{p=0}^{\infty}\mathrm{d}p \; p^{d+2m-1} \tilde{H} e^{-\frac{\tilde{H}}{T}} \right)  \; , \qquad \; \; \\
			\rho 
		&=& \frac{2 \pi^{\frac{d-1}{2}}}{T} \sum_{m=0}^{\infty}  c_{m} \left(\frac{v}{T}\right)^{2m} \left( \int_{p=0}^{\infty}\mathrm{d}p \; \frac{p^{d+2m+1}}{d+2m} e^{-\frac{\tilde{H}}{T}}\right)  \; .
	\end{eqnarray}
	\end{subequations}
Given the expression for the number density \eqref{Eq:GenericChargeDensity}, we see that if $\tilde{H}$ has the Galilean form $\tilde{H}=\vec{p}^2$ then 
	\begin{eqnarray}
			\int_{p=0}^{\infty}\mathrm{d}p \; p^{d+2m-1} e^{-\frac{p^2}{2 T}} 
		&=& \left[ \frac{p^{d+2m}}{d+2m} e^{-\frac{p^2}{2 T}}  \right]_{p=0}^{\infty}
			+ \frac{1}{(d+2m) T} \int_{p=0}^{\infty}\mathrm{d}p \; p^{d+2m+1} e^{-\frac{p^2}{2 T}} \; , \qquad \; 
	\end{eqnarray}
i.e $\rho=n$ as expected in a Galilean theory when the boundary term vanishes. }

{\ The third as yet unknown thermodynamic parameter appearing in \eqref{Eq:ConstitutiveRelations2} and \eqref{Eq:ConstitutiveRelations3}, the pressure, can be determined from the transverse (to the spatial velocity $\vec{v}$) part of the spatial stress $T\indices{^i_j}$ i.e.
	\begin{eqnarray}
		\label{Eq:IdealGasLaw}
		\frac{P}{T}
	&=& 2 \pi^{\frac{d-1}{2}} \sum_{m=0}^{\infty} c_{m} \left(\frac{v}{T}\right)^{2m}  \int_{p=0}^{\infty}\mathrm{d}p \; p^{d+2m-1} e^{-\frac{\tilde{H}}{T}} = n \; ,
	\end{eqnarray}
where we have integrated the momentum term by parts and assumed \eqref{Eq:Decay} holds. This is nothing more than the ideal gas law in a suitable choice of units and we have thus established one of our claimed results: independently of the form of $\tilde{H}(p^2)$, a gas of free particles fulfils the ideal gas law. This independence of the ideal gas law from dispersion should be compared to the effect of turning on interaction potential which generically leads to departures from $PV=NT$.}

{\ To complete our objective of demonstrating the currents take the form of ideal hydrodynamics, it remains only to check that the thermodynamic relations \eqref{Eq:ThermoRelations} are satisfied and that the totally parallel parts of $T\indices{^i_j}$ and the energy current $J^{i}_{\epsilon}$ have the desired form given in \eqref{Eq:ConstitutiveRelations2} and \eqref{Eq:ConstitutiveRelations3}. Let's achieve the second aim first; the energy current evaluates to
	\begin{eqnarray}
			J^{i}_{\epsilon}
		&=& \left[ 2 \pi^{\frac{d-1}{2}} \sum_{m=0}^{\infty} c_{m}  \left(\frac{v}{T}\right)^{2m} \int_{p=0}^{\infty}\mathrm{d}p \; p^{d+2m-1} \left( \tilde{H} + T \right) e^{-\frac{\tilde{H}}{T}}   \right] v^{i} \qquad \nonumber \\
		&=& \left( \epsilon + P \right) v^{i} \; , 
	\end{eqnarray}
where we have used
	\begin{eqnarray}
			\int_{p=0}^{\infty}\mathrm{d}p \;  p^{d+m} \tilde{H} \frac{\partial}{\partial p} e^{-\frac{\tilde{H}}{T}} 
		&=& - (d+m) \int_{p=0}^{\infty}\mathrm{d}p \; p^{d+m-1} \left( \tilde{H} + T \right) e^{-\frac{\tilde{H}}{T}}
	\end{eqnarray}
assuming again suitable behaviour at $p=0$ and $p \rightarrow \infty$. Meanwhile the totally parallel part of the stress tensor is given by the integral
	\begin{eqnarray}
			\label{Eq:Stresscontraction}
			\frac{v_{i} T\indices{^{i}_{j}} v^{j}}{v^2}
		&=&  2 \pi^{\frac{d-1}{2}} \left[  c_{0} T  \int_{p=0}^{\infty}\mathrm{d}p \; p^{d-1}  e^{-\frac{\tilde{H}}{T}} \right. \nonumber \\
		&\;& \left. \hphantom{2 \pi^{\frac{d-1}{2}} \left[  \right.} + T \sum_{m=1}^{\infty} (2 m +1) c_{m}  \left(\frac{v}{T}\right)^{2m} \int_{p=0}^{\infty}\mathrm{d}p \; p^{d+2m-1}  e^{-\frac{\tilde{H}}{T}} \right] \nonumber \\
		&=& P + \rho v^2 \; , 
	\end{eqnarray}
where we have used
	\begin{eqnarray}
			\rho v^{2}
		&=& 2 T \pi^{\frac{d-1}{2}} \sum_{m=1}^{\infty}  2 m c_{m} \left(\frac{v}{T}\right)^{2m} \left( \int_{p=0}^{\infty}\mathrm{d}p \; p^{d+2m-1} e^{-\frac{\tilde{H}}{T}}\right)  \; . 
	\end{eqnarray}
Thus these terms are as expected from ideal boost-agnostic hydrodynamics which completes our check for the form of the currents.}

{\ Finally, we check that the relevant thermodynamic relations \eqref{Eq:ThermoRelations} hold. To do this we contract \eqref{Eq:EntropyDefinition} with the clock form to find
	\begin{eqnarray}
		\label{Eq:IntegratedEntropyDensity}
		s = \tau_{\mu} J^{\mu}_{s}
		   &=& \frac{\epsilon - \rho v^2 - n T (\ln \kappa - \ln A)}{T} \; , 
	\end{eqnarray}
where we have employed the definition of the charges \eqref{Eq:KineticConservedCharges}. Rearranging we see that
	\begin{eqnarray}
		\epsilon + P - \rho v^2 - s T &=& n T \left( \ln \kappa - \ln A + 1 \right)
	\end{eqnarray}
and used the ideal gas law \eqref{Eq:IdealGasLaw}. This expression is almost the integrated form of the first law, we need only identify
	\begin{eqnarray}
		\frac{\mu}{T} &=& \ln \left(\frac{\kappa}{A}\right) + 1 \; ,
	\end{eqnarray}
from which it then follows that
	\begin{eqnarray}
		\epsilon &=& - P + s T + \mu n + \rho v^2 \; . 
	\end{eqnarray}
Without loss of generality, we choose $\kappa = e^{\frac{\mu}{T}}$, with $\mu$ then chemical potential, and $\ln A=1$ so that
	\begin{subequations}
	\begin{eqnarray}
		f_{s} &=&  \exp \left( - \frac{1}{T} \left( \tilde{H}(\vec{p}^2) - \mu - \vec{v} \cdot \vec{p} \right) \right) \; , \\
		J^{\mu}_{s} &=&  \left[ \int\mathrm{d}^{d}\vec{p} \; f_{s} \left( \ln f_{s} + 1 \right) \right] \partial_{t} + \left[ 2 h^{ij} \int\mathrm{d}^{d}\vec{p} \; \frac{\partial \tilde{H}}{\partial \vec{p}^2} p_{j} f_{s} \left( \ln f_{s} + 1 \right)  \right] \partial_{i} \; .	
	\end{eqnarray}
	\end{subequations}
Consequently, having identified the chemical potential we construct the grand canonical potential, $\Omega = - P(\mu,T,v^2) V$, and find the following thermodynamic identities are satisfied
	\begin{subequations}
	\begin{eqnarray}
			\left(\frac{\partial \Omega}{\partial \mu} \right)_{T,\vec{v},V}
		&=& - 2 \mu \pi^{\frac{d-1}{2}} V e^{\frac{\mu}{T}} \sum_{m=0}^{\infty} c_{m} \left(\frac{v}{T}\right)^{2m}  \int_{p=0}^{\infty}\mathrm{d}p \; p^{d+2m-1} e^{-\frac{\tilde{H}}{T}} = - n V \; , \qquad \; \; \\
			\left(\frac{\partial \Omega}{\partial v^2} \right)_{T,\mu,V}
		&=& - 2 \pi^{\frac{d-1}{2}} T V e^{\frac{\mu}{T}} \sum_{m=1}^{\infty} \frac{m}{T^2} c_{m} \left(\frac{v}{T}\right)^{2(m-1)}  \int_{p=0}^{\infty}\mathrm{d}p \; p^{d+2m-1} e^{-\frac{\tilde{H}}{T}} \nonumber \\
		&=& - \frac{\rho}{2} V \; , \\
			\left(\frac{\partial \Omega}{\partial T} \right)_{T,\vec{v},V}
		&=&  - 2 \pi^{\frac{d-1}{2}} T V e^{\frac{\mu}{T}} \nonumber \\
		&\;& \sum_{m=0}^{\infty}  c_{m} \left(\frac{v}{T}\right)^{2m}  \int_{p=0}^{\infty}\mathrm{d}p \; p^{d+2m-1} \left[ \frac{\tilde{H} - \mu - ( 2 m - 1 ) T}{T^2} \right] e^{-\frac{\tilde{H}}{T}} \nonumber \\
		&=& - \left( \frac{\epsilon + P - \mu n - \rho v^2}{T} \right) V = - s V \; ,
	\end{eqnarray}
	\end{subequations}
which are indeed the relations given in \eqref{Eq:ThermoRelations} and justifies our identification of the chemical potential. This completes our demonstration that our formalism reproduces ideal hydrodynamics independently of the form of $\tilde{H}(\vec{p}^2)$, and demonstrates the utility of our formalism to derive useful results i.e. the ubiquity of the ideal gas law and hydrodynamics.}

\begin{figure}
	\centering
	\includegraphics[width=0.99\linewidth]{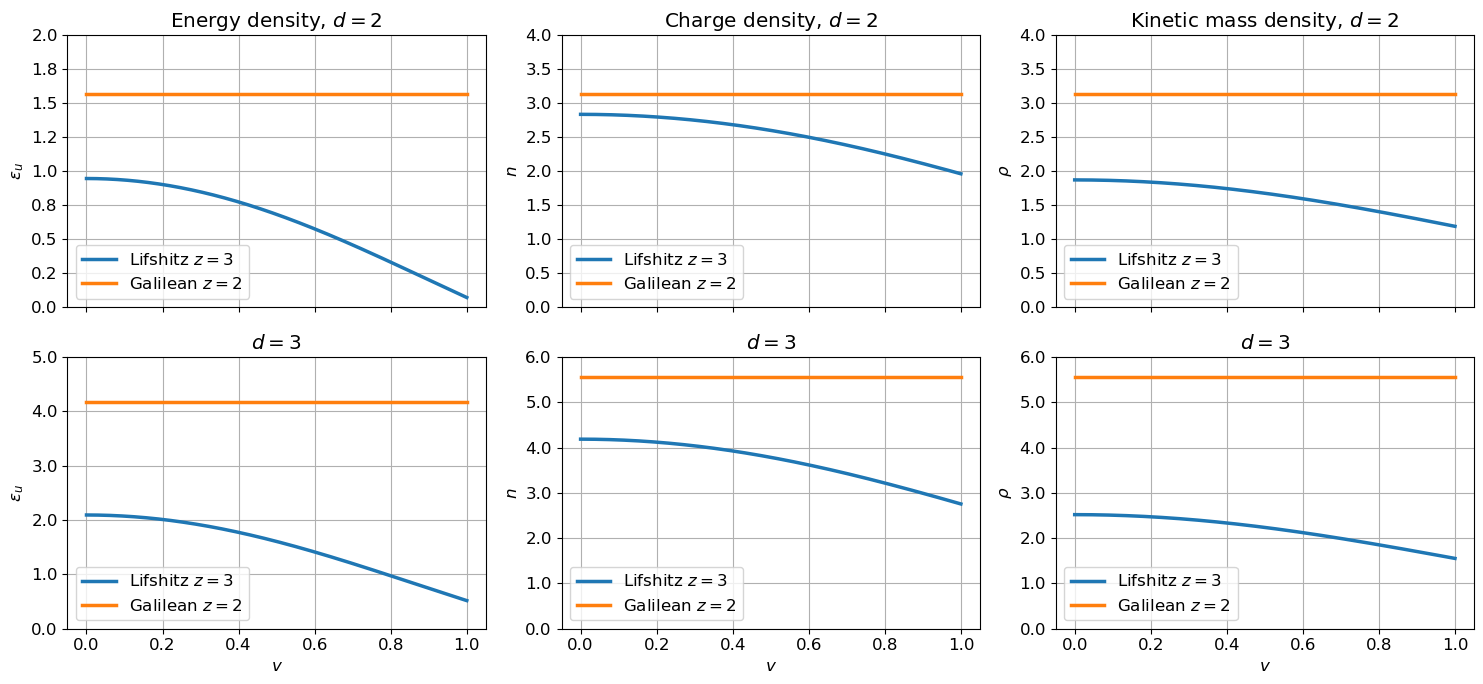}
	\caption{Energy density $\epsilon_u$, charge density $n$ and kinetic mass density $\rho$ as functions of the velocity $v$ in 2 and 3 dimensions for a Lifshitz and Galilean gas. The parameters used are $T=0.5$ and $\alpha=2m=2$.}
	\label{fig:lifshitzcharges}
\end{figure}

{\ As we can see from  \eqref{Eq:GenericChargeDensity} and \eqref{Eq:InternalEnergyDensity}, for a generic dispersion relation $\tilde{H}$, the thermodynamic quantities are functions of the gas velocity $\vec{v}$, and moreover $\rho\neq n$, which holds only for systems with Galilean symmetry. We can see in practice how the expressions we obtained for $\epsilon$, $\rho$ and $n$ behave for a specific example, namely that of a gas of Lifshitz particles. Consider a gas of free particles with a Lifshitz dispersion relation $\tilde{H}=\frac{(\vec{p}^2)^{z/2}}{\alpha}$, where $z$ is the dynamical exponent and $\alpha$ a constant that fixes the units; see appendix~\ref{appendix:Lifshitz} for further details. The non-relativistic case is obtained when $z=2$ and $\alpha=2m$, with $m$ the particles mass. Some thermodynamic quantities, for this specific class of Hamiltonians, are plotted in fig.~\ref{fig:lifshitzcharges} as functions of the fluid velocity. In particular, from the plots, we immediately see that $\epsilon_u$, $\rho$ and $n$ are independent of the velocity for a Galilean gas, furthermore $n=\rho$ both in $d=3$ and $d=2$ dimensions. Clearly, the same is not true for a Lifshitz gas ($z=3$ in the example): all functions decrease with the fluid velocity, and, moreover, $\rho\neq n$ even at zero velocity. This is something still under-looked and often missed in the literature \cite{PhysRevE.74.022101,Bertin_2009,PhysRevLett.133.128301,disalvo2025}.}

\section{Developing Hamiltonian mechanics}\label{sec:Hamiltonian}

{\noindent In the previous section we applied our formalism, now we shall derive in generality the results used. We shall pay special attention to coupling  Hamiltonian mechanics to a curved Aristotelian manifold, making allowances for the fact that there is no inner product in our formulation. We shall discuss how Killing vectors can be lifted from the Aristotelian manifold to a subclass of diffeomorphisms on the phase space. Subsequently, we shall introduce the ``free'' Hamiltonians and establish Liouville's theorem on the constraint surfaces of exotic free particle motion.}

\subsection{Phase space geometry }\label{sec:PhaseSpaceGeo}

{\noindent The phase space of a system is the manifold generated by all possible configurations of generalised coordinates and momenta for that system. Trajectories in phase space then represent how the system evolves from one configuration to another. We take as our fundamental starting point the Hamiltonian description (as opposed to using the Lagrangian and deriving Hamiltonian mechanics). Nevertheless, in a nod to Lagrangian mechanics we take our generalised momenta to be covectors and thus to transform covariantly under coordinate transformations. Subsequently, given a Hamiltonian function that describes the behaviour of the system, trajectories through phase space are generated by the corresponding Hamiltonian vector field whose integral curves are solutions to Hamilton's equations.}

{\ We begin with an orientable Aristotelian manifold $M$ which describes all the spacetime positions of the system. The cotangent bundle of any orientable manifold $M$ is defined as the collection of points
    \begin{eqnarray}
	 T^{*} M = \left\{ (x,p) : x \in M, p \in T_{x}^{*} M \right\}
    \end{eqnarray}
where $T_{x}^{*}M$ is the cotangent space to the manifold $M$ at $x$. If the manifold $M$ is orientable, then so is the cotangent bundle \cite{Acu_a_C_rdenas_2022}. The points in the cotangent space to a given point $x$ in the manifold $M$, i.e. $T_{x}^{*} M$, are covectors representing all the potential generalised momenta of the system. We shall see shortly that picking a preferred Hamiltonian will then select just one element of the cotangent space at each point on the manifold.}

{\ We will often make use of local adapted coordinates $x^{\mu}$ and $p_{\mu}$ to describe the neighbourhood of a given point in $T^{*}M$. A generic vector field $X(x,p)$ in the cotangent bundle $T^{*}M$ can be written in terms of locally adapted coordinates as
    \begin{eqnarray}
        X(x,p) = X^{\mu}(x,p) \left. \frac{\partial}{\partial x^{\mu}} \right|_{x}
            + Y_{\mu}(x,p) \left. \frac{\partial}{\partial p_{\mu}} \right|_{x}
    \end{eqnarray}
where the coordinate basis is evaluated at the point $x \in M$. However, this will not be the most convenient decomposition for our Hamiltonian theory on curved space. Consequently, we introduce a slightly modified (non-coordinate) basis
  \begin{eqnarray}
    \label{Eq:BasisVectors}
    D_{\mu} &=& \frac{\partial}{\partial x^{\mu}} + \Gamma^{\rho}_{\mu \nu} p_{\rho} \frac{\partial}{\partial p_{\nu}} \; , \qquad
    \mathrm{and} \qquad \frac{\partial}{\partial p_{\mu}} \; , 
  \end{eqnarray}
so that a generic vector field on the cotangent bundle is given by
	\begin{eqnarray}
	 X &=& X^{\mu} D_{\mu} + Z_{\mu} \frac{\partial}{\partial p_{\mu}} = X^{\mu} \partial_{\mu} + \left( Z_{\nu} + X^{\mu} \Gamma^{\rho}_{\mu \nu} p_{\rho} \right) \frac{\partial}{\partial p_{\nu}} 
	\end{eqnarray}
The spaces spanned by $D_{\mu}$ and $\frac{\partial}{\partial p_{\mu}}$ are called horizontal and vertical vector spaces respectively and they have a metric independent definition that we discuss in appendix \ref{appendix:horizontalandvertical}. These basis elements satisfy the following relations:
	\begin{subequations}
	\label{Eq:BasisCommutators}
		\begin{align}
     			& \left[ \frac{\partial}{\partial p_{\mu}}, \frac{\partial}{\partial p_{\nu}} \right] = 0 \; ,
			\qquad \left[ D_{\mu}, \frac{\partial}{\partial p_{\nu}} \right]
			   = -  \Gamma^{\nu}_{\mu \sigma}  \frac{\partial}{\partial p_{\sigma}}  \; , \\
			& \left[ D_{\mu}, D_{\nu} \right]
			= \left( 2 \partial_{[\mu} \Gamma^{\rho}_{\nu] \sigma}  + 2 \Gamma^{\rho}_{[\mu| \alpha}  \Gamma^{\alpha}_{|\nu] \sigma} \right) p_{\rho} \frac{\partial}{\partial p_{\sigma}}
			 = R\indices{_{\mu \nu \sigma}^{\rho}} p_{\rho} \frac{\partial}{\partial p_{\sigma}} \; ,
  		\end{align}
  	\end{subequations}
where non-commutivity of the basis elements indicates that this is generally a non-coordinate basis with non-trivial anholonomy coefficients \eqref{Eq:AnholonomyCoeffs}.}

{\ In the relativistic case, the decomposition of vectors into horizontal and vertical spaces is very useful as it can be used to uniquely uplift the spacetime metric to the Sasaki metric. The Sasaki metric is the unique metric on $T^{*}M$ respecting the decomposition into horizontal and vertical spaces \cite{Acu_a_C_rdenas_2022}. In our case, while there is no unique analogue of the Sasaki metric, we shall nevertheless find that our free Hamiltonian vector fields lie entirely in the horizontal space.}

{\  As we have a basis for the tangent spaces to the cotangent bundle, $T(T^{*}M)$, we can also construct a basis for the corresponding cotangent space. Each element of the dual basis will be defined such that it is one on a single given element of the original basis, and zero for all other elements. In particular, take the following one-form basis
  \begin{eqnarray}
  	\label{Eq:BasisOneForms}
   	Dp_{\mu} &=& \mathrm{d} p_{\mu} - \Gamma_{\alpha \mu}^{\beta} p_{\beta} \mathrm{d}x^{\alpha} \qquad \mathrm{and} \qquad \mathrm{d}x^{\mu} \; .
  \end{eqnarray}
These span the space of one-forms at a point on the cotangent bundle and satisfy the following relations:
	\begin{subequations}
	\label{Eq:CoTangentTangentBases}
	\begin{align}
		& \mathrm{d}x^{\mu} \left( D_{\nu} \right) = \delta^{\mu}_{\nu} \; , 
		 & \mathrm{d}x^{\mu} \left( D p_{\nu} \right) = 0 \; , \\
		 & Dp_{\mu} \left( D_{\nu} \right) 
		= 0 \; , \qquad
			& Dp_{\mu} \left( \frac{\partial}{\partial p_{\nu}} \right)
		= \delta_{\mu}^{\nu} \; . 
	\end{align}	
	\end{subequations}
}

{\ So far we have discussed coordinates on and (co)tangent spaces to the cotangent bundle $T^{*}M$. Now we introduce the Hamiltonian which is some preferred function $H$ defined on the cotangent bundle; this Hamiltonian defines the dynamics of particles. In terms of our basis one forms \eqref{Eq:CoTangentTangentBases}, the exterior derivative of any such Hamiltonian can be written in local coordinates as
  \begin{eqnarray}
    	   \mathrm{d}  H 
    &=& \frac{\partial H}{\partial p_{\mu}} \mathrm{d} p_{\mu} + \frac{\partial H}{\partial x^{\mu}} \mathrm{d} x^{\mu} = \frac{\partial H}{\partial p_{\mu}} D p_{\mu} + \left(  \frac{\partial H}{\partial x^{\mu}} + \frac{\partial H}{\partial p_{\nu}} \Gamma_{\mu \nu}^{\rho} p_{\rho}  \right) \mathrm{d}x^{\mu} \; . 
  \end{eqnarray}
For our free Hamiltonians the second term will vanish identically as it is proportional to terms like $\nabla_{\mu} \nu^{\nu}$ and $\nabla_{\mu} \tilde{h}^{\rho \sigma}$ which are zero for structure invariant compatible connections.}

{\ Our cotangent bundle is not yet equipped with any additional geometry. In fact, unlike the relativistic case, it is not clear that there is a (canonical) lift of the structure invariants from the manifold $M$ to the tangent bundle $T^{*}M$ which would make the cotangent bundle Aristotelian. This should be compared to the canonical lift of a spacetime metric to the Sasaki metric \cite{Acu_a_C_rdenas_2022}. Nevertheless, we can in a natural manner equip this $2(d+1)$-dimensional manifold with a closed, non-degenerate two-form $\Omega$ called the canonical symplectic form which has some useful properties. In local coordinates $\Omega$ is written
  \begin{eqnarray}
    \Omega = \mathrm{d} p_{\mu} \wedge \mathrm{d}x^{\mu} \; .
  \end{eqnarray}
We then define the Hamiltonian vector field $X_{H}$, corresponding to a given Hamiltonian $H$, by the following identity,
  \begin{eqnarray}
    \label{Eq:HamiltonianVectorFieldSymplecticForm}
    \mathrm{d} H(Z) &=& \Omega(Z,X_{H})
  \end{eqnarray}
where $Z$ is any vector in the tangent space to cotangent bundle. In a coordinate basis one finds
  \begin{eqnarray}
    \label{Eq:HamiltonianVectorField}
    X_{H}
    =  \frac{\partial H}{\partial  p_{\mu}} \partial_{\mu} - \frac{\partial H}{\partial x^{\mu}} \frac{\partial}{\partial p_{\mu}} 
    = \frac{\partial H}{\partial p_{\mu}} D_{\mu} - \left( \frac{\partial H}{\partial x^{\nu}} + \frac{\partial H}{\partial p_{\mu}} \Gamma^{\rho}_{\mu \nu} p_{\rho} \right) \frac{\partial}{\partial p_{\nu}} \; .
  \end{eqnarray}
so that integral curves of $X_{H}$ in an adapted basis give us Hamilton's equations displayed in \eqref{Eq:HamiltonsEqns1} and \eqref{Eq:HamiltonsEqns2}, with $\lambda$ some parameterisation of the integral curve. The integral curves of the vector field $X_{H}$ are the solutions to Hamilton's equations of motion. We will also have need of the following identity involving the interior product
      \begin{eqnarray}
      	\label{Eq:InteriorProductonOmega}
          i_{X_{H}} \Omega
      = X_{\mu}^{H} \mathrm{d}x^{\mu} - X^{\mu}_{H} \mathrm{d} p_{\mu} = - \Omega(\cdot,X_{H}) \; .
    \end{eqnarray}
We can now reaffirm the generic validity of Liouville's theorem on the cotangent bundle:}

\begin{theorem}[Liouville's theorem on $T^{*}M$]\label{thm:LiouvilleTM}
{\ There exists a canonical phase-space volume $\mathrm{vol}_{T^{*}M}$, defined in terms of the symplectic form $\Omega$, and given by
    \begin{eqnarray}
     \label{Eq:CanonicalCotangentMeasure}
      \mathrm{vol}_{T^{*}M} &=& - \frac{(-1)^{\frac{d(d+1)}{2}}}{(d+1)!} \bigwedge_{d+1} \Omega \; ,   
    \end{eqnarray}
which is conserved under the flow of \underline{any} Hamiltonian vector field.}
\end{theorem}

\begin{proof}
{\ Using Cartan's magic formula \eqref{Eq:CartanMagic} on the symplectic form and applying \eqref{Eq:InteriorProductonOmega} we have
    \begin{eqnarray}
            \mathcal{L}_{X_{H}} \Omega
        = \mathrm{d} \left(i_{X_{H}} \Omega\right) + i_{X_{H}}\left( d \Omega \right) = \mathrm{d} \left( - \mathrm{d} H \right) = 0 \; . \nonumber
    \end{eqnarray}
Now acting on \eqref{Eq:CanonicalCotangentMeasure} with the Lie derivative along the Hamiltonian vector field and employing the above identity gives the desired result.}
\end{proof}

{\ In local coordinates on the cotangent bundle, the volume form looks like
	\begin{subequations}
	\begin{eqnarray}
		\mathrm{vol}_{T^{*}M}
	&=& - \frac{1}{(d+1)!} \mathrm{d}p_{\mu_{1}} \wedge \ldots \wedge \mathrm{d} p_{\mu_{d+1}} \wedge \mathrm{d}x^{\mu_{1}} \wedge \ldots \wedge \mathrm{d} x^{\mu_{d+1}} \; , \\
		\label{Eq:canvolformcoords}
	&=& - \mathrm{d}p_{0} \wedge \ldots \wedge \mathrm{d} p_{d} \wedge \mathrm{d}x^{0} \wedge \ldots \wedge \mathrm{d} x^{d} \; , \\
		\label{Eq:ProductDecompVol}
	&=& -  \mathrm{vol}_{T^{*}_{x}M} \wedge  \mathrm{vol}_{M} \; , 
	\end{eqnarray}
	\end{subequations}
where the last equality follows from the fact that it is always possible, due to the local product manifold nature of fibre bundles, to write the volume form as a product of the base manifold volume form and a form restricted to the cotangent space at each point of the base manifold. In particular, one can take
	\begin{eqnarray}
		\mathrm{vol}_{T^{*}_{x}M} &=& \frac{1}{e} Dp_{0} \wedge \ldots \wedge Dp_{d} \; , 
	\end{eqnarray}
in the vicinity of a given point. Morally this latter form measures the volume on the cotangent space at a point $x$ of the manifold.}

{\ With this said, Liouville's theorem \ref{thm:LiouvilleTM} as stated above applies to the entire cotangent bundle $T^{*}M$ and tells us that for generic motions the phase-space volume along integral curves of the Hamiltonian vector field is conserved. However, for free exotic particles - just as in the relativistic case - motion will be constrained to codimension one submanifolds. In particular, our exotic particles will be described by equations with world-line reparameterisation invariance. It is a known result that the Hamiltonians of such systems vanish \cite{Henneaux:1992ig}. Consequently, to describe particle motion, we have to introduce new Hamiltonians given by an auxiliary field $\lambda$ multiplying constraints that determine the relevant submanifold (see appendix \ref{appendix:Lifshitz} for the example of a Lifshitz particle). It is far from given that Liouville's theorem applies to any such $(2d+1)$-form, so we must work to select a suitable one in the next sections.}

\subsection{Killing vectors on phase space }\label{sec:KillingUplift}

{\noindent As we noted in our discussion on Aristotelian geometry in section \ref{sec:AristotleanKilling}, infinitesimal symmetries on a manifold are encoded in Killing vectors \eqref{Eq:Killingconditions}. In this section we show how diffeomorphisms of the base manifold $M$, in particular Killing vectors, can be lifted up to the cotangent bundle $T^{*}M$. We will also discuss how the structure of infinitesimal symmetries manifests in Hamiltonian mechanics. We shall use these uplifted Killing vectors in the (next) section \ref{sec:FreeHamiltonians} to define a free Hamiltonian.}

{\ Following \cite{Rioseco_2017} we define the uplift of any infinitesimal diffeomorphism $\xi$ on the manifold $M$ to be a vector field on the cotangent bundle, denoted $\hat{\xi}$, by
	\begin{eqnarray}
		      \hat{\xi} 
		&=& \xi^{\mu} \partial_{\mu} - p_{\alpha} \frac{\partial \xi^{\alpha}}{\partial x^{\mu}} \frac{\partial}{\partial p_{\mu}} =  \xi^{\mu} D_{\mu} - p_{\nu} \left( \nabla_{\alpha} \xi^{\nu} - \Sigma\indices{_{\alpha\mu}^{\nu}} \xi^{\mu} \right) \frac{\partial}{\partial p_{\alpha}} \; .
	\end{eqnarray}
These uplifts can be shown to be infinitesimal diffeomorphisms on the cotangent bundle \cite{Acu_a_C_rdenas_2022}. It is useful to know the action of such an uplift  on the basis elements of the tangent space to the cotangent bundle (i.e. \eqref{Eq:BasisVectors}),
	\begin{subequations}
	\begin{eqnarray}
			\label{Eq:LieDerivativeXHDmu}
			\mathcal{L}_{\hat{\xi}} D_{\mu}
		&=& - \partial_{\mu} \xi^{\nu}  D_{\nu} + \left(  \nabla_{\mu} \nabla_{\nu} \xi^{\sigma} + R\indices{_{\rho \mu \nu}^{\sigma}} \xi^{\rho} \right) p_{\sigma} \frac{\partial}{\partial p_{\nu}}  \; ,  \\
			 \mathcal{L}_{\hat{\xi}} \frac{\partial}{\partial p_{\mu}}
		&=&  \frac{\partial \xi^{\mu}}{\partial x^{\nu}}  \frac{\partial}{\partial p_{\nu}} \; .
	\end{eqnarray}
	\end{subequations}
The index order of the Riemann tensor in \eqref{Eq:LieDerivativeXHDmu} should be compared to the relativistic case in \cite{Acu_a_C_rdenas_2022} where, in that case, the first (algebraic) Bianchi identity was employed.	}

{\ We remind ourselves that an Aristotelian Killing field (def.~\ref{def:Killingfield}) is a subclass of infinitesimal diffeomorphisms. In the next proposition we discuss properties of the uplift of such a Killing field $\xi$ from the base manifold $M$ to the cotangent bundle $T^{*}M$. The resultant vector field can then be treated as the Hamiltonian vector field for some new scalar potential $F$ defined on $T^{*}M$. However, to define this potential we first need one additional definition - the symplectic potential, otherwise known as the tautological one-form. As $\Omega$, the canonical symplectic form, is closed $\mathrm{d} \Omega = 0$ we should be able to locally write it as the derivative of a one-form which motivates:}

\begin{definition}[The tautological one-form/symplectic potential]\label{def:symplectic}
{\  Let $\Omega$ be the symplectic form on $M$. The one form $\Theta$, called the tautological one-form or symplectic potential, is such that it satisfies $\mathrm{d} \Theta = \Omega$. In local coordinates, this form can be written
    \begin{eqnarray}
    	\label{Eq:TautologicalOneFormCoords}
        \Theta &=& p_{\mu} \mathrm{d}x^{\mu} \; .
    \end{eqnarray}
}
\end{definition}

{\noindent This now allows us to make the following observations:}

\begin{proposition}[Properties of the uplift of a Killing vector field]\label{prop:upliftxi}
\label{prop:UpliftKilling}
{\ There are a couple of properties that the uplift of the Killing field has that are interesting for our purposes:
	\begin{enumerate}
		\item The uplifted Killing vectors satisfy the same algebra as the original Killing vectors
			\begin{eqnarray}
				\left[ \hat{\xi}, \hat{\psi} \right] &=& \hat{\left[ \xi, \psi \right]} \; ,
			\end{eqnarray}
		where $\xi$ and $\psi$ are two Killing vectors on the base manifold.
		\item $\hat{\xi}$ generates a symplectic flow on $T^{*}M$ i.e.
			\begin{eqnarray}
				\mathcal{L}_{\hat{\xi}} \Omega = 0 \; .
			\end{eqnarray}
		For this flow the Hamiltonian is $F=\Theta(\hat{\xi})$ and $\hat{\xi}=X_{F}$ the corresponding Hamiltonian vector field.	Consequently, the canonical volume form $\mathrm{vol}_{T^{*}M}$ is conserved along integral curves of $\hat{\xi}$.
	\end{enumerate}
}
\end{proposition}

\begin{proof}
{\ These proofs are given in \cite{Acu_a_C_rdenas_2022}.}
\end{proof}

{\ The definition of $F$ in the above proposition \ref{prop:upliftxi}, naturally extends to multiple Killing vectors $\xi_{a}$ and their uplifts $\hat{\xi}_{a}$ where $F_{a} = \Theta(\hat{\xi}_{a})$. We caution the reader that we are not claiming the uplifted Killing vector is a Killing vector on $T^{*}M$, as we have not defined any geometric structure on $T^{*}M$. Instead it is only the projection $\xi$ that has the Killing properties. Exploring the geometry of $\hat{\xi}$ and introducing suitable structure to make them in a sense Killing vectors is an interesting question beyond the scope of this work. With this said, the potential for multiple Killing vectors on $M$ motivates the introduction of one further pieces of formalism:}

\begin{definition}[Poisson bracket]\label{def:Poissonbracket}
{\ Let $F$ and $G$ be two functions on $T^{*}M$. The Poisson bracket of $F$ and $G$, denoted $\left\{ F,G \right\}$ is defined by
	\begin{eqnarray}
		\left\{ F, G \right\} = \Omega \left( X_{F}, X_{G} \right)
	\end{eqnarray}
where $X_{F}$ and $X_{G}$ are the Hamiltonian vector fields \eqref{Eq:HamiltonianVectorFieldSymplecticForm} associated with $F$ and $G$ respectively.}
\end{definition}

{\ The Poisson bracket satisfies several well-known properties including bilinearity, anticommutivity, the Jacobi identity and the Leibniz rule. In local coordinates this new object can be decomposed as \eqref{Eq:PoissonLocal}. Further, when acting on the Hamiltonians associated with Killing vectors one finds
	\begin{eqnarray}
		\left\{ F_{a}, F_{b} \right\} = \left[ \hat{\xi}_{a}, \hat{\xi}_{b} \right]= \hat{\left[ \xi_{a},  \xi_{b} \right] }
		= \mathcal{C}\indices{_{ab}^{c}} \hat{\xi}_{c} \; , 
	\end{eqnarray}
with $\mathcal{C}\indices{_{ab}^{c}}$ being the structure constants of the relevant Lie algebra associated with $\left\{ \xi_{a} \right\}$. Subsequently, the scalar $F_{a}$ is conserved under Hamiltonian flow if
	\begin{eqnarray}
		X_{H}[F_{a}] &=& \mathrm{d}F_{a}[X_{H}] = \Omega(\hat{\xi}_{a},X_{H}) = \left\{ H, F_{a} \right\} = 0 \; , 
	\end{eqnarray}
which implies the Hamiltonian analogue of Noether's theorem \cite{Acu_a_C_rdenas_2022} i.e. infinitesimal symmetries correspond to conserved charges on every solution to Hamilton's equations of motion.}

\subsection{``Free'' Hamiltonians }\label{sec:FreeHamiltonians}

{\noindent In Galilean and relativistic mechanics a central role is played by the ``free'' Hamiltonians. These Hamiltonians are invariant along integral curves of all the infinitesimal symmetry generators of the base manifold. Here we generalise the usual free Hamiltonians of Galilean and relativistic particles to boost agnostic systems. To do this, we first begin by listing the invariant scalars and vectors:}

\begin{lemma}[Conservation of some free invariants]
{\ Given \underline{any} Killing vector $\xi$ on an Aristotelian manifold $M$, the scalar quantities 
    \begin{subequations}
    \label{Eq:AllPrimitiveInvariants}
    \begin{eqnarray}
    	\label{Eq:PrimitiveInvariantsScalars}
      p_{\mu} \tilde{h}^{\mu \nu} p_{\nu} \; , \qquad \nu^{\mu} p_{\mu} \; , 
    \end{eqnarray}
and the vector quantities
	\begin{eqnarray}
		\label{Eq:PrimitiveInvariantsVectors}
		\tau_{\mu} \frac{\partial}{\partial p_{\mu}} \; , \qquad p_{\mu} \frac{\partial}{\partial p_{\mu}} \; , 
	\end{eqnarray}
	\end{subequations}
are invariant under integral curves of the uplift $\hat{\xi}$.}
\end{lemma}

\begin{proof}
{\ This can be shown by explicit computation in locally adapted coordinates:
	\begin{eqnarray}
          \mathcal{L}_{\hat{\xi}} \left( p_{\mu} \nu^{\mu} \right)
	    &=& p_{\nu} \mathcal{L}_{\xi} \nu^{\mu} = 0 \; , \nonumber \\
       		\mathcal{L}_{\hat{\xi}}\left( p_{\mu} \tilde{h}^{\mu \nu} p_{\nu} \right)
       &=& p_{\rho} ( \mathcal{L}_{\xi} \tilde{h}^{\rho \sigma} ) p_{\sigma} = 0 \; , \nonumber
    \end{eqnarray}
where we have used Killing conditions in Aristotelian spaces \eqref{Eq:Killingconditions} to arrive at the final result. Similarly,
	\begin{eqnarray}
	\mathcal{L}_{\hat{\xi}} \left[ \tau_{\mu} \frac{\partial}{\partial p_{\mu}}	\right]
		&=& \left[ \hat{\xi}, \tau_{\mu} \frac{\partial}{\partial p_{\mu}} \right]
		=  \left( - \tau_{\nu} \partial_{\mu} \xi^{\nu} \right)  \frac{\partial}{\partial p_{\mu}}  + \tau_{\mu} \frac{\partial \xi^{\mu}}{\partial x^{\nu}} \frac{\partial}{\partial p_{\nu}} = 0 \; , \nonumber \\
			\mathcal{L}_{\hat{\xi}} \left[ p_{\mu} \frac{\partial}{\partial p_{\mu}}	\right]
		&=& \left[ \hat{\xi}, p_{\mu} \frac{\partial}{\partial p_{\mu}} \right] = \left( - p_{\alpha} \partial_{\mu} \xi_{\alpha} \right)  \frac{\partial}{\partial p_{\mu}}  + p_{\mu} \frac{\partial \xi^{\mu}}{\partial x^{\nu}} \frac{\partial}{\partial p_{\nu}} = 0 \; , \nonumber
	\end{eqnarray}
where we have used \eqref{Eq:BasisCommutators} and the Killing conditions \eqref{Eq:Killingconditions}.}
\end{proof}

{\ We should note that the following vector quantities,
	\begin{eqnarray}
		\nu^{\mu} D_{\mu} , \qquad p_{\mu} \tilde{h}^{\mu \nu} D_{\nu} 	
	\end{eqnarray}
which we might be tempted to add to our set of free invariants are not generally invariant under flows generated by $\hat{\xi}$. In particular, it is generally the case that
	\begin{subequations}
	\begin{eqnarray}
			\mathcal{L}_{\hat{\xi}} \left[ \nu^{\mu} D_{\mu} \right]
		&=& \left( \nu^{\mu}  p_{\alpha} \frac{\partial^{2} \xi^{\alpha}}{\partial x^{\mu} \partial x^{\nu}} \right) \frac{\partial}{\partial p_{\nu}} \neq 0 \; , \qquad \\
			\mathcal{L}_{\hat{\xi}}\left[ p_{\mu} \tilde{h}^{\mu \nu} D_{\nu} \right]
		&=& \left(  \tilde{h}^{\mu \sigma} p_{\sigma}  p_{\alpha} \frac{\partial^{2} \xi^{\alpha}}{\partial x^{\mu} \partial x^{\nu}}   \right)  \neq 0 \; , \qquad 
		\end{eqnarray}
	\end{subequations}
where we have used \eqref{Eq:InverseKillingconditions}. It may be possible to construct vectors and scalars that are invariant using the Riemann curvature tensor, yet we have found this difficult to achieve given that, unlike the relativistic case, invariance of the structure invariants does not imply invariance of the connection under flows of the Killing field (see \eqref{Eq:Variationofcurvaturequantities}). We leave the discovery of a conclusive answer to future work.}

{\ The statements of the above lemma are true independent of the Hamiltonian (they make no reference to such) or the exact form of the structure invariants. This motivates the following definition:}

\begin{definition}[Free Hamiltonian]
{\ A free Hamiltonian on an Aristotelian manifold $M$ is any scalar function of the free scalar invariants in \eqref{Eq:PrimitiveInvariantsScalars}.}
\end{definition}

{\ For future reference, in local coordinates, our free Hamiltonians take the form
	\begin{subequations}
	\begin{eqnarray}
        \label{Eq:FreeHamiltonian}
        H_{\mathrm{f}} &=& H_{\mathrm{f}}\left[  \nu^{\mu} p_{\mu},  \tilde{h}^{\mu \nu} p_{\mu} p_{\nu} \right] \; , \\
        \label{Eq:FreeHamiltoniansSimplify}
        X_{\mathrm{f}} &=& \left[ \frac{\partial H_{\mathrm{f}}}{\partial (\nu^{\rho} p_{\rho})} \nu^{\mu} + 2 \frac{\partial H_{\mathrm{f}}}{\partial (\tilde{h}^{\rho \sigma} p_{\rho} p_{\sigma})} \tilde{h}^{\mu \nu} p_{\nu}  \right] D_{\mu} \; ,  \\ 
        	\label{Eq:FreeHExterior}
        \mathrm{d} H_{\mathrm{f}}
	&=& X_{\mathrm{f}}^{\mu} D p_{\mu} \; ,
    \end{eqnarray}
    \end{subequations}
i.e. the Hamiltonian vector field associated with a free Hamiltonian is entirely horizontal. The motivation for such free Hamiltonians is nicely physical as in flat space we readily identify $p_{\mu} h^{\mu \nu} p_{\nu}$ with the spatial momentum squared while $-\nu^{\mu} p_{\mu}$ is the usual notion of particle energy. We remark to the reader that the property of being horizontal is quite useful (as we shall see) and remind them that it is shared by the relativistic free particle Hamiltonian.}

\begin{theorem}[The free invariants are preserved under free Hamiltonian flow]
{\ The free invariants of \eqref{Eq:AllPrimitiveInvariants} are preserved under Hamiltonian flow of \underline{any} free Hamiltonian.}
\end{theorem}

\begin{proof}
{\ We notice that
    \begin{subequations}
    \label{Eq:TrivialSolnsBoltzmann}
    \begin{eqnarray}
          \mathcal{L}_{X_{\mathrm{f}}} \left( p_{\mu} \nu^{\mu} \right)
      &=& X_{\mathrm{f}}^{\nu} p_{\mu}  \nabla_{\nu} \nu^{\mu} = 0 \; , \nonumber \\
      	     \mathcal{L}_{X_{\mathrm{f}}}\left( p_{\mu} \tilde{h}^{\mu \nu} p_{\nu} \right)
       &=& X_{\mathrm{f}}^{\rho} p_{\mu}  \nabla_{\rho} \tilde{h}^{\mu \nu} p_{\nu} = 0 \; , \nonumber 
    \end{eqnarray}
    \end{subequations}
which shows the free scalar invariants are preserved under Hamiltonian flow. As for the vector invariants, we first compute the action of a Hamiltonian vector field associated with a generic free Hamiltonian on a generic vector field of the form
	\begin{eqnarray}
		V = V^{\mu} D_{\mu} + V_{\mu} \frac{\partial}{\partial p_{\mu}} \; . \nonumber
	\end{eqnarray}
The result is:
		\begin{eqnarray}
			\mathcal{L}_{X_{\mathrm{f}}} V &=& X_{\mathrm{f}}^{\mu} \left(  \nabla_{\mu}  V_{\rho} + \Gamma_{\mu \sigma}^{\nu} p_{\nu} \frac{\partial}{\partial p_{\sigma}} V_{\rho} + V^{\nu} R\indices{_{\mu\nu\rho}^{\sigma}} p_{\sigma} \right) \frac{\partial}{\partial p_{\rho}} \nonumber \\
		&\;& + \left( X^{\nu}_{H_{\mathrm{f}}} \partial_{\nu} V^{\mu} - V^{\nu} \partial_{\nu} X_{\mathrm{f}}^{\mu} \right) D_{\mu} \nonumber \\
		&\;& + \Gamma_{\nu \sigma}^{\rho} p_{\rho} \left( X^{\nu}_{H_{\mathrm{f}}} \frac{\partial}{\partial p_{\sigma}} V^{\mu} - V^{\nu} \frac{\partial}{\partial p_{\sigma}} X_{\mathrm{f}}^{\mu} \right) D_{\mu}   -  V_{\nu} \left( \frac{\partial}{\partial p_{\nu}}  X_{\mathrm{f}}^{\mu} \right)  D_{\mu} \; . \nonumber
	\end{eqnarray}
where we have employed \eqref{Eq:FreeHamiltoniansSimplify}. Upon substituting in our free vector invariants \eqref{Eq:PrimitiveInvariantsVectors} we see this is identically zero and thus the free vector invariants are constant along Hamiltonian flows.}
\end{proof}

{\  As an example, consider the following Hamiltonians introduced in \ref{sec:applications},
     \begin{eqnarray}
     	\label{Eq:GenericCurvedOPDispersion}
        H_{*} &=& \lambda \left( p_{\mu} \nu^{\mu} + \tilde{H} (p_{\rho} \tilde{h}^{\rho \sigma} p_{\sigma}) \right) \; .
    \end{eqnarray}
of which the Lifshitz Hamiltonians in appendix \ref{appendix:Lifshitz} are a subclass. This Hamiltonian has a single dispersion relation as it's level-set, the zero value of which we can parameterise by 
	\begin{eqnarray}
		p_{\mu} \nu^{\mu} &=& -  \tilde{H} (p_{\rho} \tilde{h}^{\rho \sigma} p_{\sigma}) \; . 
	\end{eqnarray}
This is the origin of the expression \eqref{eq:HSingleDispersion}. Notice that this is not true for a general free Hamiltonian as in principle there could be multiple disconnected solutions to $H=0$ when $p_{\mu} \nu^{\mu}$ occurs non-linearly. In fact, this is precisely what happens in the relativistic case where the analogue of $p_{\mu} \nu^{\mu}$ appears quadratically. This manifests as the appearance of disconnected future and past light cones for massive particles; which intersect at zero momentum for a massless particle.}

{\ If we specialise to flat spacetime and introduce Cartesian coordinates on the Aristotelian manifold $x^{0}, \ldots x^{d}$ such that $\tau = (1,\vec{0})$ and $h=\delta_{ij}$ then the associated cotangent bundle is also flat. We can introduce Cartesian coordinates on the cotangent space at a given point in the base manifold such that
	\begin{eqnarray}
	\label{Eq:HypersurfaceCoords}
	 p_{\mu} \nu^{\mu} = - p_{0} \; , \qquad p_{\mu} \tilde{h}^{\mu \nu} p_{\nu} = \delta_{ij} p_{i} p_{j} = \vec{p}^2 \; ,
	\end{eqnarray}
so that our Hamiltonians  \eqref{Eq:GenericCurvedOPDispersion} become  \eqref{eq:HSingleDispersion}. As the Hamiltonian is conserved under its own integral curves\footnote{A short proof is given in \cite{Acu_a_C_rdenas_2022}.},  we find that particle motion is constrained to surfaces where $H_{*}=c$, which can be denoted by $H_{*}^{-1}(c) \subset T^{*}M$. We can parameterise $H_{*}^{-1}(c)$ as the set of points $(x^{0},\ldots,x^{d},p_{0}(\vec{p}^2),p_{1},\ldots,p_{d})$ and discover \eqref{Eq:p0sub}. The constant $c$ is nothing more than the usual choice of what we call zero energy. Additionally, for later use, we find that the Hamiltonian vector fields of systems defined by \eqref{eq:HSingleDispersion} in our coordinate system take the form:
	\begin{eqnarray}
		\label{Eq:XHSingleDispersion}
		X_{*} = \lambda \left( - \partial_{t} + 2 \frac{\partial \tilde{H}}{\partial \vec{p}^2} \vec{p} \cdot \vec{\partial} \right) \; , \; \; \; 
		\frac{\partial H_{*}}{\partial (p_{\mu} \nu^{\mu})} = \lambda \; , \; \;  \; \frac{\partial H_{*}}{\partial (p_{\rho} \tilde{h}^{\rho \sigma} p_{\sigma})} = \lambda \frac{\partial \tilde{H}}{\partial \vec{p}^2}  \; . \qquad
	\end{eqnarray}
}

\subsection{Embedding codimension one hypersurfaces of $M$ in $T^{*}M$ }

{\noindent  When defining conserved charges and conservation equations on the Aristotelian manifold $M$ in section \ref{sec:aristotlechargecons}, we isolated some subset of the particles using a bounding (closed) surface $S$.  The particles inside of this bounding surface are still generically able to explore the full range of momenta; hence, we expect that an uplift of this bounding surface to the cotangent bundle $T^{*}M$ spans the full range of momenta but a limited subset of configuration space. In this section we discuss how to uplift these closed surfaces from $M$ to $T^{*}M$. Then, in the next section, we shall demonstrate how the charge conservation of section \ref{sec:aristotlechargecons} can also be uplifted.}

{\ To begin, suppose that we have a closed co-dimension one submanifold $S \subset M$ with boundary $\partial S$. Further, let $n = n_{\mu} \mathrm{d}x^{\mu}$ be a normal form of the boundary $\partial S$ so that $n(X)=0$ for any $X$ tangent to $\partial S$.  We consider the corresponding $2(d+1)$ and $(2d+1)$-dimensional sub-manifolds defined locally by
    \begin{subequations}
    \label{Eq:EmbeddingPhaseSpace}
    \begin{eqnarray}
        \Sigma = \left\{ (x,p): x \in S, p \in \left. T^{*}M \right|_{x} \right\} \; , \\
        \partial \Sigma = \left\{ (x,p): x \in \partial S, p \in \left. T^{*}M  \right|_{x} \right\} \; ,
    \end{eqnarray}
    \end{subequations}
and note that the projection of $\Sigma$ and $\partial \Sigma$ onto real space are $S$ and $\partial S$ respectively. These new manifolds \eqref{Eq:EmbeddingPhaseSpace} contain all the possible configurations of the system compatible with their being bounded by $S$ in spacetime. We now shall determine the normal form $N$ defining the space of tangent vectors to $\partial \Sigma$ in terms of $n$.}

{\ Let $Z$ be a vector tangent to $\partial \Sigma$ at a point $(x,p)$ in $T^{*}M$ and $\gamma(t)$ a curve in $T^{*}M$ with tangent $Z$ at $(x,p)$. The projection down to $M$ of the curve $\gamma(t)$ is another curve $\tilde{\gamma}(t) = \pi \cdot \gamma(t)$ entirely in $M$. If $Z$ has the form
	\begin{eqnarray}
		Z &=& X^{\mu} \partial_{\mu} + Y_{\mu} \frac{\partial}{\partial p_{\mu}}
	\end{eqnarray}
then the tangent to $\tilde{\gamma}(t)$ at the point $x$ of $M$ is
    \begin{eqnarray}
    	\label{Eq:TangentVectorS}
        \frac{d}{dt} \gamma'(t) &=& X^{\mu} \left. \partial_{\mu} \right|_{x} \; .
    \end{eqnarray}
From our definition of $\partial \Sigma$ in \eqref{Eq:EmbeddingPhaseSpace} we immediately know that this vector, \eqref{Eq:TangentVectorS}, is tangent to $\partial S$. As $n_{\mu} \mathrm{d}x^{\mu}$ is the normal form of $\partial S$ it follows that 
	\begin{eqnarray}
		\label{Eq:TangetVectorScondition}
		n_{\mu} X^{\mu} =0 \; . 
	\end{eqnarray}	
Let the normal form to $\partial \Sigma$ be $N = N_{M} \mathrm{d}x^{M}$. As $Z$ is tangent to $\partial \Sigma$, it must be set to zero by $N$ i.e.
    \begin{eqnarray}
       N_{M} Z^{M} = N_{\mu} X^{\mu} + N^{\mu} Y_{\mu} = 0 \; ,
    \end{eqnarray}
Using \eqref{Eq:TangetVectorScondition} in the above expression tells us that $N^{\mu} Y_{\mu} =0$, for any $Y_{\mu}$. Consequently it follows that $N^{\mu} \equiv 0$. In an abuse of notation, we can therefore identify
    \begin{eqnarray}
    	\label{Eq:NormalIdentification}
        N &=& N_{\mu} \mathrm{d}x^{\mu} = n_{\mu} \mathrm{d}x^{\mu} = n
    \end{eqnarray}
as the shared normal form of $\partial \Sigma$ and $\partial S$. We shall make use of this normal form when we consider charge conservation in section \ref{sec:PhaseSpaceChargeCons}.}

\subsection{Level sets of the Hamiltonian and Liouville's theorem }\label{sec:embedding}

{\noindent Drawing together the developments in previous sections we now seek to achieve the aim we detailed at the end of section \ref{sec:PhaseSpaceGeo}; we want a volume form defined on the $H=c$ energy level sets, denoted $\mathrm{vol}_{H^{-1}(c)}$, which is conserved under free particle motion in $H^{-1}(c)$.  Such volume forms are naturally more ambiguous than in the relativistic case as there is no Sasaki metric to aid us, nevertheless we shall isolate at least one suitable candidate with all the reasonable properties one might desire.}

{\ Let $X \in T(T^{*}M)$ be any vector which is tangent to the $H=c$ hypersurface. It follows that
  \begin{eqnarray}
      \mathrm{d}H(X) = - \left(i_{X_{H}} \Omega\right)(X) = 0 \; , 
  \end{eqnarray}
from which we can define a class of normal one-forms
  \begin{eqnarray}
  	\label{Eq:PrimitiveNormalForm}
	c \left[ \frac{\partial H}{\partial p_{\mu}} D p_{\mu}
           + \left(  \frac{\partial H}{\partial x^{\mu}} + \Gamma_{\mu \nu}^{\rho} p_{\rho} \frac{\partial H}{\partial p_{\nu}} \right) \mathrm{d}x^{\mu} \right] \; , 
  \end{eqnarray}
where $c$ is some arbitrary non-zero function of the hypersurface coordinates. We note in particular that $X_{H}$ is tangent to the $H=c$ hypersurface as
	\begin{eqnarray}
		\mathrm{d}H (X_{H}) = \Omega(X_{H},X_{H}) = 0 \; . 
	\end{eqnarray}
This is a special property of the Hamiltonian vector field based on its definition and does not apply to general vectors. Moreover, as $X_{H}$ is tangent to the $H=c$ surface, it is straightforward to argue from the definition of the Lie derivative that a tensor defined on $H^{-1}(c)$, acted upon by $\mathcal{L}_{X_{H}}$, remains a tensor of the same type defined on $H^{-1}(c)$.}

{\ We now encounter a core ambiguity related to the Aristotelian nature of $M$; even if we assume that the normal form \eqref{Eq:PrimitiveNormalForm} is well-defined up to scaling, without a bilinear map such as a metric, we cannot map $n$ uniquely to a vector. Compare this to the relativistic case where we can find a unique vector $V$ such that
	\begin{eqnarray}
		n(X) &=& g(V,X) = 0 \; , \qquad g(V,V) = \pm 1 \; ,
	\end{eqnarray}
for any vector $X$ tangent to the $H=c$ surface where $g$ is the Sasaki metric and the sign depends on whether the surface is timelike or spacelike.  As $g$ is non-degenerate and the surface is not null, $V$ is unique. The Sasaki metric provides a canonical (basis independent) map between the cotangent $T^{*}(TM)$ and the tangent bundle $T(TM)$.}

{\ With this said, consider a general vector $V$ which for want of a better name we shall term the Liouville quasinormal. Let $V$ contain at least some component that does not lie tangent to a given level set of $H$ i.e.
	\begin{subequations}
	\label{Eq:TangentCondition}
  	\begin{eqnarray}
    	\mathrm{d}H[V] &=& \frac{\partial H}{\partial p_{\mu}} V_{\mu} +  \left(  \frac{\partial H}{\partial x^{\mu}} + 	\Gamma_{\mu \nu}^{\rho} p_{\rho} \frac{\partial H}{\partial p_{\nu}} \right) V^{\mu} \neq 0 \; , \\
   	 V &=& V^{\mu} D_{\mu} + V_{\mu} \frac{\partial}{\partial p_{\mu}} \; .
  \end{eqnarray}
  \end{subequations}
It is our eventual purpose to use $V$ to define a volume form on $H^{-1}(c)$ via the interior product
	\begin{eqnarray}
		\label{Eq:ConstantHVolForm}
		\mathrm{vol}_{H^{-1}(c)} := \iota^{*}(i_{V} \mathrm{vol}_{T^{*}M}) \; , 
	\end{eqnarray}
where $\iota : H^{-1}(c) \rightarrow T^{*}M$ is the inclusion map which is necessary so that $\mathrm{vol}_{H^{-1}(c)} \in \Omega^{2d+1}(H^{-1}(c))$ - the space of $(2d+1)$-forms defined on $H^{-1}(c)$. In particular, $\mathrm{vol}_{H^{-1}(c)}$ takes as arguments vector fields in $T(H^{-1}(c))$ and not $T(T^{*}M)$. The pullback via the inclusion map is thus an important detail so that the Lie derivative of $\mathrm{vol}_{H^{-1}(c)}$ along $X_{H}$ remains in $\Omega^{2d+1}(H^{-1}(c))$ - the space of $(2d+1)$-forms defined on $H^{-1}(c)$.}

{\  That $V$ is not entirely tangent to the $H=c$ hypersurface \eqref{Eq:TangentCondition} is not much to go on. A natural restriction we can impose on $V$ is to require it satisfies the core property we were searching for, namely, conservation of \eqref{Eq:ConstantHVolForm} under Hamiltonian flow. We embody this in the following lemma:}

\begin{lemma}[Liouville's theorem on the constant energy hypersurface]\label{lem:LiouvilleonH}
{\ Any non-zero Liouville quasinormal $V$ to the constant Hamiltonian ($H=c$) hypersurface that satisfies
	\begin{subequations}
	\begin{eqnarray}
		\label{Eq:TangencyCondition}
		\mathrm{d}H[\mathcal{L}_{X_{H}} V] = 0 \; , 
	\end{eqnarray}
induces a volume form $\mathrm{vol}_{H^{-1}(c)}$ through \eqref{Eq:ConstantHVolForm} that satisfies Liouville's theorem 
	\begin{eqnarray}
		\mathcal{L}_{X_{H}} \mathrm{vol}_{H^{-1}(c)}= 0 \; . 
	\end{eqnarray}
	\end{subequations}
}
\end{lemma} 

\begin{proof}
{\  The details are a minor modification of the proof presented in \cite{Acu_a_C_rdenas_2022} for the relativistic volume form on a constant mass surface. Essentially one acts the volume form \eqref{Eq:ConstantHVolForm} on vectors tangent to the levelset and then applies the Lie derivative along $X_{H}$ to the result. As the Lie derivative of $V$ is tangent to the surface \eqref{Eq:TangencyCondition}, one can show every resultant term is zero either by over-saturating forms with vectors tangent to the surface, or by Liouville's theorem (theorem \ref{thm:LiouvilleTM}) in $T^{*}M$.}
\end{proof}

{\ For example, in the case of the relativistic particle, one can use the Sasaki metric to find the canonical normal to the Hamiltonian level sets. It has the form \cite{Acu_a_C_rdenas_2022}
	\begin{eqnarray}
		\label{Eq:RelativisticNormal}
		N = \frac{1}{m} p_{\mu} \frac{\partial}{\partial p_{\mu}}
	\end{eqnarray}
with $m$ the mass and one quickly finds that $\mathcal{L}_{X_{H}} N \propto X_{H}$. Subsequently, because the Lie derivative of any vector field along itself vanishes identically i.e.
	\begin{eqnarray}
		\mathcal{L}_{X_{H}} X_{H} = 0 \; , \nonumber
	\end{eqnarray}
and $\mathrm{d}H[\mathcal{L}_{X_{H}} N] \propto \mathrm{d}H[X_{H}] = 0$, the volume form $\mathrm{vol}_{H^{-1}(c)}$ \underline{in the relativistic case} is preserved under integral curves of the Hamiltonian vector field.}

{\ More generally, given a Liouville quasinormal $V$, we can compute the Lie derivative of this object along a generic Hamiltonian vector field. When restricting to our free Hamiltonians we find
	\begin{eqnarray}
			\mathcal{L}_{X_{\mathrm{f}}} V 
		&=& X_{\mathrm{f}}^{\mu} \left(  \nabla_{\mu}  V_{\rho} + \Gamma_{\mu \sigma}^{\nu} p_{\nu} \frac{\partial}{\partial p_{\sigma}} V_{\rho} + V^{\nu} R\indices{_{\mu\nu\rho}^{\sigma}} p_{\sigma} \right) \frac{\partial}{\partial p_{\rho}} \nonumber \\
		&\;& + \left( X^{\nu}_{H_{\mathrm{f}}} \partial_{\nu} V^{\mu} - V^{\nu} \partial_{\nu} X_{\mathrm{f}}^{\mu} \right) D_{\mu} \nonumber \\
		&\;& + \Gamma_{\nu \sigma}^{\rho} p_{\rho} \left( X^{\nu}_{H_{\mathrm{f}}} \frac{\partial}{\partial p_{\sigma}} V^{\mu} - V^{\nu} \frac{\partial}{\partial p_{\sigma}} X_{\mathrm{f}}^{\mu} \right) D_{\mu} \nonumber \\
		&\;&  -  V_{\nu} \left( \frac{\partial}{\partial p_{\nu}}  X_{\mathrm{f}}^{\mu} \right)  D_{\mu} \; . 
	\end{eqnarray}
We require this to be totally tangent to the constant Hamiltonian level set - independently of the geometry (including the connection) or the Hamiltonian (modulo that the Hamiltonian be of the free type). Thus
	\begin{eqnarray}
			\label{Eq:XHFlowLiouvilleNorm}
			\mathrm{d} H_{\mathrm{f}}[\mathcal{L}_{X_{\mathrm{f}}} V]
		&=& X_{\mathrm{f}}^{\mu} X_{\mathrm{f}}^{\rho} \left(  \nabla_{\mu}  V_{\rho} + \Gamma_{\mu \sigma}^{\nu} p_{\nu} \frac{\partial}{\partial p_{\sigma}} V_{\rho} + V^{\nu} R\indices{_{\mu\nu\rho}^{\sigma}} p_{\sigma} \right) = 0 \; , \qquad
	\end{eqnarray}
where we have used \eqref{Eq:FreeHExterior}.}

{\ The above conditions \eqref{Eq:XHFlowLiouvilleNorm} are still quite loose. One additional condition that we can impose to narrow down the choice is to require the Liouville quasinormal be a vertical vector. This is a well-defined restriction (see appendix \ref{appendix:horizontalandvertical}) and makes the Liouville quasinormal independent of any connection we impose on the base Aristotelian manifold $M$. Another motivation for such a restriction is that the Hamiltonian vector field of the free Hamiltonian \eqref{Eq:FreeHamiltoniansSimplify} kills anything which is horizontal i.e.
	\begin{eqnarray}
		\mathrm{d}H_{\mathrm{f}}[D_{\mu}] = 0 \; . 
	\end{eqnarray}
This aligns well with what is known from the relativistic case as the relevant normal vector \eqref{Eq:RelativisticNormal} in that situation also lacks a horizontal component; there are a large number of equivalent choices and this one seems reasonable, applies generically and most importantly is convenient.}

{\ By restricting our Liouville quasinormal to be vertical, we now have
	\begin{eqnarray}
		 X_{\mathrm{f}}^{\mu} X_{\mathrm{f}}^{\rho} \left(  \nabla_{\mu}  V_{\rho} + \Gamma_{\mu \sigma}^{\nu} p_{\nu} \frac{\partial}{\partial p_{\sigma}} V_{\rho} \right)  = 0 \; . 
	\end{eqnarray}
Again, this is quite a loose constraint. Consequently, for further restrictions on $V$, we can consider spacetimes with infinitesimal symmetries which we wish to preserve on the Hamiltonian level-sets. This leads to the next proposition:}

\begin{proposition}[Level sets of the free Hamiltonians and Killing vector fields]\label{lem:hypersurfaceliouville}
{\ Let $\xi$ be an Aristotelian Killing vector field \eqref{Eq:Killingconditions} on $M$ and $\hat{\xi}$ be the uplift of this Killing vector to $T^{*}M$. The uplift is tangent to the level sets ($H_{\mathrm{f}}=c$ hypersurfaces) of any $H_{\mathrm{f}}$.  Consequently, if the Liouville quasinormal additionally satisfies
	\begin{eqnarray}
		\label{Eq:TangencyConditionV}
		\mathrm{d} H_{\mathrm{f}}[\mathcal{L}_{\hat{\xi}} V] = 0 \; , 
	\end{eqnarray}
i.e. the Lie derivative is totally tangent, then the volume form on any constant free Hamiltonian level set \eqref{Eq:ConstantHVolForm} is preserved along integral curves of $\hat{\xi}$.}
\end{proposition}

\begin{proof}
{\ For a vector field to be tangent to a $H=c$ hypersurface we require that
			\begin{eqnarray}
					dH [ \hat{\xi} ]
				&=& \left\{ H , F \right\}  = \frac{\partial H}{\partial p_{\mu}} p_{\nu} \frac{\partial \xi^{\nu}}{\partial x^{\mu}} - \frac{\partial H}{\partial x^{\mu}} \xi^{\mu} = 0 \;, \nonumber
			\end{eqnarray}
where we have used the Poisson bracket given in definition \ref{def:Poissonbracket}. For our free Hamiltonians
			\begin{eqnarray}
					dH_{\mathrm{f}} [ \hat{\xi} ]&=& \frac{\partial H_{\mathrm{f}}}{\partial (p_{\sigma} \nu^{\sigma})} (\mathcal{L}_{\xi} \nu)^{\mu} +  \frac{\partial H_{\mathrm{f}}}{\partial  (p_{\rho} h^{\rho \sigma} p_{\sigma})} p_{\mu} (\mathcal{L}_{\xi} h)^{\mu \nu} p_{\nu} \; .  \nonumber
			\end{eqnarray}
Consequently, whenever $\xi$ is a Killing vector \eqref{Eq:Killingconditions}, we have that $\hat{\xi}$ is tangent to the $H_{\mathrm{f}}=c$ hypersurfaces.}

{\ For the volume form on the level sets of the Hamiltonian we modify our earlier proof of lemma \ref{lem:LiouvilleonH} replacing $V$ by $\hat{\xi}$.  Upon using the fact that the canonical volume form on $T^{*}M$ is conserved along $\hat{\xi}$ according to point 2 of proposition \ref{prop:UpliftKilling}, and our constraint on total tangency of the Lie derivative of $V$ in the direction $\hat{\xi}$ \eqref{Eq:TangencyConditionV}, we see that the volume form is preserved along integral curves of $\hat{\xi}$.}
\end{proof}

{\ Let us compute the flow of our vertical Liouville quasinormal along a Killing vector field. We find,
	\begin{eqnarray}
			\mathcal{L}_{\hat{\xi}} V
		&=&  \left( \xi^{\mu} \partial_{\mu} V_{\nu} -  p_{\alpha} \frac{\partial \xi^{\alpha}}{\partial x^{\mu}} \frac{\partial}{\partial p_{\mu}} V_{\nu} + V_{\mu} \frac{\partial \xi^{\mu}}{\partial x^{\nu}}  \right) \frac{\partial}{\partial p_{\nu}} \; . 
	\end{eqnarray}
To fulfil our condition \eqref{Eq:TangencyConditionV} we must then impose that
	\begin{eqnarray}
		X^{\nu}_{\mathrm{f}} \left( \xi^{\mu} \partial_{\mu} V_{\nu} + V_{\mu} \partial_{\nu} \xi^{\mu} -  p_{\alpha} \frac{\partial \xi^{\alpha}}{\partial x^{\mu}} \frac{\partial}{\partial p_{\mu}} V_{\nu}  \right) = 0 \; . 
	\end{eqnarray}
We want this to be satisfied independently of the form of the free Hamiltonian (which can depend on a large range of parameters) and also on the particular symmetries of the base Aristotelian manifold. While far from a proof, this is strongly suggestive that the vector $V$ should be constructed from a basis of the free invariant vectors of \eqref{Eq:PrimitiveInvariantsVectors} with coefficients that depend on the free scalar invariants \eqref{Eq:PrimitiveInvariantsScalars}. Moreover, these invariants, belonging to the vertical space, are independent of the arbitrary choice of connection. While there are two such invariants, we note that in the relativistic case one has
	\begin{eqnarray}
		V = \frac{1}{m} p_{\mu} \frac{\partial}{\partial p_{\mu}} \; , 
	\end{eqnarray}
which is independent of the metric and built only from the second of the primitive vectors. As we want to trivially match onto this case, and use the minimal amount of geometric information, we restrict our vertical Liouville quasinormals to be of the form
	\begin{eqnarray}
		V =v\left( \nu^{\mu} p_{\mu},p_{\mu} \tilde{h}^{\mu \nu} p_{\nu} \right) p_{\mu} \frac{\partial}{\partial p_{\mu}}
	\end{eqnarray}
which leaves us with one arbitrary function of the scalar invariants, $v$, to fix.}

{\ So far, most of the restrictions we have imposed have been geometric - now we request one that follows from conventions on the definition of Dirac delta distributions. We first note that there is a way to define $\delta$-function distributions without invoking a metric, in particular, let $g$ be a test function and consider the case where there is a single $H=0$ level set (the generalisation to more will be clear). The $\delta$-function distribution (in general but we shall use the cotangent bundle as a placeholder) can be defined in terms of an integral over the $H=0$ level set as
	\begin{eqnarray}
		\int_{T^{*}M} \mathrm{vol}_{T^{*}M} \; \delta(H) g  := \int_{H^{-1}(0)} \iota^{*}(g \omega) \; , 
	\end{eqnarray}
where $\omega$ is the Gelfand-Leray form which can always be defined locally \cite{Arnold2012-it} by
	\begin{eqnarray}
		 \mathrm{vol}_{T^{*}M} &=& \mathrm{d} H \wedge \omega
	\end{eqnarray}
and $\iota$ is the inclusion map of the level set $H^{-1}(0)$ in $T^{*}M$. Typically one chooses $\omega = i_{X} \mathrm{vol}_{T^{*}M}$ where $X$ is any vector field satisfying $dH[X]=1$. In cases where there is a metric, the usual choice is something like
	\begin{eqnarray}
		X \sim \frac{\nabla H}{\left\| \nabla H \right\|^2 } 	\; . 
	\end{eqnarray} 
In our more general setting, using the normalisation $dH[X]=1$, we can establish the identity
	\begin{eqnarray}
			i_{X} \left( dH \wedge i_{X} \mathrm{vol}_{T^{*}M} \right)
		&=& dH(X) i_{X}  \mathrm{vol}_{T^{*}M} + dH \wedge i_{X} i_{X}  \mathrm{vol}_{T^{*}M} \nonumber \\
		&=&  dH(X) i_{X}  \mathrm{vol}_{T^{*}M} = i_{X}  \mathrm{vol}_{T^{*}M} \; . 
	\end{eqnarray}
In particular, upon rearranging, we have
	\begin{eqnarray}
		i_{X} \left( dH \wedge i_{X} \mathrm{vol}_{T^{*}M} - \mathrm{vol}_{T^{*}M} \right) = 0 
	\end{eqnarray}
for every vector $X$ which contains a component which is not tangent to $H^{-1}(0)$. However, notice that this is a difference of $2(d+1)$-forms, and is thus also a $2(d+1)$-form, or top form, on $T^{*}M$. The term in parentheses takes $2(d+1)$-arguments and is defined on a $2(d+1)$-dimensional spacetime. The contraction of a non-zero $2(d+1)$-form with a generic vector, including those not entirely tangent to the surface, cannot be zero. The only way for the term in parentheses to be zero when contracted with a generic $X$ is if the relevant term is identically zero. Thus we establish that
	\begin{eqnarray}
		dH \wedge i_{X} \mathrm{vol}_{T^{*}M} = \mathrm{vol}_{T^{*}M} \; . 
	\end{eqnarray}
Given these conventions, we now impose that $\mathrm{d} H_{\mathrm{f}}[V]=1$, where $V$ is our vertical Liouville quasinormal, to match what is conventional in the literature. This introduces a normalisation for $V$,
	\begin{subequations}
	\label{Eq:Vnormalisation}
	\begin{eqnarray}
		\mathrm{d} H_{\mathrm{f}}[V] &=& \frac{\partial H_{\mathrm{f}}}{\partial p_{\mu}} Dp_{\mu} \left[ v p_{\nu} \frac{\partial}{\partial p_{\nu}} \right]
					 = v \frac{\partial H_{\mathrm{f}}}{\partial p_{\mu}} p_{\mu} = 1 \; , \\
		\Rightarrow v &=& \frac{1}{\frac{\partial H_{\mathrm{f}}}{\partial p_{\mu}} p_{\mu}} \; , 
	\end{eqnarray}
	\end{subequations}
where we assume that $\frac{\partial H_{\mathrm{f}}}{\partial p_{\mu}} p_{\mu} \neq 0$ so that $v$ is well-defined (also for its pullback under the inclusion map onto the $H_{\mathrm{f}}^{-1}(0)$ hypersurface). Hence we find that
	\begin{eqnarray}
		\label{Eq:ConventionalNormal}
		V &=& \frac{p_{\mu}}{p_{\nu} \frac{\partial H_{\mathrm{f}}}{\partial p_{\nu}}} \frac{\partial}{\partial p_{\mu}} 
		     = \frac{p_{\mu}}{\frac{\partial H_{\mathrm{f}}}{\partial (p_{\sigma} \nu^{\sigma})} p_{\nu} \nu^{\nu} + 2 \frac{\partial H_{\mathrm{f}}}{\partial (p_{\sigma} h^{\sigma \rho} p_{\rho})} p_{\nu} \tilde{h}^{\nu \lambda} p_{\lambda}}  \frac{\partial}{\partial p_{\mu}} \; , 
	\end{eqnarray}
and it subsequently follows that
	\begin{eqnarray}
		\int_{T^{*}M} \mathrm{vol}_{T^{*}M} \; \delta(H_{\mathrm{f}}) g  := \int_{H_{\mathrm{f}}^{-1}(0)} \mathrm{vol}_{\mathrm{f}} \; \iota^{*}(g)  \; . 
	\end{eqnarray}
This is of course rather abstract and we shall apply it to a simple case below.}

{\ We start by considering the canonical volume form on $T^{*}M$ given in local coordinates in \eqref{Eq:canvolformcoords}. Employing the single dispersion relation Hamiltonians of \eqref{eq:HSingleDispersion}, in coordinates adapted to the hypersurface \eqref{Eq:HypersurfaceCoords}, we find the conventional Liouville quasinormal of \eqref{Eq:ConventionalNormal} takes the form
	\begin{eqnarray}
		V_{*} &=& \frac{1}{- p_{0} + 2 \vec{p}^2 \frac{\partial \tilde{H}(\vec{p}^2)}{\partial \vec{p}^2}} \left[ p_{0} \frac{\partial}{\partial p_{0}} + \vec{p} \cdot \frac{\partial}{\partial \vec{p}} \right] \; , 
	\end{eqnarray}
where we have without loss of generality set $\lambda =1$. Consequently, the volume form on $H_{*}^{-1}(0)$ is the $(2d+1)$-form,
	\begin{eqnarray}
		\label{Eq:PullbackDeltaMeasure}
		\mathrm{vol}_{H_{*}^{-1}(0)} &=& dp_{1} \wedge \ldots \wedge dp_{d} \wedge \mathrm{d}x^{0} \wedge \ldots \wedge \mathrm{d}x^{d} \; , 
	\end{eqnarray}
where we have pulled back onto the surface using \eqref{Eq:p0sub}. Let $g(p_{0},\vec{p})$ be a test function, then we find the following expression
	\begin{eqnarray}
			\label{Eq:SimpleDeltaIdentity}
		&\;& - \int_{T^{*}M} \; dp_{0} \wedge \ldots \wedge dp_{d} \wedge \mathrm{d}x^{0} \wedge \ldots \wedge \mathrm{d}x^{d} \; \delta(p_{0}-\tilde{H}(\vec{p}^2)) g(p_{0},\vec{p}) \nonumber \\
		&\;& \qquad = \int_{H_{\mathrm{f}}^{-1}(0)} \; dp_{1} \wedge \ldots \wedge dp_{d} \wedge \mathrm{d}x^{0} \wedge \ldots \wedge \mathrm{d}x^{d} \; g(\tilde{H},\vec{p})  \; ,
	\end{eqnarray}
exactly as we might have desired.}

{\ The Liouville quasinormal is of course a naturally ambiguous object; but this is a problem of integral measures in general. In the case of a Galilean boost-invariant theory we identify one among a class of (boost-invariant) integral measures but even this measure is not more privileged than others as one can determine by multiplying the measure by an arbitrary function of the invariant $\vec{p}^2$. The relativistic case is different as the only boost invariant measure is $d^{d}p/\left\|\vec{p}\right\|$. What matters is that this choice is shared by all  the models we consider. Because of the nice properties we have discussed here we henceforth take $V$ to be as in \eqref{Eq:ConventionalNormal}. This vector is covariant, vertical, a Liouville quasinormal, reproduces the canonical normal vector of a relativistic theory in the right limit and is expressed entirely in terms of quantities that are defined for all free Hamiltonians\footnote{With the notable single exception of Hamiltonians where $p_{\mu} \frac{\partial H_{\mathrm{f}}}{\partial p_{\mu}} = 0$.}.}

\subsection{Current conservation from Liouville's theorem }\label{sec:PhaseSpaceChargeCons}

{\noindent So far in our development we have uplifted symmetry generators (see section \ref{sec:KillingUplift}) and closed surfaces (section \ref{sec:embedding}) from the manifold up to the cotangent bundle. We make our final uplifts here and show how to lift conservation of charge, and invariance of the currents under Killling symmetries, from the base Aristotelian manifold to the level sets of $H_{\mathrm{f}}$. This involves the following steps:
	\begin{enumerate}
		\item \label{item:def} defining an analogue of the particle current on phase space that is tangent to constant energy hypersurfaces,
		\item \label{item:integral} showing that this phase space current is conserved under Hamiltonian flow - in particular, that a certain integral of this current over a closed surface embedded in co-dimension one, constant energy hypersurfaces of the phase space is identically zero,
		\item \label{item:project} showing how this surface integral can be rewritten as a surface integral in the base Aristotelian manifold, which thus also vanishes by the previous point.
	\end{enumerate}
Once we reach this last point, we have a vector field in phase space whose surface integral is identically zero independently of the shape of the surface. Therefore from our work in section \ref{sec:aristotlechargecons} this vector field satisfies a conservation equation.}

{\ Concerning item \ref{item:def}, let us define the vector field
    \begin{eqnarray}
    	\label{Eq:PhaseSpaceCurrent}
        \mathcal{J}_{\mathrm{f}} &=& f X_{\mathrm{f}} \; , 
    \end{eqnarray}
which will eventually be the uplift of the charge current on the Aristotelian manifold. Here $f$ is at this point an arbitrary function, but will become the one-particle distribution function when we move to kinetic theory. For $\mathcal{J}$ to be tangent to the level sets $H_{\mathrm{f}}^{-1}(c)$ we must require that
	\begin{eqnarray}
		\label{Eq:TangencyJ}
		\mathrm{d} H_{\mathrm{f}}[\mathcal{J}_{\mathrm{f}}] &=& X_{\mathrm{f}}(f) X_{\mathrm{f}} = 0 \; . 
	\end{eqnarray}
This represents a restriction on $f$ to satisfy $X_{\mathrm{f}}(f)=0$; for a free Hamiltonian. For example, any scalar function of our free scalar invariants \eqref{Eq:PrimitiveInvariantsScalars} trivially satisfies such a condition. In particular, free particles flow along $X_{\mathrm{f}}$ in phase space and $f$ can be chosen to measure the probability of a particle following such a trajectory at a given point in $H_{\mathrm{f}}^{-1}(c)$. This completes task \ref{item:def}, we have defined a current on phase space which will nominally represent the particle density.}

{\ Let $\mathcal{K}$ be any vector field tangent to a constant energy hypersurface and $\Sigma' \subset H_{\mathrm{f}}^{-1}(c)$ a closed subset. Consider the following integral
    \begin{eqnarray}
           \int_{\Sigma'} \; \mathcal{L}_{\mathcal{K}} \mathrm{vol}_{\mathrm{f}} \; ,
    \end{eqnarray}
where, as we discussed previously, the Lie derivative along a vector tangent to a level set of a tensor whose entries are in the corresponding hypersurface is also in the surface. Thus, $\mathcal{L}_{\mathcal{K}} \mathrm{vol}_{H^{-1}(c)}$ is a form on $H^{-1}(c)$. Using Cartan's magic formula \eqref{Eq:CartanMagic} we have
    \begin{eqnarray}
            \mathcal{L}_{\mathcal{K}} \mathrm{vol}_{\mathrm{f}}
        &=& \mathrm{d} \left( i_{\mathcal{K}}  \mathrm{vol}_{\mathrm{f}} \right) - i_{\mathcal{K}} \mathrm{d}  \mathrm{vol}_{\mathrm{f}} 
        	  = \mathrm{d} \left( i_{\mathcal{K}}  \mathrm{vol}_{\mathrm{f}} \right)  \; ,
    \end{eqnarray}
where the second equality follows because $\mathrm{vol}_{\mathrm{f}}$ is a top form on $\Sigma'$ and thus its exterior derivative is zero. Consequently
        \begin{eqnarray}
           \int_{\Sigma'} \; \mathcal{L}_{\mathcal{K}}  \mathrm{vol}_{\mathrm{f}} = \int_{\Sigma'} \; \mathrm{d} \left( i_{\mathcal{K}}  \mathrm{vol}_{\mathrm{f}} \right) = \int_{\partial \Sigma'} \; \iota_{\Sigma'}^{*}(i_{\mathcal{K}}  \mathrm{vol}_{\mathrm{f}}) \; , 
    \end{eqnarray}
where the last line follows from the generalised Stoke's theorem \eqref{Eq:Stokes}, $\partial \Sigma'$ is the boundary of $\Sigma'$ and $\iota_{\Sigma'}^{*}$ is the pullback of the inclusion map
	\begin{eqnarray}
		\iota_{\Sigma'}: \partial \Sigma' \rightarrow \Sigma' \; .
	\end{eqnarray}
The pullback of the inclusion map is often implicitly assumed, but we have made it explicit. Thus, independently of the form of $\mathcal{K}$ we have
	\begin{eqnarray}
		\label{Eq:GenericIntegral}
		   \int_{\Sigma'} \; \mathcal{L}_{\mathcal{K}} \mathrm{vol}_{\mathrm{f}} =  \int_{\partial \Sigma'} \; \iota_{\Sigma'}^{*}(i_{\mathcal{K}}  \mathrm{vol}_{\mathrm{f}}) \; . 
	\end{eqnarray}
Now, when we identify $\mathcal{K}$ with $\mathcal{J}_{\mathrm{f}}$ we see that
	\begin{eqnarray}
	   \label{Eq:ConstancyVolunderJ}
            \mathcal{L}_{\mathcal{J}_{\mathrm{f}}} \mathrm{vol}_{\mathrm{f}} = \mathcal{L}_{(f X_{\mathrm{f}})} \mathrm{vol}_{\mathrm{f}} = f \mathcal{L}_{X_{\mathrm{f}}} \mathrm{vol}_{\mathrm{f}} + X_{\mathrm{f}}(f) \mathrm{vol}_{\mathrm{f}} = 0 \; ,
    \end{eqnarray}
where the final equality follows by Liouville's theorem developed in the previous section (lemma \ref{lem:hypersurfaceliouville}) and our constraint that $\mathcal{J}_{\mathrm{f}}$ be tangent to the $H_{\mathrm{f}}^{-1}(c)$ \eqref{Eq:TangencyJ}. Thus \eqref{Eq:GenericIntegral} becomes
    \begin{eqnarray}
     \label{Eq:VanishingCurrentIntegral}
     \int_{\partial \Sigma'} \; i_{\mathcal{J}_{\mathrm{f}}} \mathrm{vol}_{\mathrm{f}} &=& 0 \; .
    \end{eqnarray}
 which satisfies our goal in item \ref{item:integral} - i.e. we have constructed a conserved integral in phase space corresponding to $\mathcal{J}_{\mathrm{f}}$.}

{\ A natural question is what corresponds to $\mathcal{J}_{\mathrm{f}}$ on the base manifold $M$. This is the content of point \ref{item:project}. To figure this out we are going to have to integrate over the momentum variables to leave an object defined only on $M$, i.e. in terms of configuration space variables.}

{\ Now rather than $\Sigma'$ being a generic closed submanifold in $H^{-1}_{\mathrm{f}}(c)$ we must restrict. In particular, let $S$ be a closed submanifold of $M$. With a minor generalisation of the discussion in section \ref{sec:embedding} we can lift $S$ to a closed submanifold of $H_{\mathrm{f}}^{-1}(c)$ if we modify the definition \eqref{Eq:EmbeddingPhaseSpace} so that $p$ belongs only to relevant constant Hamiltonian level set. Now $\Sigma'$ refers to any such uplift and its boundary be $\partial \Sigma'$. From the discussion in section \ref{sec:embedding}, if $n$ is the normal form defining $\partial S$, then $N=n_{\mu} \mathrm{d}x^{\mu}$ is the normal form defining the $2d$-dimensional surface $\partial \Sigma'$ bounding the $2d+1$-dimensional submanifold $\Sigma'$ of $H^{-1}(c)$.}
 
{\ Having defined the region of the base manifold $M$ that $\Sigma'$ corresponds to, we can now return to trying to integrate out the momentum degrees of freedom in $\mathcal{J}_{\mathrm{f}}$. As $N$ is the normal form to $\partial \Sigma'$ we can write the volume forms in a patch near $\partial \Sigma'$ as
    \begin{eqnarray}
    	\label{Eq:DecomposedVolumeForm}
    	\mathrm{vol}_{\mathrm{f}} &=& N \wedge i_{X} \mathrm{vol}_{\mathrm{f}}  \; ,
    \end{eqnarray}
where $i_{X} \mathrm{vol}_{\mathrm{f}}$ is the Gelfand-Leray form associated with some vector $X$ that satisfies $N[X]=1$. Additionally, on this local patch we shall work in a coordinate system where 
	\begin{eqnarray}
	\label{Eq:VolumeFormDecomp}
	 \mathrm{vol}_{\mathrm{f}} &=& N \wedge \sigma_{N} \wedge \mathrm{vol}_{P}
	\end{eqnarray}
where  $\sigma_{N}$ is built only from horizontal forms i.e. $dx^{\mu}$. Consequently, $\mathrm{vol}_{P}$ parameterises the part of the volume form $\mathrm{vol}_{\mathrm{f}}$ involving integration over generalised momentum restricted to the level-set. Once we are allowed to decompose the volume form as \eqref{Eq:DecomposedVolumeForm} we can rewrite $\iota^{*}(i_{\mathcal{J}} \mathrm{vol}_{\mathrm{f}})$ as
	\begin{subequations}
    	\begin{align}
        		& \iota^{*}( i_{\mathcal{J}} \mathrm{vol}_{\mathrm{f}} ) = N(\mathcal{J})  \iota^{*} (\sigma_{N}) \wedge \iota^{*}(\mathrm{vol}_{P} )  = N(\mathcal{J})  \sigma_{n} \wedge \mathrm{vol}_{p} \; , \qquad \\ 
     & \sigma_{n} =   \iota^{*} (\sigma_{N}) \in \Omega^{d}(M) \; , \qquad  \left. \mathrm{vol}_{p} \right|_{x \in M} = \left. \iota^{*}(\mathrm{vol}_{P} )  \right|_{x \in M} \in \Omega^{d+1}(T_{x}^{*}M) \; , 
    \end{align}
    \end{subequations}
where we have assumed that our inclusion map is such that the pullback of $\mathrm{vol}_{\mathrm{f}}$ can be decomposed into the wedge product of a $d$-dimensional form on $M$ and a $(d+1)$-dimensional form $\mathrm{vol}_{p}$ that at any point $x \in M$ belongs to the forms on cotangent space at $x$.  In deriving this result we used that the pullback of the inclusion map on any normal form vanishes and also that the pullback distributes across forms. Consequently,
    \begin{eqnarray}
         \int_{\partial \Sigma'} \; i_{\mathcal{J}} \mathrm{vol}_{\mathrm{f}}
     &=& \int_{\partial S} \left[ \int_{\left. T_{x}^{*}M \right|_{\mathrm{f}}}  \mathrm{vol}_{p} \; N[\mathcal{J}] \right] \sigma_{n}
    \end{eqnarray}
where $\left. T_{x}^{*}M \right|_{\mathrm{f}}$ should be interpreted as the submanifold of $\Sigma'$ that consists only of momenta in a given level set of $H_{\mathrm{f}}$.}

{\ Having allowed ourselves a volume form that separates out the momentum degrees of freedom, given any tangent vector $\mathcal{K}$  to $H_{\mathrm{f}}^{-1}(c)$ we can identify a corresponding current $K$ which, being defined on $M$, is dependent only on base manifold coordinates. This current $K$ is given by 
    \begin{eqnarray}
    	\label{Eq:Definitionofcurrents}
        n[K] &=&  \int_{\left. T_{x}^{*}M \right|_{\mathrm{f}}} \mathrm{vol}_{p} \; N[\mathcal{K}] 
    \end{eqnarray}
for any surface $S$. As this is true for any surface we can then identify
	\begin{eqnarray}
		\label{Eq:DefofK}
		K &=&   \int_{\left. T_{x}^{*}M \right|_{\mathrm{f}}} \mathrm{vol}_{p} \; \mathcal{K} \; ,
	\end{eqnarray}
up to matching across coordinate patches. It follows that 
	\begin{eqnarray}
     	    		 \int_{\Sigma'} \; \mathcal{L}_{\mathcal{K}} \mathrm{vol}_{\mathrm{f}}
    		  = \int_{\partial \Sigma'} \; i_{\mathcal{K}} \mathrm{vol}_{\mathrm{f}} = \int_{\partial S} n[K] \sigma_{n} 
     		 =  \int_{\partial S} \iota_{S}^{*} (i_{K} \mathrm{vol}_{M}) \; , 
	\end{eqnarray}
for any vector $\mathcal{K}$ tangent to the $H_{\mathrm{f}}^{-1}(c)$ hypersurface. Now once again we identify $\mathcal{K}$ with $\mathcal{J}_{\mathrm{f}}$ and use \eqref{Eq:ConstancyVolunderJ} to show that
	\begin{eqnarray}
		 \label{Eq:Identificationofconscharges}
		  \int_{\partial S} \iota_{S}^{*} (i_{J} \mathrm{vol}_{M}) &=&   0
	\end{eqnarray}
where $J$ is defined by replacing $\mathcal{K} \rightarrow \mathcal{J}_{\mathrm{f}}$. This completes our aim of identifying the current in $M$ corresponding to $\mathcal{J}_{\mathrm{f}}$ and showing that a particular surface integral of this current is conserved. Thus we have achieved item \ref{item:project}.}

{\  From our work in section \ref{sec:aristotlechargecons}, we know that given any vector $J$ and closed submanifold $S$, the following integrals are equivalent
	\begin{eqnarray}
		\label{Eq:StokesEquality}
		 \int_{\partial S} \iota_{S}^{*} (i_{J} \mathrm{vol}_{M})  = \int_{S} \mathrm{vol}_{M} \left( \nabla_{\mu} J^{\mu} - \Sigma\indices{_{\mu \nu}^{\mu}} J^{\mu} \right) 	
	\end{eqnarray}
Moreover, as we have shown in \eqref{Eq:Identificationofconscharges}, the left hand-side vanishes for any $J$ defined in terms of $\mathcal{J}_{\mathrm{f}}$ by \eqref{Eq:DefofK}, independently of the closed surface $S$ and whichever among the functions satisfying $X_{\mathrm{f}}(f)=0$ that we are considering. Thus, it follows from \eqref{Eq:StokesEquality} that
	\begin{eqnarray}
		\nabla_{\mu} J^{\mu} - \Sigma\indices{_{\mu \nu}^{\mu}} J^{\nu} \equiv 0 \; ,
	\end{eqnarray}
i.e. the current $J$ is conserved. Consequently, we can now think of the various $\mathcal{J}_{\mathrm{f}}$, differing by choices of $f$, as the class of uplifts of a conserved current $J$ from the base Aristotelian manifold $M$ to $H_{\mathrm{f}}^{-1}(c)$. This completes items \ref{item:def} to \ref{item:project}.}

{\ We now turn to generalising the above result from the number current to higher currents. To generalise to the SEM tensor complex, and higher index currents in the free case, notice that we can construct a vector current defined on a level set of $H_{\mathrm{f}}$ in the following manner
	\begin{eqnarray}
		\mathcal{J}_{s}(X_{1},\ldots,X_{s}) = p(X_{1}) p(X_{2}) \ldots p(X_{s}) f X_{\mathrm{f}}
	\end{eqnarray}
where $X_{1}$, \ldots $X_{s}$ are arbitrary vector fields on a given tangent space to $M$. The Lie derivative of the volume form along $\mathcal{J}_{s}$ is given by
	\begin{eqnarray}
			\mathcal{L}_{\mathcal{J}_{s}} \mathrm{vol}_{\mathrm{f}}
		&=& X_{\mathrm{f}}[p(X_{1}) p(X_{2}) \ldots p(X_{s}) f] \mathrm{vol}_{\mathrm{f}} \nonumber \\
		&=& \left[ \vphantom{ \sum_{i=1}^{i=s}} p_{\nu_{1}} X_{1}^{\nu_{1}} \ldots p_{\nu_{s}} X_{s}^{\nu_{s}} X_{\mathrm{f}}(f) \right. \\
		&\;& \left. + \sum_{i=1}^{i=s} ( p_{\nu_{1}} X_{1}^{\nu_{1}} ) \ldots  ( p_{\nu_{i}} X_{\mathrm{f}}^{\mu} \nabla_{\mu} X_{i}^{\nu_{i}} )  \ldots ( p_{\nu_{s}} X_{s}^{\nu_{s}} f ) \right] \mathrm{vol}_{\mathrm{f}}\; ,   \qquad \nonumber
	\end{eqnarray}
in a local coordinate system where we have used
	\begin{eqnarray}
		\mathcal{L}_{X_{\mathrm{f}}} \left( p_{\mu} X^{\mu} \right) = X_{\mathrm{f}}^{\nu} D_{\nu} \left( p_{\mu} X^{\mu} \right)
		= X_{\mathrm{f}}^{\nu} p_{\mu} \nabla_{\nu} X^{\mu} \; .
	\end{eqnarray}
Assuming that once again $X_{\mathrm{f}}(f)=0$, integrating the above over the uplift $\Sigma'$ of some closed manifold $S \subset M$ and applying Cartan's magic formula we arrive at
	\begin{eqnarray}
		&\;&  \int_{\partial \Sigma'} \iota_{\Sigma'}^{*} i_{\mathcal{J}_{s}[X_{1},\ldots,X_{s}]}  \mathrm{vol}_{\mathrm{f}} \nonumber \\
		&-& \int_{\Sigma'}  \mathrm{vol}_{\mathrm{f}}  \left[ \sum_{i=1}^{i=s} \mathcal{J}\indices{_{(s)}^{\mu}_{\nu_{1} \ldots \nu_{s}}} X_{1}^{\nu_{1}} \ldots \nabla_{\mu} X^{\nu_{i}}  \ldots X_{s}^{\nu_{s}} \right]  =0\; , \\
		&\;& \mathcal{J}\indices{_{(s)}^{\mu}_{\nu_{1} \ldots \nu_{s}}}
		= p_{\nu_{1}} \ldots p_{\nu_{s}} X^{\mu}_{\mathrm{f}} f \; , 
	\end{eqnarray}
which is the generalisation of \eqref{Eq:VanishingCurrentIntegral} for $s>0$. When we considered the $s=0$ current the integral over $\Sigma'$ was identically zero, now there is a remainder coming from the vector entries of $\mathcal{J}_{s}$. Using the decomposition of the volume form, and identifying
	\begin{eqnarray}
		J\indices{_{(s)}^{\mu}_{\nu_{1} \ldots \nu_{s}}} 
		&=& \int_{_{\left. T_{x}^{*}M \right|_{\mathrm{f}}}} \mathrm{vol}_{p} \; \mathcal{J}\indices{_{(s)}^{\mu}_{\nu_{1} \ldots \nu_{s}}} 
	\end{eqnarray}
we arrive at the following relation
	\begin{eqnarray}
		\label{Eq:ZeroBoundaryHigherS}
		&\;& \int_{\partial S} \; \iota_{\partial S}^{*} i_{J[X_{1},\ldots,X_{s}]} \mathrm{vol}_{M}
		- \int_{S} \;  \left[ \sum_{i=1}^{i=s} J\indices{_{(s)}^{\mu}_{\nu_{1} \ldots \nu_{s}}} X_{1}^{\nu_{1}} \ldots \nabla_{\mu} X_{i}^{\nu_{i}} \ldots X_{s}^{\nu_{s}} \right] \mathrm{vol}_{M} \nonumber \\
		&\;& = 0 \; .  
	\end{eqnarray}
which is satisfied for any choice of $f$ such that $X_{\mathrm{f}}(f)=0$, any set of vector fields $X$ belonging to the tangent space of the manifold and closed surface $S \subset M$. This should be compared to \eqref{Eq:RelatingchargesanddivergenceX} where $s=1$.}

{\ Now we consider the generalisation of \eqref{Eq:StokesEquality} to find
\begin{eqnarray}
		\label{Eq:StokesEqualitys}
		 \int_{\partial S} \iota_{\partial S}^{*} (i_{J[X_{1},\ldots,X_{s}]} \mathrm{vol}_{M})  
		 &=& \int_{S} \mathrm{vol}_{M} \left( \nabla_{\mu} (J\indices{_{(s)}^{\mu}_{\nu_{1} \ldots \nu_{s}}} X_{1}^{\nu_{1}} \ldots X_{s}^{\nu_{s}}) \right. \nonumber \\
		  &\;& \left. \hphantom{ \int_{S} \mathrm{vol}_{M} \left(\right.} - \Sigma_{\mu \nu}^{\nu} J\indices{_{(s)}^{\mu}_{\nu_{1} \ldots \nu_{s}}} X_{1}^{\nu_{1}} \ldots X_{s}^{\nu_{s}} \right) 	
	\end{eqnarray}
Rearranging this expression we arrive at
	\begin{eqnarray}
		 &\;& \int_{\partial S} \; \iota_{S}^{*} i_{J[X_{1},\ldots,X_{s}]} \mathrm{vol}_{M}
		- \int_{S} \;  \left[ \sum_{i=1}^{i=s} J\indices{_{(s)}^{\mu}_{\nu_{1} \ldots \nu_{s}}} X_{1}^{\nu_{1}} \ldots  \nabla_{\mu} X_{i}^{\nu_{i}}   \ldots X_{s}^{\nu_{s}} \right] \mathrm{vol}_{M} \nonumber \\
		 &=& \int_{S} \mathrm{vol}_{M} \left( \nabla_{\mu} J\indices{_{(s)}^{\mu}_{\nu_{1} \ldots \nu_{s}}} - \Sigma\indices{_{\mu \nu}^{\mu}} J\indices{_{(s)}^{\mu}_{\nu_{1} \ldots \nu_{s}}}  \right) X_{1}^{\nu_{1}} \ldots X_{s}^{\nu_{s}} \; , 
	\end{eqnarray}
If we apply \eqref{Eq:ZeroBoundaryHigherS} then the left hand side of the equality is zero and we have a set of conserved currents
	\begin{eqnarray}
		\nabla_{\mu} J\indices{_{(s)}^{\mu}_{\nu_{1} \ldots \nu_{s}}} - \Sigma\indices{_{\nu \mu}^{\nu}} J\indices{_{(s)}^{\mu}_{\nu_{1} \ldots \nu_{s}}}  = 0 \; ,
	\end{eqnarray}
for any free Hamiltonian.}

{\ We now turn to our last mathematical result, that invariance of $f$ under the uplift of a Killing vector implies that the base manifold currents are invariant under the Lie derivative along the original Killing vector.}

\begin{lemma}[Invariance of the currents under Killing symmetries]\label{lem:conservationunderuplift}.
{\ Let $\xi$ be a Killing vector and $\hat{\xi}$ its uplift, moreover let the decomposition of the constant hypersurface volume form \eqref{Eq:VolumeFormDecomp} hold. It follows that
	\begin{eqnarray}
		\mathcal{L}_{\xi} \mathcal{J}\indices{_{s}^{\mu}_{\nu_{1} \ldots \nu_{s}}} &=& \int_{\left. T_{x}^{*}M \right|_{\mathrm{f}}} \mathrm{vol}_{p} \; X_{\mathrm{f}}^{\mu} p_{\nu_{1}} \ldots p_{\nu_{s}} \mathcal{L}_{\hat{\xi}} f
	\end{eqnarray}
for any higher current derived from a free Hamiltonian.}
\end{lemma}

\begin{proof}
{\ We first remind ourselves that $\hat{\xi}$ is tangent to the level sets of $H_{\mathrm{f}}$ and that the volume form we constructed, $\mathrm{vol}_{H^{-1}(c)}$ is invariant under the action of the Lie derivative along $\hat{\xi}$.}

{\ Let us first replace $X_{\mathrm{f}}$ in \eqref{Eq:PhaseSpaceCurrent} with $\hat{\xi}$, so that our current is
	\begin{eqnarray}
		\mathcal{J}_{\xi} = f \hat{\xi} \;. \nonumber
	\end{eqnarray}
This new current $\mathcal{J}_{\xi}$ is tangent to the free Hamiltonian level-set if $\hat{\xi}(f)=0$ as can be seen by acting on it with $\mathrm{d} H_{\mathrm{f}}$. Following through the previous steps we arrive at 
	\begin{eqnarray}
		\int_{\left. T_{x}^{*}M \right|_{\mathrm{f}}} \mathrm{vol}_{p} \; f n_{\mu} \hat{\xi}^{\mu} = n_{\mu} \xi^{\mu} F \; , \qquad F = \int_{\left. T_{x}^{*}M \right|_{\mathrm{f}}} \mathrm{vol}_{p} \; f  \; , \nonumber
	\end{eqnarray}
where we have used that $n_{\mu} \hat{\xi}^{\mu} = n_{\mu} \xi^{\mu}$ as $n = N = n_{\mu} \mathrm{d}x^{\mu}$. Consequently, if $\hat{\xi}(f)=0$, then following the remaining steps of our earlier demonstration, one sees that
	\begin{eqnarray}
		\int_{S} \mathrm{vol}_{M} \; \frac{1}{e} \partial_{\mu} \left( e \xi^{\mu} F \right) = 0 \; . 	\nonumber
	\end{eqnarray}
The Killing conditions imply \eqref{Eq:KillingIdentityTrace} and consequently,
	\begin{eqnarray}
		\int_{S} \mathrm{vol}_{M} \; \mathcal{L}_{\xi} F &=& 0 \; . \nonumber
	\end{eqnarray}
For this to be true generally we have that $\mathcal{L}_{\xi} F = 0$ whenever $\hat{\xi}(f)=0$.}

{\ To demonstrate that the current $J$ is invariant under the Killing field when $\hat{\xi}(f) = 0$ we define
	\begin{eqnarray}
		\mathcal{J}_{\mathcal{L}_{\xi}} = f X_{\mathrm{f}}[Y] \hat{\xi} \nonumber
	\end{eqnarray}
where $Y$ is any form defined on the base manifold. The Lie derivative along $\hat{\xi}$ of this current is not zero when $\hat{\xi}(f)=0$, but differs by the Lie derivative of $Y$ i.e.
	\begin{eqnarray}
			\mathcal{L}_{\mathcal{J}_{\mathcal{L}_{\xi}}} \mathrm{vol}_{\mathrm{f}}
		&=& \hat{\xi}(Y_{\mu} X^{\mu}_{\mathrm{f}})  \mathrm{vol}_{\mathrm{f}} 
		= \mathcal{L}_{\xi}(Y_{\mu}) X^{\mu}_{\mathrm{f}} \mathrm{vol}_{\mathrm{f}} \nonumber
	\end{eqnarray}
where we have used $\mathcal{L}_{\hat{\xi}} X_{\mathrm{f}} =0$. Moreover, the corresponding base manifold current to $\mathcal{J}_{\mathcal{L}_{\xi}}$ is
	\begin{eqnarray}
			\int_{\left. T_{x}^{*}M \right|_{\mathrm{f}}} \mathrm{vol}_{p} \; N[\mathcal{J}_{\mathcal{L}_{\xi}}] = \int_{\left. T_{x}^{*}M \right|_{\mathrm{f}}} \mathrm{vol}_{p} \;  f X_{\mathrm{f}}^{\mu} Y_{\mu} n_{\nu} \xi^{\nu} &=& Y_{\mu} n_{\nu} \xi^{\nu} \int_{\left. T_{x}^{*}M \right|_{\mathrm{f}}} \mathrm{vol}_{p} \;  f X_{\mathrm{f}}^{\mu} \nonumber \\
			&=& Y_{\mu} J^{\mu} n [  \xi ] \nonumber
	\end{eqnarray}
We can follow the usual steps noting that
	\begin{eqnarray}
			\int_{\Sigma'} \mathrm{vol}_{M} \; (\mathcal{L}_{\xi} Y)_{\mu} \left( \int_{\left. T_{x}^{*}M \right|_{\mathrm{f}}} \mathrm{vol}_{p} \; X^{\mu}_{\mathrm{f}} f \right)
		&=& \int_{\partial S} \; Y[J] n[\xi] \sigma_{n} = \int_{\partial S} \iota^{*}_{ S}  \left( i_{Y[J] \xi} \mathrm{vol}_{M} \right) \nonumber \\
			\int_{S} \mathrm{vol}_{M} \;  (\mathcal{L}_{\xi} Y)_{\mu} J^{\mu}
		&=& \int_{S} \mathrm{vol}_{M} \; \frac{1}{e} \partial_{\mu} \left( e Y_{\nu} J^{\nu} \xi^{\mu} \right) \nonumber \\
		&=& \int_{S} \mathrm{vol}_{M} \; \left( Y_{\nu} (\mathcal{L}_{\xi} J)^{\mu} + J^{\mu} (\mathcal{L}_{\xi} Y)_{\mu} \right) \; , \nonumber
	\end{eqnarray}
which finally implies
	\begin{eqnarray}
			0
		&=&  \int_{S} \mathrm{vol}_{M} \; Y_{\nu} \mathcal{L}_{\xi} J^{\mu} \; , \nonumber
	\end{eqnarray}
for any $Y_{\mu}$ defined on the base manifold whenever $\mathcal{L}_{\hat{\xi}} (f)=0$. Thus we conclude that $\mathcal{L}_{\xi} J^{\mu} =0 $ as desired.  The higher currents, such as the SEM tensor complex, then follow by replacing $f \rightarrow X_{H}[Y] p[X_{1}] \ldots p[X_{s}] f$ and using the similar steps to our derivation of $J\indices{_{(s)}^{\mu}_{\nu_{1} \ldots \nu_{s}}}$.}
\end{proof}

{\ This latter lemma is a crucial one as it relates space-time symmetries - such as stationarity of the currents - to constraints on the one particle distribution function. It forms a core part of the earlier section \ref{sec:applications}, being as it allows us to establish how the time-like Killing vector which represents the fluid velocity leaves the constitutive relations invariant. In reference to that section, we find that Hamiltonians in \eqref{Eq:GenericCurvedOPDispersion} on flat space led to a conserved charge current and SEM tensor complex of the form
	\begin{subequations}
	\label{Eq:SingleDisperionCurrentSEMIntegral}
	\begin{eqnarray}
		J^{\mu} \partial_{\mu} &=& \left(  \frac{1}{e} \int d^{d}p \; f \right) \partial_{t}+ \left( \frac{2 h^{ij}}{e} \int d^{d}p \; \left(\frac{\partial \tilde{H}}{\partial \vec{p}^2}\right) p_{j} f \right) \partial_{i} \; , \\
		T\indices{^\mu_\nu} \partial_{\mu} \otimes \mathrm{d}x^{\nu} &=&   \left( \frac{1}{e} \int\mathrm{d}^{d}p \; \tilde{H} f \right) \partial_{t} \otimes \mathrm{d}t  + \left( \frac{1}{e}  \int d^{d}p \;  p_{i} f \right) \partial_{t} \otimes \mathrm{d}x^{i}  \nonumber \\
		&\;& +  \left( \frac{2 h^{ik}}{e} \int\mathrm{d}^{d}p \; \frac{\partial \tilde{H}}{\partial \vec{p}^2} p_{k} \tilde{H} f \right) \partial_{i} \otimes \mathrm{d}t \nonumber \\
		&\;& +  \left( \frac{2 h^{ik}}{e}  \int d^{d}p \; \frac{\partial \tilde{H}}{\partial \vec{p}^2} p_{j} p_{k} f \right) \partial_{i} \otimes \mathrm{d}x^{j}  \; , 
	\end{eqnarray}
	\end{subequations}
where we have used the volume form defined in \eqref{Eq:PullbackDeltaMeasure} and set $\lambda=1$. Importantly, on a flat background $e=1$ we see that the integral of the one-particle distribution function is just the particle number density - as one might have desired - and we reproduce \eqref{Eq:SingleDisperionCurrentSEMIntegralSimpleI} and  \eqref{Eq:SingleDisperionCurrentSEMIntegralSimpleII}.}

{\ As our final point, let us consider the generalisation of the standard distribution \eqref{Eq:Standard1Distribution} to arbitrary spacetimes and free Hamiltonians. In particular, we take
	\begin{eqnarray}
		f_{s} = e^{\Theta[\hat{\beta}]}
	\end{eqnarray}
with $\hat{\beta}$ the uplift of a timelike Killing field. Locally, on a flat spacetime, the quantity  $\Theta[\hat{\beta}]$ has the form
	\begin{eqnarray}
		\label{Eq:RoughBoltzmannfactor}
		\Theta[\hat{\beta}] \sim p_{\mu} \beta^{\mu} = - \frac{p_{0} - \vec{p} \cdot \vec{v}}{T} \; , \qquad \beta^{\mu}= \frac{1}{T} (1,\vec{v}) \; , 
	\end{eqnarray}
where we have employed \eqref{Eq:HydroTempandu} and thus our generalisation reproduces \eqref{Eq:Standard1Distribution} in the appropriate limit. In what follows we shall show that $f_{s}$ satisfies 
	\begin{eqnarray}
		\mathcal{L}_{\hat{\beta}} f_s = \mathcal{L}_{X_{\mathrm{f}}} f_s = 0 \; .
	\end{eqnarray}
This is exactly as desired for a stationary or ``equilibrium'' distribution as it is independent of time (roughly $\beta$) and free particle motion ($X_{\mathrm{f}}$).}

{\ In the relativistic case, the Lie derivative of the symplectic potential $\Theta$ (given in definition \ref{def:symplectic}) along the Hamiltonian vector field gives the derivative of the Hamiltonian. This only follows because of the simplicity of the relativistic Hamiltonian. In that case one can easily show that $\Theta[\hat{\beta}]$ is preserved along integral curves of the Hamiltonian vector field. In our boost agnostic situation we must work harder to show that $\Theta[\hat{\beta}]$ is conserved under the action of the Lie derivative along the Hamiltonian vector field. In particular, computing the Lie derivative along $X_{\mathrm{f}}$ one finds
	\begin{eqnarray}
			\label{Eq:Symplecticflowtheta}
			\mathcal{L}_{X_{\mathrm{f}}} \Theta 
		&=& \mathrm{d} \left( i_{X_{\mathrm{f}}} \Theta \right) + i_{X_{\mathrm{f}}} \mathrm{d} \Theta =  \mathrm{d} \left( p_{\mu} X_{\mathrm{f}}^{\mu} \right) + i_{X_{\mathrm{f}}} \Omega = \mathrm{d} \left( p_{\mu} X_{\mathrm{f}}^{\mu} \right) - \mathrm{d} H_{\mathrm{f}} \; . 
	\end{eqnarray}
It is easy to convince oneself that the first term is generally not proportional to the derivative of the Hamiltonian by explicit computation
	\begin{subequations}
	\begin{eqnarray}
		 	\mathcal{L}_{X_{\mathrm{f}}} \Theta
		&=&  \left( p_{\mu}  \nu^{\mu} \frac{\partial^2 H_{\mathrm{f}}}{\partial (p)^2} + 2 p_{\mu} \tilde{h}^{\mu \alpha} p_{\alpha} \frac{\partial^2 H_{\mathrm{f}}}{\partial p^2 \partial p}  \right) \nu^{\nu} D p_{\nu}   \nonumber \\
		&\;& + 2 \left( \frac{\partial H_{\mathrm{f}}}{\partial p^2}  + p_{\mu}  \nu^{\mu} \frac{\partial^2 H_{\mathrm{f}}}{\partial p \partial p^2} + p_{\mu} \tilde{h}^{\mu \alpha} p_{\alpha}  \frac{\partial^2 H_{\mathrm{f}}}{ \partial (p^2)^2}  \right) \tilde{h}^{\nu \sigma} p_{\sigma} Dp_{\nu} \; , 
	\end{eqnarray}
	\end{subequations}
where we have employed \eqref{Eq:FreeHExterior}. For generic $H_{\mathrm{f}}$ this is clearly not proportional to $\mathrm{d} H_{\mathrm{f}}$. Nevertheless, $\Theta[\hat{\beta}]$ will be conserved along Hamiltonian flows i.e. any scalar function of this expression will be a solution to the Boltzmann equation as we now demonstrate:}
    
\begin{proposition}[Conservation of $\Theta(\hat{\beta})$ along free Hamiltonian flows]
{\ If $H_{\mathrm{f}}$ is a free Hamiltonian and $\hat{\beta}$ the uplift of some Killing field $\beta$, then $\Theta[\hat{\beta}]$ is conserved along integral curves of both the Hamiltonian vector field and those of $\hat{\beta}$.}
\end{proposition}

\begin{proof}
{\ Let us begin with the second easier statement. Firstly, we notice that
	\begin{eqnarray}
			\mathcal{L}_{\hat{\beta}} \Theta
		&=& \mathrm{d} \left[ i_{\hat{\beta}} \Theta \right] + i_{\hat{\beta}} \mathrm{d} \Theta = \mathrm{d} \left[ \Theta[\hat{\beta}] \right] + i_{\hat{\beta}} \Omega = \mathrm{d} F - \mathrm{d} F  = 0 \; , \nonumber
	\end{eqnarray}
where we have used lemma \ref{prop:UpliftKilling} for the definition of $F$. The above identity holds for any vector on which $\Theta$ acts, therefore
	\begin{eqnarray}
		   \mathcal{L}_{\hat{\beta}} \Theta[\hat{\beta}]
		= \left( \mathcal{L}_{\hat{\beta}}[\Theta] \right)[\hat{\beta}] + \Theta\left[ \mathcal{L}_{\hat{\beta}} \hat{\beta} \right] = 0 \; , \nonumber
	\end{eqnarray}
where we have used that the Lie derivative of a vector field along itself is zero. Consequently, $\Theta[\hat{\beta}]$ is conserved along integral curves of $\hat{\beta}$.}

{\ On the other hand, we can demonstrate conservation of $\Theta[\hat{\beta}]$ along the free Hamiltonian vector field by explicit computation in local coordinates where
	\begin{eqnarray}
		\Theta[\hat{\beta}] = p_{\mu} \beta^{\mu} = F \; . \nonumber
	\end{eqnarray}
We find that
    \begin{eqnarray}
            \mathcal{L}_{X_{\mathrm{f}}} \Theta[\hat{\beta}]
       = \frac{\partial H_{\mathrm{f}}}{\partial p} p_{\rho} \mathcal{L}_{\beta} \nu^{\rho}
            + \frac{\partial H_{\mathrm{f}}}{\partial p^2} p_{\mu} \mathcal{L}_{\beta} \tilde{h}^{\mu \nu} p_{\nu} \; . \nonumber
    \end{eqnarray}
Employing the Killing conditions of \eqref{Eq:Killingconditions} we see that this vanishes. Therefore $\Theta[\hat{\beta}]$, or any function of this quantity, is constant along Hamiltonian flows.}
\end{proof}

{\ We also note that any function of the free invariant scalars, \eqref{Eq:PrimitiveInvariantsScalars}, will  satisfy invariance under flows of the Hamiltonian vector field and uplifted Killing vector. Consequently there is a very large class of functions that are stationary with respect to $\beta$ as is discussed in lemma \ref{lem:conservationunderuplift}. This is also true of the Galilean case and we must turn to collisions to find further constraints, as we did in section \ref{sec:applications}.}

\section{Conclusions and outlook}\label{sec:discussion}

{\noindent In this work, we have established a general framework for Hamiltonian mechanics on Aristotelian spacetimes, focusing on systems that lack local boost symmetry. We constructed invariant phase-space dynamics, introduced a class of free Hamiltonians, and established a generalized Liouville theorem valid on reparameterization-invariant constraint surfaces. We further showed that conserved quantities naturally arise from uplifted Aristotelian Killing vectors, and that ensembles of exotic free particles yield ideal hydrodynamic behavior at leading derivative order. Remarkably, the ideal gas law emerges universally, despite the absence of any boost symmetry.}

{\ Our results open several directions for future investigation. On the technical side, it would be important to extend the formalism to include interactions between particles, analyze the role of torsion in kinetic and hydrodynamic equations, and explore coupling to background gauge fields or external potentials. As we discussed in the introduction we would like to revisit several formal results \cite{Amoretti:2022vxq,Amoretti:2023vhe,Amoretti:2023hpb}, generalise and demonstrate their applicability in numerical simulations of classical exotic particles in the (quasi-)hydrodynamic limit. We also plan to develop the quantum version of our framework and investigate quantization on Aristotelian phase space, where the lack of Lorentz or Galilei symmetry gives rise to novel structures.}

{\ A particularly important direction is the construction of a generating functional for the conserved currents. By systematically varying the structure invariants, namely the clock form \( \tau_\mu \) and the degenerate spatial metric \( h_{\mu\nu} \), we expect to obtain all relevant currents, including the stress-energy-momentum complex, in a unified and covariant way. This would establish a direct link between our geometric framework and effective actions for non-relativistic fluids, and may enable the systematic inclusion of dissipative and higher-order effects.}

{\ Beyond formal developments, we expect our framework to have concrete applications. In particular, we propose applying boost-agnostic kinetic theory to systems of active matter and collective motion. Many existing treatments of flocking, such as kinetic theories derived from the Vicsek model, assume Galilean-invariant dispersion relations to relate mass and momentum \cite{PhysRevE.83.030901,PhysRevLett.101.268101,Patelli_2021,disalvo2025}. We argue that this assumption is conceptually inconsistent: such systems are typically coupled to a medium or substrate and lack boost invariance. The formalism developed here provides a natural and consistent alternative, enabling a better understanding of non-equilibrium collective behaviors and their emergent hydrodynamics \cite{PhysRevLett.75.4326}.}

\backmatter

\bmhead{Acknowledgements}

{\ D.B. would like to thank Kirill Krasnov for discussions. A.A. \& D.B. have received support from the project PRIN 2022A8CJP3 by the Italian Ministry of University and Research (MUR). D.B. is currently funded by PNRR GIOVANI RICERCATORI, CUP D33C25000470006. L.M. acknowledges support from DFA PARD GRANT,  CUP C93C23004320005. This project has also received funding from the European Union’s Horizon 2020 research and innovation programme under the Marie Sk\l{}odowska-Curie grant agreement No. 101030915.}

\begin{appendices}

\section{Local symmetry and Aristotelian manifolds}\label{appendix:review}
\renewcommand{\theequation}{A.\arabic{equation}}
\setcounter{equation}{0}
{\noindent On a $(d+1)$-dimensional Lorentzian spacetime, one can in principle construct vielbeins that carry a $(d+1)$-dimensional vector representation of the Lorentz group at each point. At the heart of what makes a $(d+1)$-dimensional Aristotelian manifold distinct is that the vielbeins instead carry a reducible vector representation of the $d$-dimensional spatial rotation group. In particular, let $\rho_{\mathrm{standard}}[g]$ be the standard (or defining) irreducible representation of $O(d)$ given by $d$-dimensional, real orthogonal matrices and $\rho$ a reducible, matrix representation of $O(d)$ by $(d+1)$-dimensional matrices acting on $\mathbbm{R}^{d+1}$. For example, one way to represent the action of $O(d)$ is by
	\begin{subequations}
	\label{Eq:ExampleRep}
	\begin{eqnarray}
	  \rho &:& O(d) \rightarrow GL(\mathbbm{R}^{d+1}) \; ,   \\
	  	\label{Eq:MatrixRep}
	  	 &:& g \mapsto \left( \begin{array}{ccc} 1 & & 0 \\ 0 & & \rho_{\mathrm{standard}}[g] \end{array} \right) \; . 
	\end{eqnarray}
	\end{subequations}
We note that $\left\{ \rho[g] \right\}_{g \in O(d)}$ is a subgroup of $O(d+1)$;  in particular for every element $\rho[g]$ there is an inverse which is given by $\rho[g]^{-1} = \rho[g]^{T}$ as $\rho[g]$ is an orthogonal matrix. Consequently, one will find three numerical invariants under the action of the spatial rotation matrices, a $(d+1)$-dimensional vector which we denote by $\tau$; a symmetric $(d+1)$-dimensional matrix $\delta^{(d)}$ with signature $(0,1,1,\ldots,1)$, and trivially the $(d+1)$-dimensional unit matrix. The non-trivial tensors satisfy
	\begin{eqnarray}
		\rho[g] \tau=\tau \; , \qquad \rho[g] \delta^{(d)} (\rho[g])^{T} = \delta^{(d)} \; . 	
	\end{eqnarray}
For comparison purposes, the Lorentz group of special relativity, $O(1,d)$ has an irreducible representation $\Lambda : O(1,d) \rightarrow GL(\mathbbm{R}^{d+1})$ that preserves the $(d+1)$-dimensional Minkowski metric $\eta$ i.e.
	\begin{eqnarray}
		\Lambda[g] \eta (\Lambda[g])^{T} = \eta \; , 
	\end{eqnarray}
where $g \in O(1,d)$. Similarly, the Galilean group includes shear matrices that do not have the form \eqref{Eq:ExampleRep}. A vector space with $\tau$ and $\delta^{(d)}$ as defined above is an Aristotelian vector space.}

{\ We remarked above that the matrices $\rho[g]$, the reducible representation of $O(d)$ that we are interested in, are orthogonal. Thus the action of $O(d)$ on the dual space can be represented by $\rho[g]^{T}$ as $\rho[g]^{-1}=\rho[g]^{T}$. Hence, on the dual vector space, we can also identify three numerical invariants - a $(d+1)$-dimensional vector which we denote by $\nu$; a symmetric $(d+1)$-dimensional matrix $\tilde{\delta}^{(d)}$ with signature $(0,1,1,\ldots,1)$, and trivially the $(d+1)$-dimensional unit matrix. The non-trivial tensors satisfy
	\begin{eqnarray}
		\nu \rho[g]^{T} = \nu \; , \qquad \rho[g]^{T} \tilde{\delta}^{(d)} \rho[g] = \tilde{\delta}^{(d)} \; . 	
	\end{eqnarray}
We can additionally choose to normalise $\nu$ so that $\tau(\nu)=-1$. This overall sign choice is a convention to make analogies with Lorentz symmetry more natural. We can also use $\nu$ to impose constraints on $\delta^{(d)}$. While there is no reason to think that $\delta^{(d)}(\nu,X)=0$ for a given $X$, we can construct a new tensor $h$ that has this property and is still a numerical invariant. We introduce the left and right projectors
	\begin{eqnarray}
		\label{Eq:Projectors}
		P^{(R)} = \delta  + \nu \otimes \tau \; , \qquad P^{(L)} = \delta + \tau \otimes \nu \; ,
	\end{eqnarray}
and subsequently define
	\begin{eqnarray}
		\label{Eq:hdef}
		h = P^{(L)} \delta^{(d)} P^{(R)} \; , \qquad \tilde{h} =  P^{(R)} \tilde{\delta}^{(d)} P^{(L)} \; . 
	\end{eqnarray}
which are also numerical invariants under the reducible representation of $O(d)$. These new tensors satisfy the required property i.e.
	\begin{eqnarray}
		h(\nu,X) = 0 \; , \qquad \tilde{h}(\tau,Y) = 0 	
	\end{eqnarray}
for all $X$ and $Y$ where $X$ is any vector in the original Aristotelian space and $Y$ an element of the corresponding dual vector space.}

{\ Turning now to curved manifolds, we shall assume on some manifold $M$ that patchwise we are given a smooth set of vielbeins 
  \begin{eqnarray}
     \label{Eq:Coframe}
    e^{I} &=& e^{I}_{\mu} \mathrm{d} x^{\mu} \; ,
  \end{eqnarray}
in a neighbourhood of every point. Thus, at each point $x$ on the manifold $M$ this introduces a basis, $\left\{ e^{I} \right\}$, of the cotangent space $T_{x}^{*}M$. Furthermore, given the basis of the cotangent space $\left\{ e^{I} \right\}$ (the coframe \eqref{Eq:Coframe}) at a point $x \in M$, we can construct a dual basis of the dual vector space, denoted $\left\{e_{I}\right\}$ (the frame), through demanding that
	\begin{eqnarray}
		\label{Eq:DualDef}
		e^{I}(e_{J}) = \delta\indices{^I_J} \; ,
	\end{eqnarray}
where $\delta\indices{^I_J}$ is the $(d+1)$-dimensional identity matrix. In general the vector field basis $e_{I}$ defined by \eqref{Eq:DualDef}, will not form a coordinate system in the neighbourhood of a point on the manifold. To see this when explicitly given such vectors, it is only necessary to compute the action of the $e_{I}$ on themselves
		\begin{eqnarray}
			\label{Eq:AnholonomyCoeffs}
		\left[ e_{I}, e_{J} \right] &=& C\indices{^{K}_{IJ}} e_{K} \; , 
		\end{eqnarray}
which is the defining relationship of the anholonomy coefficients $C\indices{^{K}_{IJ}}$. When the coefficients $C\indices{^{K}_{IJ}}$ vanish then the $\left\{ e_{I} \right\}$ form a coordinate basis (coordinate derivatives commute on a smooth manifold). On the contrary, if the anholonomy coefficients are non-zero then the basis vectors cannot be integrated to form a coordinate system.}

{\ Given some ordered basis $\left\{ e^{I} \right\}$ on $M$ this might not be the most appropriate to work with. We can, without loss of generality, choose and/or replace the leading element by $e^{0} \mapsto \tau$. Subsequently, the remaining elements $\left\{ e^{I} \right\}/ \left\{\tau\right\}$ can be acted upon by the projector \eqref{Eq:Projectors} to enforce $\nu(\mathrm{span}\left\{ e^{I} \right\} /\left\{\tau\right\}) = 0$. As $\nu$ spans the kernel of the tensor $h$, defined in \eqref{Eq:hdef}, one finds that $h$ is a non-degenerate bilinear form on these remaining elements. We can then perform the Gram-Schmidt procedure to create a new orthonormal basis $\left\{ e^{i=1,\ldots,d} \right\}$ such that
	\begin{eqnarray}
		\label{Eq:StandardBasis}
		\tau(e^{0}) = \tau(\nu) = -1 \; , \qquad \tau(e^{i}) = 0 \; , \qquad h(e^{0},X) = 0 \; , \qquad h(e^{i},e^{j}) = \delta_{ij} \; . 
	\end{eqnarray}
Using \eqref{Eq:DualDef} we also have the dual basis $\left\{ \tau, e_{i} \right\}$. However, it is not necessarily the case that $\tilde{h}(e^{i},e^{j})=\delta^{ij}$. Nevertheless rather than using $\tilde{h}$ defined in \eqref{Eq:hdef}, which was some tensor supplied to us, we can construct a new $\tilde{h} = \delta^{ij} e_{i} \otimes e_{j}$ that transforms appropriately under $O(d)$. The following identity between coefficients in an arbitrary basis then follows
	\begin{eqnarray}
		\label{Eq:AbstractSpaceRelationsappendix}
		\tilde{h}^{IK} h_{KJ} = \delta\indices{^I_J} + \nu^{I} \tau_{J}
	\end{eqnarray}
which relates the numerically invariant tensors of the vector space and its dual, with $\delta\indices{^I_J}$ the $(d+1)$-dimensional identity matrix.}

{\  With the symmetry structure explained, we can now use the frame and coframe to define relevant tensors on the tangent and cotangent spaces to the base manifold $M$. In particular, we define the components of the ``structure invariants'' in the coordinate basis to be given by
  \begin{eqnarray}
    \label{Eq:StructureInvariantsDef}
    h_{\mu \nu} := h_{IJ} e^{I}_{\mu} e^{J}_{\nu}  \; , \qquad  \tau_{\mu} :=  \tau_{I} e^{I}_{\mu} \; .
  \end{eqnarray}
In other works \cite{deBoer:2017ing,de_Boer_2020}, the Aristotelian manifold is defined in terms of just the quantities without explicit mention of the vielbeins $\left\{ e^{I} \right\}$. Subsequently, the existence of vielbeins is then at some point, almost always, assumed to allow us to define $\tilde{h}$ ($\nu$ is always definable in terms of the kernel of $h$). We however have made this part of our definition from the start. The inverse structure invariant components are defined by
	\begin{eqnarray}
			\label{Eq:InverseStructureDef}
		\nu^{\mu} := \nu^{I} e_{I}^{\mu}  \; , \qquad \tilde{h}^{\mu \nu} := h^{IJ} e_{I}^{\mu} e_{J}^{\nu} \;  .
	\end{eqnarray}
From these expressions we can derive \eqref{Eq:RelationsSIandinvSIIntro} given in the main body of the text. This construction we have pursued entirely encodes our physical intuition that a space locally look Aristotelian. Any space that did not allow us to perform these operations in a local neighbourhood, or at least argue away isolated such points, we would reject as being ``unphysical''. It is this attitude that we shall hold from now on.}

{\ Our approach here should be compared to that taken by some other works such as \cite{Armas:2020mpr}. These authors for example start from $\tau$ and $h$, define the matrix $\tau_\mu\tau_\nu+h_{\mu\nu}$ and $\nu$ as the covector such that $h(\nu,\cdot)=0$. From here $\tilde h$ as the block matrix of the inverse $\nu^\mu \nu^\nu+\tilde h^{\mu\nu}$, is argued to exist.}

{\ Given our basis, we can construct a volume form:}

\begin{definition}[The volume form]\label{def:StandardBasis}
{\ Let our manifold $M$ be orientable, and suppose we are given an arbitrary set of vielbeins $\left\{ e^{I} \right\}$ and we perform the construction above to diagonalise this basis. We define the volume form on $M$ to be
  \begin{eqnarray}
    \label{Eq:canonicalvolume}
    \mathrm{vol}_{M} &=& \tau \wedge e^{1} \wedge \ldots \wedge e^{d} \; , 
  \end{eqnarray}
which in local coordinates takes the form of \eqref{Eq:etalocalIntro}.}
\end{definition}

{\ The above discussion suffices to outline the necessary geometric concepts. We now turn to the simplest notion of derivative - the exterior derivative, which requires no additional structure beyond the manifold being differentiable.  The exterior derivative acts on forms to produce higher forms in the usual way. In particular, the exterior derivative of coframe fields $e^{I}$ can be expressed in terms of the anholonomy coefficients \eqref{Eq:AnholonomyCoeffs} without having to introduce a connection (additional structure). To demonstrate this, consider
	\begin{eqnarray}
			(\mathrm{d} e^{I})(e_{J}, e_{K})	
		&=& e_{J} \left( e^{I}(e_{K}) \right) - e_{K} \left( e^{I}(e_{J}) \right) - e^{I}\left( [e_{J}, e_{K}] \right) =  - C\indices{^{I}_{JK}} \; , 
	\end{eqnarray}
where $\mathrm{d}$ is the exterior derivative and we have used \eqref{Eq:DualDef} and \eqref{Eq:AnholonomyCoeffs}. It follows that
	\begin{eqnarray}
		\label{Eq:ExteriorBasisForms}
		\mathrm{d} e^{I} = - \frac{1}{2} C\indices{^{I}_{JK}} e^{J} \wedge e^{K} \; . 
	\end{eqnarray}
In a coordinate basis this becomes
	\begin{eqnarray}
		\label{Eq:ExteriorBasisFormsCoordinate}
		\left( \partial_{[\mu} e_{\nu]}^{I} + \frac{1}{2} C\indices{^{I}_{JK}} e_{\mu}^{J} e_{\nu}^{K} \right) \mathrm{d}x^{\mu} \wedge \mathrm{d}x^{\nu} = 0 \; , 
	\end{eqnarray}
which is an expression which will be useful shortly.}

{\ Now we wish to extend our notion of derivative to vectors, and as such we need to introduce additional structure - a connection - and thus the covariant derivative. Among the space of derivatives we shall suppose there exists at least one compatible with our structure invariants i.e.
	\begin{eqnarray}
		\label{Eq:DerivativeCompatibilityAbstract}
		\nabla_{v} \tau = 0 \; , \qquad \nabla_{v} h = 0 \; ,
	\end{eqnarray}
for any vector field $v: M \rightarrow TM$. Such covariant derivatives have the attractive property that they preserve products between the structure invariants and vectors (e.g. angles, spatial lengths and time intervals) along integral curves of any given vector field $v$.}
	
{\  Leaving aside structure compatibility for the moment, we can act on a given set of frame fields $e_{I}$ with any covariant derivative of our liking. The connection coefficients are defined by auto-parallel transport of the frame fields with respect to this covariant derivative:
	\begin{eqnarray}	
		\label{Eq:ConnectionDefinition}
		\nabla_{e_{I}} e_{J} &=:&  e_{K} \omega\indices{^{K}_{JI}} \; .
	\end{eqnarray}
The connection coefficients $\omega\indices{^{K}_{JI}} $ measure how the basis $e_{I}$ changes as we flow in the respective directions of the basis elements\footnote{Note that $\omega\indices{^{K}_{JI}}$ is not necessarily antisymmetric in $IJ$.}. Given $\omega\indices{^{K}_{JI}} $, using linearity and the Leibnitz rule, we can construct the directional derivative of any vector field, along any vector field i.e.
	\begin{eqnarray}
			\nabla_{v} u &=& \nabla_{v^{I} e_{I}} \left( u^{J} e_{J} \right)
			=  v^{I} \left( e_{I}(u^{J}) + \omega\indices{^{J}_{KI}} u^{K} \right) e_{J} \; , 
	\end{eqnarray}
where we have used the fact that $u^{J}$ are functions on $M$ and not vector components. Moreover, employing the defining relationship between frame and coframe \eqref{Eq:DualDef} we can also determine the covariant derivative of the coframe elements
	\begin{eqnarray}
		\label{Eq:CovariantDerivCoframe}
		\nabla_{e_{I}} e^{J} &=& - \omega\indices{^{J}_{KI}} e^{K} \; , 
	\end{eqnarray}
allowing us to define the derivative of dual vectors. Standard techniques then extend the covariant derivative of our choosing to arbitrary tensors.}

{\ Given the connection coefficients \eqref{Eq:ConnectionDefinition} and the relation \eqref{Eq:DualDef}, the torsion two-form of the system is defined to be the quantity
	\begin{eqnarray}
		\label{Eq:Torsion2formdef}
		\Theta^{I} := \mathrm{d} e^{I} + \omega\indices{^I_J} \wedge e^{J} \; , \qquad \omega\indices{^I_J}  &=& \omega\indices{^I_{JK}} e^{K} \; . 
	\end{eqnarray}
This quantity describes how parallel transport of a vector can fail to be independent of the path (see \cite{Nakahara:2003nw}, chapter 7 for more details). It contains two parts - the failure of a basis to become coordinates and torsion effects arising from the connection (second term). Unlike in the relativistic case, where we can always explicitly construct a torsion-free connection, there is no a priori reason to assume that torsion vanishes in our Aristotelian cases. Thus we must bear the burden of carrying it with us. In particular, using our definition of the torsion two-form \eqref{Eq:Torsion2formdef} and \eqref{Eq:ExteriorBasisForms}, we find that we can write:
	\begin{subequations}
	\label{Eq:TorsionDef}
	\begin{eqnarray}
		\Theta^{I} = &=& \left( - \frac{1}{2}  C\indices{^{I}_{JK}} - \omega\indices{^{I}_{KJ}} \right) e^{J} \wedge e^{K} \; , 
	\end{eqnarray}
or in components,
	\begin{eqnarray}
		\Sigma\indices{_{\Theta}^{I}_{JK}} &=& 2 \omega\indices{^{I}_{[JK]}} - C\indices{^{I}_{JK}} \; , \qquad \Theta^{I} = \frac{1}{2} \Sigma\indices{_{\Theta}^{I}_{JK}} e^{J} \wedge e^{K} \; ,
	\end{eqnarray}
	\end{subequations}
where the script on $\Sigma_{\Theta}$ indicates this is the torsion components with respect to $\Theta$; as we shall see there is an annoying sign difference that appears when we define the torsion components in terms of the antisymmetric part of a coordinate connection $\Gamma^{\rho}_{\mu \nu}$. Importantly, if we pick a coframe $e^{I}$ and a connection form $\omega\indices{_{J}^{I}}$ then the torsion is fixed to a particular value.}

{\ What we have discussed thus far applies to any connection. Let us now return to our desire to find a structure invariant compatible connection \eqref{Eq:DerivativeCompatibilityAbstract}. Compatibility of any such covariant derivative with the conditions of \eqref{Eq:DerivativeCompatibilityAbstract}, using \eqref{Eq:ConnectionDefinition}, means that
	\begin{eqnarray}
		\label{Eq:ConnectionConstraints}
		 \tau_{J} \omega\indices{^J_I} = 0 \; , \qquad  h_{IK} \omega\indices{^K_J} + h_{JK} \omega\indices{^{K}_{I}} = 0 \; , 
	\end{eqnarray}
where we have used that $\tau_{I}$ and $h_{IJ}$ are constant scalars on $M$ to set $\nabla_{e_{K}} \tau_{I}=\nabla_{e_{K}} h_{IJ}=0$. In the relativistic case we only have the latter equation with $h_{IJ}$ replaced by the Minkowski metric $\eta_{IJ}$. These conditions tells us that at a point the connection coefficients must be a representation of the respective group (a reducible representation of $O(d)$ for Aristotelian spaces and an irreducible representation of $O(1,d)$ for Lorentzian spaces) but otherwise leaves how they vary in spacetime arbitrary. In the relativistic case, because the connection can be entirely constructed in terms of the metric, this missing information is supplied by the behaviour of the metric on the manifold under the assumption that torsion vanishes. The Aristotelian case is more complex due to both the fact that we have no basis to ignore torsion, and even in its absence there is an ambiguity in the connection built from the structure invariants.}

{\ Suppose we have picked a connection form $\omega\indices{_{J}^I}$ that varies smoothly from one point on the manifold to another - that is, we smoothly assign a matrix from the reducible representation of $O(d)$ on a $(d+1)$-dimensional Aristotelian vector space to each point on our manifold $M$. Our next goal is to relate $\omega\indices{^I_{J}}$ to the usual coordinate connection\footnote{We drop the additional vector in defining the covariant derivative from this point onward, with it being understood that a Greek index indicates projection of the derivative along a coordinate direction which a uppercase Latin index indicate projection along a non-coordinate basis direction.} defined by
	\begin{eqnarray}
		\nabla_{\mu} \partial_{\nu} &=& - \Gamma_{\mu \nu}^{\rho} \partial_{\rho}
	\end{eqnarray}
for a given coordinate basis $\left\{ \partial_{\mu} \right\}$. As the basis is arbitrary, we have $v=v^{I} e_{I} = v^{\mu} \partial_{\mu}$ with $v^{I} = v^{\mu} e^{I}_{\mu}$. Consequently, taking the covariant derivative projected along a coordinate direction we find
	\begin{eqnarray}
			\nabla_{\mu} \left( v^{I} e_{I} \right)
		&=& \left[ \partial_{\mu} v^{\rho} + v^{\sigma} \left( \partial_{\mu}  e^{J}_{\sigma} + e^{I}_{\sigma} \omega\indices{^{J}_{I \mu }}  \right) e_{J}^{\rho} \right] \partial_{\rho} \; , \qquad \omega\indices{^{I}_{J \mu}}   
		= \omega\indices{^{I}_{JK}}  e^{K}_{\mu}  \; , 
	\end{eqnarray}
which upon comparing with the usual coordinate expression for the covariant derivative implies
	\label{Eq:LinkConnections} 
	\begin{eqnarray}
		\label{Eq:LinkConnections} 
		 \Gamma_{\mu \nu}^{\rho} &=&  e_{I}^{\rho} \left( \partial_{\mu}  e^{I}_{\nu} + e^{J}_{\nu} \omega\indices{^{I}_{J\mu}}  \right)  \; . 
	\end{eqnarray}
The above expression allows us to derive the usual coordinate space connection $\Gamma_{\mu \nu}^{\rho}$ from objects that naturally represent the relevant local symmetries of the spacetime.}

{\ Using \eqref{Eq:ConnectionConstraints} and \eqref{Eq:LinkConnections} it is straightforward to show that
	\begin{eqnarray}
		\label{Eq:DerivativeCompatibilityCoord}
		\nabla_{\mu} \tau_{\nu} = 0 \; , \qquad \nabla_{\mu} h_{\nu \rho} = 0 \; ,
	\end{eqnarray}
as desired from a structure invariant compatible connection. Comparable expressions are true for the inverse structure invariants.  However, the canonical volume form $\mathrm{vol}_{M}$ defined in \eqref{Eq:canonicalvolume} is generally \underline{not} compatible with structure invariant compatible covariant derivatives. To see this we compute
	\begin{eqnarray}
			\nabla_{e_{J}} \mathrm{vol}_{M}
		&=& - \omega\indices{^{I}_{IJ}} \mathrm{vol}_{M} \;  .
	\end{eqnarray}
There is no a priori reason\footnote{This should be compared to the relativistic case where, if we define $\omega_{IJL} = \eta_{IK} \omega\indices{^{K}_{JL}}$, requiring that $\omega\indices{^I_{J}}$ is a local representation of the Lorentz algebra implies $\omega_{(IJ)}=0$ or more usefully $\eta^{IJ} \omega_{IJK} = \omega\indices{^I_{IK}} = 0$.} to assume that $\omega\indices{^I_{IJ}} =0$.}

{\ We can further use \eqref{Eq:LinkConnections} to express the connection $\Gamma_{\mu \nu}^{\rho}$ in terms of the clock-form $\tau_{\mu}$, the spatial distance $h_{\mu \nu}$ and the coordinate components of the torsion $\Sigma\indices{_{\mu \nu}^{\rho}}$. We first write \eqref{Eq:LinkConnections} projected entirely onto non-coordinate directions
 	\begin{eqnarray}
		\Gamma_{IJ}^{K} &=& \omega\indices{^{K}_{JI}} + e^{\mu}_{I} e^{\nu}_{J} \partial_{\nu} e_{\mu}^{K} \; , \qquad 
		\Gamma_{IJ}^{K} = e^{\mu}_{I} e^{\nu}_{J} \Gamma_{\mu \nu}^{\rho} e_{\rho}^{K} \; . 
	\end{eqnarray}
We note that the antisymmetric part of $\Gamma_{IJ}^{K}$ is just the torsion defined in \eqref{Eq:Torsion2formdef}, 
	\begin{eqnarray}
		\label{Eq:AdditionalSign}
		\Sigma\indices{^K_{IJ}} = 2 \Gamma_{[IJ]}^{K} &=& 2 \omega\indices{^{K}_{[JI]}} + 2 e^{\mu}_{[I} e^{\nu}_{J]} \partial_{\nu} e_{\mu}^{K} = -\left( 2 \omega\indices{^{K}_{[IJ]}} - C\indices{_{IJ}^{K}} \right) = - \Sigma\indices{_{\Theta}^{K}_{IJ}} \; , \qquad
	\end{eqnarray}
where we have employed \eqref{Eq:ExteriorBasisFormsCoordinate} and defined $\Sigma$ without the subscript $\Theta$ to be the antisymmetric part of the Christoffel connection; up to the earlier noted sign. The definition of the torsion $\Theta$ in \eqref{Eq:TorsionDef} applies whether the basis is a coordinate basis or a non-coordinate basis; thus if one picks a coordinate basis so that $C\indices{^\rho_{\mu \nu}}=0$, then one sees that
	\begin{eqnarray}
		\label{Eq:CoordinateBasisTorsion}
		\Sigma\indices{_{\Theta}^{\rho}_{\mu \nu}} &=& \Sigma\indices{^{\rho}_{[\mu \nu]}} = 2 \Gamma^{\rho}_{[\mu \nu]} \; . 	
	\end{eqnarray}
The extra sign in \eqref{Eq:AdditionalSign} between a given non-coordinate basis and the coordinate basis is due to the identity \eqref{Eq:LinkConnections} which ultimately has its origin in the overall sign in the definition of the anholonomy coefficients \eqref{Eq:AnholonomyCoeffs} and connection coefficients \eqref{Eq:ConnectionDefinition}.}

{\ The expressions \eqref{Eq:ConnectionConstraints} tell us certain components of $\Gamma_{IJ}^{K}$ are entirely expressed in terms of derivatives of the frame and coframe (i.e. no $\omega\indices{^I_J}$ components)
	\begin{eqnarray}
		\Gamma_{IJ}^{K} \tau_{K}	&=& e^{\mu}_{I} e^{\nu}_{J} \partial_{\nu} \tau_{\mu} \; , \qquad
			\Gamma_{IJ}^{P} h_{PK} + \Gamma_{KJ}^{P} h_{PI}
		= e^{\mu}_{I} e^{\nu}_{J}  e^{\rho}_{K}  \partial_{\nu} h_{\rho \mu}   \; , 
	\end{eqnarray}
where in deriving the second expression we have used the following identity 
	\begin{eqnarray}
		 \left( h_{IP} e^{\nu}_{K} + h_{KP} e^{\nu}_{I} \right) e^{\rho}_{J}  \partial_{\rho} e^{P}_{\nu} &=& e^{\mu}_{I} e^{\nu}_{K} e^{\rho}_{J}  \partial_{\rho} h_{\mu \nu} =  e^{\mu}_{I} e^{\nu}_{J}  e^{\rho}_{K}  \partial_{\nu} h_{\rho \mu} \; , 
	\end{eqnarray}
which can be obtained by differentiating \eqref{Eq:hdef} and acting with the coframe coefficients on the result. Contracting the second identity with $\tilde{h}$ subsequently leads to 
	\begin{eqnarray}
			\Gamma_{(IJ)}^{L} h_{LM} \tilde{h}^{MK}
		&=& - \frac{1}{2} \Sigma\indices{_{IM}^{L}} h_{LJ} \tilde{h}^{MK} - \frac{1}{2}  \Sigma\indices{_{JM}^{L}} h_{LI} \tilde{h}^{MK} \nonumber \\
		&\;& +  e^{\mu}_{I} e^{\nu}_{J} e^{K}_{\rho} \left[ \frac{1}{2}  \tilde{h}^{\rho \sigma} \left(  \partial_{\nu} h_{\sigma\mu}  + \partial_{\sigma} h_{\nu \mu} -  \partial_{\mu} h_{\sigma \nu} \right) \right] \; .
	\end{eqnarray}
Using the $(d+1)$-dimensional decomposition of the identity given in \eqref{Eq:AbstractSpaceRelationsappendix}, we find
	\begin{eqnarray}
			\Gamma{_{IJ}^{K}}
		&=& - e^{\mu}_{I} e^{\nu}_{J} e_{\rho}^{K} \left( \nu^{\rho} \partial_{\nu} \tau_{\mu} \right) + \Gamma{_{(IJ)}^{L}} h_{LM} \tilde{h}^{MK}
			+ \frac{1}{2} \Sigma\indices{_{IJ}^{L}} h_{LM} \tilde{h}^{MK} \; . 
	\end{eqnarray}
Multiplying by appropriate transformation factors between frame, coframe and coordinate bases arrive at the following expression
	\begin{eqnarray}
		\label{Eq:ConnectionSymbolCoordinate}
		\Gamma_{\mu \nu}^{\rho} &=& - \nu^{\rho} \partial_{(\mu} \tau_{\nu)}
            + \frac{1}{2} \tilde{h}^{\rho \lambda} \left(  \partial_{\mu} h_{\nu \lambda} + \partial_{\nu} h_{\mu \lambda} - \partial_{\rho} h_{\mu \lambda} \right) \nonumber \\
        &\;& + \frac{1}{2} \tilde{h}^{\rho \lambda} \left( \Sigma\indices{_{\mu \nu}^{\sigma}} h_{\lambda \sigma} - \Sigma\indices{_{\mu \lambda}^{\sigma}} h_{\nu \sigma} - \Sigma\indices{_{\nu \lambda}^{\sigma}} h_{\mu \sigma}  \right) - \nu^{\rho} \partial_{[\mu} \tau_{\nu]} \; . 
	\end{eqnarray}
The second line is the Aristotelian version of the contorsion tensor from relativistic theory i.e.
	\begin{eqnarray}
			\Gamma^{\rho}_{[\mu \nu]} = \frac{1}{2} \Sigma\indices{_{\mu \nu}^{\rho}} 
		&=& \frac{1}{2} \Sigma\indices{_{\mu \nu}^{\sigma}} h_{\sigma \lambda} \tilde{h}^{\lambda \rho}  - \nu^{\rho} \partial_{[\mu} \tau_{\nu]} \; .
	\end{eqnarray}
We can equivalently derive this expression \eqref{Eq:ConnectionSymbolCoordinate} directly in the coordinate basis as we discuss in section \ref{sec:review}, however in doing so we lose the importance of local symmetries. Indeed, certain versions of Newton-Cartan theory, describing Galilean gravity, have exactly the same structure as Aristotelian spacetimes. What makes them differ is the local symmetries.}

{\ For our structure invariant connection, we can see that the torsion is generically non-zero; and forcing it to be will at least impose a constraint on the clock form. Crucially, in the presence of torsion, certain familiar results from (pseudo-)Riemannian geometry become more complicated. For example, the commutator of derivatives acting on any tensor fields includes an additional term
 	\begin{subequations}
	\begin{eqnarray}
			\label{Eq:CommutatorTensor}
			 \left[ \nabla_{\mu}, \nabla_{\nu} \right] T\indices{_{\rho_{1} \ldots \rho_{n}}^{\sigma_{1} \ldots \sigma_{m}}}
		&=&  \sum_{j=1}^{j=m} R\indices{_{\mu \nu \lambda}^{\sigma_{j}}} T\indices{_{\rho_{1} \ldots \rho_{n}}^{\sigma_{1} \ldots \sigma_{j-1} \lambda \sigma_{j+1} \ldots \sigma_{m}}} \nonumber \\
 		&\;& - \sum_{i=1}^{i=n}  R\indices{_{\mu \nu \rho_{i}}^{\lambda}} T\indices{_{\rho_{1} \ldots \rho_{i-1} \lambda \rho_{i+1} \ldots \rho_{n}}^{\sigma_{1} \ldots \sigma_{m}}} \nonumber \\
		&\;& -  \Sigma\indices{_{\mu \nu}^{\lambda}} \nabla_{\lambda} T\indices{_{\rho_{1} \ldots \rho_{n}}^{\sigma_{1} \ldots \sigma_{m}}} \; , 
	\end{eqnarray}
where we have defined the curvature tensor in the usual manner
 	\begin{eqnarray}
		\label{Eq:CurvatureDef}
		 R\indices{_{\mu \nu \sigma}^{\rho}} = 
		2 \left( \partial_{[\mu} \Gamma^{\rho}_{\nu] \sigma}  + \Gamma^{\rho}_{[\mu| \alpha}  \Gamma^{\alpha}_{|\nu] \sigma}  \right) \; . 
	\end{eqnarray}
	\end{subequations}
Extreme caution must be exercised by those familiar with relativistic physics in manipulating $R\indices{_{\mu \nu \rho}^{\sigma}}$ as common features of the Riemann tensor fail to hold for the Aristotelian curvature tensor.}

\begin{lemma}[Collected properties of the curvature tensor]\label{lem:Riemann}
{\ The Aristotelian curvature tensor has the following properties:
	\begin{enumerate}
		\item The curvature tensor is antisymmetric in its first pair of indices
				\begin{eqnarray}
					R\indices{_{( \mu \nu ) \rho}^{\sigma}} &=& 0 \; , 
				\end{eqnarray}
		\item The first and second Bianchi identities are
			\begin{subequations}
				\begin{eqnarray}
					R\indices{_{[\mu \nu \rho]}^{\sigma}} 
		&=& \nabla_{[\mu} \Sigma\indices{_{\nu \rho]}^{\sigma}} - \Sigma\indices{_{[\mu \nu}^{\alpha}} \Sigma\indices{_{\rho] \alpha}^{\sigma}} \; , \\
					\nabla_{[\mu} R\indices{_{\nu \rho] \lambda}^{\sigma}} 
				&=& \Sigma\indices{_{[\mu \nu}^{\alpha}} R\indices{_{\rho] \alpha \lambda}^{\sigma}}\; , 
				\end{eqnarray}
			\end{subequations}
		\item For a structure invariant compatible derivative the following combinations of components of the curvature tensor vanish
			\begin{subequations}
			\begin{eqnarray}
				\label{Eq:RiemannContracttau}
				R\indices{_{\mu \nu \rho}^{\sigma}} \tau_{\sigma} &=& 0 \; , \\
				\label{Eq:RiemannContracth}
				R\indices{_{\mu \nu \rho}^{\lambda}} h_{\lambda \sigma} + R\indices{_{\mu \nu \sigma}^{\lambda}} h_{\rho \lambda} &=& 0  \; , \\
				\label{Eq:RiemannContractnu}
				R\indices{_{\mu \nu \sigma}^{\rho}} \nu^{\sigma} &=& 0 \; , \\
				\label{Eq:RiemannContracthu}
				R\indices{_{\mu \nu \lambda}^{\rho}} \tilde{h}^{\lambda \sigma} + R\indices{_{\mu \nu \lambda}^{\sigma}} \tilde{h}^{\rho \lambda} &=& 0  \; ,
			\end{eqnarray}
			\end{subequations}
	\end{enumerate}
The first two properties are independent of whether the curvature tensor is structure-invariant compatible but apply to connections with torsion. The final property depends on the covariant derivative being structure-invariant compatible.}
\end{lemma}

\begin{proof}
{\ Our first point follows from the definition of the curvature tensor in terms of the connection \eqref{Eq:CurvatureDef}. Notice that the other index symmetries of the Riemann tensor, such as antisymmetry in the last two indices and symmetry under pair exchange of first and second index with third and fourth, do not generally hold.}

{\  To demonstrate the first Bianchi identity, we use \eqref{Eq:CommutatorTensor} and consider the tensor $\nabla_{\rho} \phi$ so that
	\begin{eqnarray}
			\left[ \nabla_{\mu}, \left[ \nabla_{\nu}, \nabla_{\rho} \right] \right] \phi
		&=&  - \nabla_{\mu} \Sigma\indices{_{\nu \rho}^{\sigma}} \nabla_{\sigma} \phi + \Sigma\indices{_{\nu \rho}^{\sigma}} \Sigma\indices{_{\mu \sigma}^{\alpha}} \nabla_{\alpha} \phi  + R\indices{_{\nu \rho \mu}^{\sigma}} \nabla_{\sigma} \phi \; . \nonumber
	\end{eqnarray}
By permuting indices and using the Jacobi identity,
	\begin{eqnarray}
		\label{Eq:JacobiIdentity}
		\left[ \nabla_{\mu}, \left[ \nabla_{\nu}, \nabla_{\rho} \right] \right] 
		+ \left[ \nabla_{\nu}, \left[ \nabla_{\rho}, \nabla_{\mu} \right] \right]
		+ \left[ \nabla_{\rho}, \left[ \nabla_{\mu}, \nabla_{\nu} \right] \right] = 0 \; , \nonumber
	\end{eqnarray}
one soon finds that
	\begin{eqnarray}
			R\indices{_{[\mu \nu \rho]}^{\sigma}} 
		&=& \nabla_{[\mu} \Sigma\indices{_{\nu \rho]}^{\sigma}} - \Sigma\indices{_{[\mu \nu}^{\alpha}} \Sigma\indices{_{\rho] \alpha}^{\sigma}} \; . 
		\nonumber
	\end{eqnarray}
In other words, the first Bianchi identity for the Riemann tensor is only algebraic in the absence of torsion. As for the second Bianchi identity, we use the Jacobi identity \eqref{Eq:JacobiIdentity} on a generic vector field $X^{\sigma}$ and employ \eqref{Eq:CommutatorTensor} to find
	\begin{eqnarray}
			 0
		&=& \left[ \nabla_{\mu}, \left[ \nabla_{\nu}, \nabla_{\rho} \right] \right] X^{\sigma}
		+  \left[ \nabla_{\nu}, \left[ \nabla_{\rho}, \nabla_{\mu} \right] \right] X^{\sigma}
		+ \left[ \nabla_{\rho}, \left[ \nabla_{\mu}, \nabla_{\nu} \right] \right] X^{\sigma} \nonumber \\
		&=&  \left( \nabla_{[\mu} R\indices{_{\nu \rho] \lambda}^{\sigma}} - \Sigma\indices{_{[\mu \nu}^{\alpha}} R\indices{_{\rho] \alpha \lambda}^{\sigma}}\right) X^{\lambda}  +  \left( R\indices{_{[\mu \nu  \rho]}^{\lambda}} - \nabla_{[\mu} \Sigma\indices{_{\nu \rho]}^{\lambda}} + \Sigma\indices{_{[\mu \nu}^{\alpha}} \Sigma\indices{_{\rho] \alpha}^{\lambda}} \right) \nabla_{\lambda} X^{\sigma} \nonumber \\
		&=&  \left( \nabla_{[\mu} R\indices{_{\nu \rho] \lambda}^{\sigma}} - \Sigma\indices{_{[\mu \nu}^{\alpha}} R\indices{_{\rho] \alpha \lambda}^{\sigma}}\right) X^{\lambda}  \; , \nonumber
	\end{eqnarray}
where to arrive at the final line we have used the first Bianchi identity to eliminate terms.}

{\ Regarding point three in our list - notice that \eqref{Eq:CommutatorTensor} and compatibility of $\nabla$ with $\tau_{\mu}$ gives \eqref{Eq:RiemannContracttau}. Meanwhile, when we replace $c_{\rho \sigma} = h_{\rho \sigma}$ in \eqref{Eq:CommutatorTensor}, and again employing compatibility of $\nabla$ with the structure invariants, we arrive at \eqref{Eq:RiemannContracth}. Similarly for the other identities.}
\end{proof}

\section{Horizontal and vertical spaces }\label{appendix:horizontalandvertical}
\renewcommand{\theequation}{B.\arabic{equation}}
\setcounter{equation}{0}
{\noindent We present here a brief discussion on the nature of horizontal and vertical spaces.  As a reminder, the cotangent bundle is the triple $(M,T^{*}M,\pi)$ with $\pi: T^{*}M \rightarrow M$ the projection map. This projection map has certain properties necessary for the triple to be a fibre bundle \cite{Krasnov_2020} . Importantly, any fibre bundle comes equipped with a preferred subset of vector fields - those that belong to the kernel of the projection map.}

\begin{definition}[Vertical vector field]
{\ Let $\mathrm{d} \pi: T_{(x,p)}(T^{*}M) \rightarrow T_{x}M$ be the differential of the projection map at a point $(x,p) \in T^{*}M$. A vector $Z \in T_{(x,p)}(T^{*}M)$ is vertical if it belongs to the kernel of the projection map, $Z \in \mathrm{ker} [\mathrm{d} \pi]$.}
\end{definition}

{\ In particular, in locally adapted coordinates the differential of the projection map takes a vector at $(x,p)$,
	\begin{eqnarray}
		\label{Eq:Zexpansion}
		Z = X^{\mu} \left. \partial_{\mu} \right|_{(x,p)} + Y_{\mu} \left. \frac{\partial}{\partial p_{\mu}} \right|_{(x,p)}
	\end{eqnarray}
and gives
	\begin{eqnarray}
		\mathrm{d} \pi(Z) = X^{\mu} \left. \partial_{\mu} \right|_{x} \; .  
	\end{eqnarray}
Consequently, a convenient basis for the vertical space is the set
	\begin{eqnarray}
		\frac{\partial}{\partial p_{\mu}} \; . 	
	\end{eqnarray}
There is one well-defined notion of a vertical space, however there are many notions of the complementary horizontal space which we now define:}

\begin{definition}[A horizontal subbundle a.k.a. an Ehresmann connection]\label{def:horizontal}
{\ Let $VM$ be the vertical subbundle of a manifold $M$ and $T(T^{*}M)$ the tangent space to the cotangent bundle. A horizontal sub-bundle is \underline{any smooth subbundle} of $T(T^{*}M)$ such that $T(T^{*}M) = HM \oplus VM$ where $\oplus$ is the direct sum.}
\end{definition}

{\ Importantly, the horizontal bundle is not unique, which is related to the arbitrariness of connections used to build covariant derivatives. Regardless, given a chosen horizontal subbundle, we can use it to define the lift of curves \cite{Krasnov_2020} from the base manifold $M$ onto the cotangent bundle. Let $\gamma(\lambda)$ be a curve in $M$ with $x=\gamma(0)$. Select a point $p \in \pi^{-1}(x)$ in the fibre over the point $x$. The horizontal lift of $\gamma(\lambda)$, denoted by $\tilde{\gamma}(\lambda)$, is the curve in $T^{*}M$ that passes through $p$ such that the tangent to this curve is in the chosen horizontal bundle and $\pi \circ \tilde{\gamma}(\lambda)=\gamma(\lambda)$.}

{\  Let us be more concrete and follow \cite{Acu_a_C_rdenas_2022}, taking a vector $Z$ in $T_{(x,p)}(T^{*}M)$. Let a curve $\tilde{\gamma}(\lambda)$ have the tangent vector $Z$ at $(x,p)$. This curve consists of the points $(x(\lambda),p(\lambda))$ where $x(\lambda) \in M$ and $p(\lambda) \in T_{x(\lambda)}^{*}M$ such that 
	\begin{eqnarray}
		(x^{\mu}(0),p_{\mu}(0))=(x^{\mu},p_{\mu}) \; , \qquad \left( \frac{d x^{\mu}}{d \lambda}(0),\frac{d p_{\mu}}{d \lambda}(0) \right) = \left( X^{\mu}, Y_{\mu} \right) \; . 
	\end{eqnarray}
Let our manifold $M$ be supplied with some connection $\nabla$. Then we can solve for the parallel transport of the covector $\xi$ along the curve $x^{\mu}(\lambda)$ according to 
	\begin{enumerate}
		\item the parallel transport equation
		\begin{eqnarray}
			\nabla_{\dot{x}(\lambda)} \xi(\lambda) = 0 
		\end{eqnarray}
		for all $\lambda \in [a,b]$,
		\item with the initial value of $\xi(a) = p(0)$.
	\end{enumerate}
Subsequently, the connection map at the point $(x,p)$ acting on our initial vector $Z$ is defined by
	\begin{eqnarray}
		K_{(x,p)}(Z) = \left. \frac{d}{d\lambda} \xi(\lambda) \right|_{\lambda = 0} \; . 	
	\end{eqnarray}
Using a locally adapted coordinate basis in which our initial tangent vector $Z$ takes the form of \eqref{Eq:Zexpansion} we can write the connection map as \cite{Acu_a_C_rdenas_2022}
	\begin{eqnarray}
			K_{(x,p)}(Z) 
		&=& \left[ \frac{d x^{\mu}}{d\lambda} \nabla_{\mu} \left( p_{\nu}(\lambda) \left. \mathrm{d}x^{\nu} \right|_{x(\lambda)} \right) \right]_{\lambda=0} \nonumber \\
		&=& \left[ \frac{d x^{\mu}}{d\lambda} \nabla_{\mu} p_{\nu}(\lambda) \left. \mathrm{d}x^{\nu} \right|_{x(\lambda)} 
		+  \frac{d x^{\mu}}{d\lambda} p_{\nu}(\lambda)  \nabla_{\mu} \left. \mathrm{d}x^{\nu} \right|_{x(\lambda)}   \right]_{\lambda=0} \nonumber \\
		&=& \left[ \frac{d}{d\lambda} p_{\nu}(\lambda)  \left. \mathrm{d}x^{\nu} \right|_{x(\lambda)} -  \frac{d x^{\mu}}{d\lambda} p_{\nu}(\lambda)  \Gamma_{\mu \rho}^{\nu} \left. \mathrm{d}x^{\rho} \right|_{x(\lambda)}  \right]_{\lambda=0} \nonumber \\
		&=& \left[ Y_{\nu}  -  X^{\mu} p_{\nu}  \Gamma_{\mu \rho}^{\nu} \right]  \left. \mathrm{d}x^{\nu} \right|_{x(0)} \; . 
	\end{eqnarray}
Therefore, in local adapted coordinates, the connection map takes the form
	\begin{eqnarray}
		K_{(x,p)}(Z) &=& \left( Y_{\alpha} - \Gamma_{\mu \alpha}^{\nu} p_{\nu} X^{\mu} \right) \left. \mathrm{d}x^{\alpha}	\right|_{(x,p)} \; . 
	\end{eqnarray}	
Just as in \cite{Acu_a_C_rdenas_2022} we make the following observations:
	\begin{enumerate}
		\item The connection map does not depend on the curve.
		\item The connection map is linear and maps the tangent space to the cotangent bundle $T_{(x,p)} T^{*}M$ to the cotangent space at $x$, $T^{*}_{x} M$.
		\item A vector $Z$ belongs to the kernel of the connection map if and only if
			\begin{eqnarray}
				Z &=& X^{\mu} \left[ \partial_{\mu} + \Gamma^{\nu}_{\mu \rho} p_{\nu} \frac{\partial}{\partial p_{\rho}} \right] \; . 
			\end{eqnarray}
		Consequently, a suitable basis for a \underline{horizontal} space defined through vectors that belong to the kernel of the connection map is the set
			\begin{eqnarray}
				\left. D_{\mu} \right|_{(x,p)} &=& \left. \partial_{\mu} + \Gamma^{\nu}_{\mu \rho} p_{\nu} \frac{\partial}{\partial p_{\rho}}  \right|_{(x,p)} \; . 
			\end{eqnarray}
		Importantly, this basis is dependent on the choice of covariant derivative (i.e. the connection). A different covariant derivative leads to a different horizontal basis. 
	\end{enumerate}
Finally then, we note that any vector $Z$ in the tangent space of the cotangent bundle can now be uniquely decomposed into
	\begin{eqnarray}
		Z = Z^{\mu} \left. D_{\mu} \right|_{(x,p)} + Y_{\mu} \left. \frac{\partial}{\partial p_{\mu}} \right|_{(x,p)}
	\end{eqnarray}
as discussed in \cite{Acu_a_C_rdenas_2022}.}

{\ This is the point where the Aristotelian case significantly diverges from the relativistic case. In the latter situation, one has the inverse metric which is an isomorphism between $T_{x}^{*}M$ and $T_{x}M$. No such isomorphism exists in the Aristotelian case and thus we cannot naively define the analogue of the Sasaki metric (the uplift of the spacetime metric to the phase space). Nevertheless, this decomposition into horizontal and vertical components has physical significance as the Hamiltonian vector fields of free Hamiltonians all belong to the horizontal subbundle.}

\section{The Hamiltonian of a free Lifshitz particle }\label{appendix:Lifshitz}
\renewcommand{\theequation}{C.\arabic{equation}}
\setcounter{equation}{0}
{\noindent A canonical example of a non-boost invariant Lagrangian is one with Lifshitz symmetry, which is a form of generalised scale invariance \cite{Brattan:2017yzx,Brattan:2018sgc,Brattan:2018cgk,Brattan:2018wxs}. In this section we shall slightly generalise some of the results in \cite{deBoer:2017ing} showing how they can be encompassed in the formalism we have developed in this paper.}

{\ Let us first demonstrate how a suitable Hamiltonian formulation can be derived from a reparameterisation and scale invariant action. In particular, let $x^{\mu}(\lambda)$  be the trajectory of a particle so that under a generic scaling transformation, the scalars of the free Hamiltonian \eqref{Eq:FreeHamiltonian} transform as
    \begin{eqnarray}
        \tau_{\mu} \dot{x}^{\mu} \rightarrow \Lambda^{-z} \tau_{\mu} \dot{x}^{\mu} \; , \qquad h_{\mu \nu} \dot{x}^{\mu} \dot{x}^{\nu} \rightarrow \Lambda^{-2} h_{\mu \nu} \dot{x}^{\mu} \dot{x}^{\nu}  \; .
    \end{eqnarray}
Under a reparameterisation of the trajectory $\lambda=\lambda(\lambda')$ these quantities transform as
    \begin{eqnarray}
    	\label{Eq:Reparaminvariance}
        \tau_{\mu} \dot{x}^{\mu} \rightarrow \frac{\partial \lambda'}{\partial \lambda} \tau_{\mu} \hat{x}^{\mu} \; , \qquad
        h_{\mu \nu} \dot{x}^{\mu} \dot{x}^{\nu} \rightarrow \left(\frac{\partial \lambda'}{\partial \lambda}\right)^2 h_{\mu \nu} \hat{x}^{\mu} \hat{x}^{\nu}  \; ,
    \end{eqnarray}
where $\hat{x}^{\mu} = \mathrm{d}x^{\mu}/ d\lambda'$. Any reparameterisation \eqref{Eq:Reparaminvariance} invariant action must necessarily take the form
    \begin{eqnarray}
      \mathscr{S} &=& \int d\lambda \; (\tau_{\mu} \dot{x}^{\mu}) \tilde{L}\left( \frac{h_{\mu \nu} \dot{x}^{\mu} \dot{x}^{\nu}}{(\tau_{\mu} \dot{x}^{\mu})^2} \right)
    \end{eqnarray}
for an arbitrary function $\tilde{L}$. However, under a scaling transformation the action transforms as
    \begin{eqnarray}
      \mathscr{S} &=& \int d\lambda \; (\tau_{\mu} \dot{x}^{\mu}) \Lambda^{-z} \tilde{L}\left( \Lambda^{2(z-1)} \frac{h_{\mu \nu} \dot{x}^{\mu} \dot{x}^{\nu}}{(\tau_{\mu} \dot{x}^{\mu})^2} \right) \; .
    \end{eqnarray}
Subsequently, we also need $\tilde{L}$ to be a homogeneous function of its argument for the action to be scale invariant. In particular
    \begin{eqnarray}
        \tilde{L}\left( \Lambda^{2(z-1)} \frac{h_{\mu \nu} \dot{x}^{\mu} \dot{x}^{\nu}}{(\tau_{\mu} \dot{x}^{\mu})^2} \right) &=& \Lambda^{2n(z-1)} \tilde{L} \left( \frac{h_{\mu \nu} \dot{x}^{\mu} \dot{x}^{\nu}}{(\tau_{\mu} \dot{x}^{\mu})^2} \right) \; .
    \end{eqnarray}
which fixes $n=\frac{z}{2(z-1)}$. For a given $z$, the only scale and reparameterisation invariant particle action on a $(d+1)$-dimensional Aristotelian manifold is then
    \begin{eqnarray}
    	\label{Eq:LifshitzAction}
        \mathscr{S} &=& \int d\lambda \; L_{z} \; , \qquad L_{z} = \frac{\left( h_{\mu \nu} \dot{x}^{\mu} \dot{x}^{\nu} \right)^{\frac{z}{2(z-1)}}}{\left(\tau_{\mu} \dot{x}^{\mu}\right)^{\frac{1}{z-1}}} \; .
    \end{eqnarray}
When $z=2$ we find
    \begin{eqnarray}
        \mathscr{S} &=& \int d\lambda \; \frac{h_{\mu \nu} \dot{x}^{\mu} \dot{x}^{\nu}}{\tau_{\mu} \dot{x}^{\mu}} \; ,
    \end{eqnarray}
as expected (see \cite{Matus:2024kyg} for an approach to such actions and their distinct structures).}

{\ To confirm that the Hamiltonian obtained from our Lifshitz Lagrangian \eqref{Eq:LifshitzAction} vanishes identically, we first identify the the generalised momenta:
    \begin{eqnarray}
    	 p_{\mu} \nu^{\mu} = \frac{L_{z}}{z-1} \frac{1}{\tau_{\sigma} \dot{x}^{\sigma}} \; , \qquad
       	p_{\mu} \tilde{h}^{\mu \nu} p_{\nu} =  \left(\frac{L_{z} z}{2(z-1)}\right)^2 \frac{1}{\dot{x}^{\mu} h_{\mu \nu} \dot{x}^{\nu}} \; .
    \end{eqnarray}
The right-hand side of these expressions implies the following scaling behaviour:
    \begin{eqnarray}
    	\label{Eq:ScalingMomenta}
        p_{\mu} \nu^{\mu} &\rightarrow& \Lambda^{z} p_{\mu} \nu^{\mu} \;, \qquad p_{\mu} \tilde{h}^{\mu \nu} \rightarrow \Lambda^{2} p_{\mu} \tilde{h}^{\mu \nu} p_{\nu} \; .
    \end{eqnarray}
It is a straightforward calculation to compute the canonical Hamiltonian using the Legendre transformation of the Lagrangian:
    \begin{eqnarray}
        H = p_{\mu} \dot{x}^{\mu} - L_{z} = \left[ \frac{z}{2(z-1)} L_{z} - \frac{1}{z-1} L_{z} \right] - L_{z} = 0 \; ,
    \end{eqnarray}
as expected.}

{\ To obtain a non-trivial Hamiltonian we must isolate the dynamical constraint obeyed by the particles. In analogy with relativistic particles we notice that
    \begin{eqnarray}
            \left( p_{\mu} \tilde{h}^{\mu \nu} p_{\nu} \right)^{\frac{z}{2(z-1)}}
        &=& \frac{1}{z-1} \left( \frac{z}{2} \right)^{\frac{z}{z-1}} \left( p_{\mu} \nu^{\mu} \right)^{\frac{1}{z-1}} \nonumber \\
            \Rightarrow
            \label{Eq:LifshitzConstraint}
            0
        &=& - p_{\mu} \nu^{\mu} + \frac{1}{\alpha} \left( p_{\mu} \tilde{h}^{\mu \nu} p_{\nu} \right)^{\frac{z}{2}} \; .
    \end{eqnarray}
This dynamical constraint between the momenta must be obeyed by every solution to the equations of motion. Consequently, a suitable Hamiltonian to describe the motion of Lifshitz particles looks like
    \begin{eqnarray}
    	\label{Eq:LifshitzConstraintH}
        H &=& \lambda \left( p_{\mu} \nu^{\mu} + \frac{1}{\alpha} \left( p_{\mu} \tilde{h}^{\mu \nu} p_{\nu} \right)^{\frac{z}{2}} \right) \; ,
    \end{eqnarray}
where $\lambda$ is an auxiliary field that enforces the constraint. The scaling of this auxiliary field can be chosen such that the full Hamiltonian scales as 
    \begin{eqnarray}
     H \rightarrow \Lambda^{z} H \; ,
    \end{eqnarray}
which is the scaling behaviour dictated by our Legendre transformation. As we have chosen  our Hamiltonian to be given by \eqref{Eq:LifshitzConstraintH}, which contains just a single power of the constraint \eqref{Eq:LifshitzConstraint}, we quickly determine that the auxiliary field $\lambda$ is just a scale invariant constant. Had we taken some power of the constraint as our Hamiltonian, the result would have been more complicated and the auxiliary field would pick up a scaling behaviour.}

\section{Integrating momentum over the solid angle }\label{appendix:integrals}
\renewcommand{\theequation}{D.\arabic{equation}}
\setcounter{equation}{0}
{\ A standard result from tensor calculus gives
\begin{subequations}
	\label{Eq:TensorIntegral}
\begin{eqnarray}
	&\;& \int_{S^{d-1}} \hat{p}_{i_1} \hat{p}_{i_2} \cdots \hat{p}_{i_{2n}} \, d\Omega_{d} \nonumber \\
	&\;& \qquad = \Omega_d \cdot \frac{\Gamma\left(n+\frac{1}{2}\right) \Gamma\left(\frac{d}{2}\right)}{\sqrt{\pi} \Gamma(n) \Gamma\left(n+\frac{d}{2}\right)} \sum_{\text{contractions}} \delta_{(i_{1} i_{2}} \delta_{i_{3} i_{4}} \ldots \delta_{i_{2n-1} i_{2n})} \; , \qquad \\
	&\;& \hat{p}_{i} = \frac{\vec{p}_{i}}{\sqrt{\vec{p}_{i} \cdot \vec{p}_{i}}} \; , \qquad \Omega_{d} = \frac{2 \pi^{\frac{d}{2}}}{\Gamma\left(\frac{d}{2}\right)} \; , 
\end{eqnarray}
\end{subequations}
where $\Omega_{d}$ is the d-dimensional solid angle of $S^{d-1} \subset \mathbbm{R}^{d}$, $\hat{p}_{i}$ is a unit vector in $d$-dimensions and we sum over all possible ways to partition $m$ indices into unordered pairs. When $m$ is odd there is an unpaired index and one should interpret the integral to be vanishing.}

{\ In the presence of a non-zero velocity, the $SO(d)$ symmetry of flat space is broken to $SO(d-1)$. Let the metric on the $d$-sphere be written in the following coordinate system
	\begin{eqnarray}
		ds^2 &=& d\theta^2 + \sin(\theta)^2 d\Omega_{d-1}^2 \; .
	\end{eqnarray}
The spatial momentum indices can be projected parallel or perpendicular to the velocity, so it is useful to work in a coordinate system that reflects this. In particular, let us write the Cartesian components of $\vec{p}$ in polar coordinates
    \begin{subequations}
    \label{Eq:MomentumDecomp}
    \begin{eqnarray}
    	\vec{p} &=& \| \vec{p} \| \left( \sin(\theta) \sum_{i=1}^{i=d} \vec{E}_{i} + \cos(\theta) \frac{\vec{v}}{v} \right) \; , \\
	\vec{E}_{2} &=& \cos(\theta_{1}) \vec{e}_{2} \; , \\
	\vec{E}_{i+1} &=& \sin(\theta_{1}) \ldots \sin(\theta_{i-1}) \cos(\theta_{i}) \vec{e}_{i+1} \; ,  \qquad i = 2 , \ldots, d-2\\
	\vec{E}_{d} &=&   \sin(\theta_{1}) \ldots \sin(\theta_{d-3}) \sin(\theta_{d-2}) \vec{e}_{d} \; . 
    \end{eqnarray}
    \end{subequations}
For a unit vector in particular we have
	\begin{eqnarray}
		\hat{p} &=& \sin(\theta) \hat{p}^{\perp} + \cos(\theta) \hat{v} \; . 
	\end{eqnarray}
Using this decomposition and \eqref{Eq:MomentumDecomp} we can determine the following integrals which we employ in section \ref{sec:Anharmonic}:
	\begin{subequations}
	\begin{eqnarray}
			\int d\Omega_{d} \; \cos^{m}(\theta)
		&=& \frac{\Gamma\left(\frac{d-1}{2}\right) \Gamma\left(\frac{m+1}{2}\right)}{\Gamma\left(\frac{d+m}{2}\right)} \Omega_{d-1}  \delta_{m \in 2 \mathbbm{N}} \; , \\
			\int d\Omega_{d} \; \hat{p}_{i}^{\perp} \cos^{m}(\theta) \sin^{n}(\theta)
		&=& 0 \; , \\
			\int d\Omega_{d} \; \hat{p}_{i}^{\perp} \hat{p}_{j}^{\perp} \cos^{m}(\theta) \sin^{2}(\theta)
		&=&  \frac{\Gamma\left(\frac{d-1}{2}\right) \Gamma\left(\frac{m+1}{2}\right)}{\Gamma\left(\frac{d+m}{2}\right)} \frac{\Omega_{d-1}}{d+m}  \delta_{m \in 2 \mathbbm{N}} \Pi_{ij} \; , \qquad \qquad
	\end{eqnarray}
where
	\begin{eqnarray}
		\Pi_{ij} &=& \delta_{ij} - \hat{v}_{i} \hat{v}_{j} \; . 
	\end{eqnarray}
		\end{subequations}
}
 
\end{appendices}


\bibliography{sn-bibliography}

\end{document}